\documentclass{article}
\usepackage{authblk}

\usepackage{cuted}

\usepackage{graphicx} 
\usepackage{amsmath}
\usepackage{algorithm}
\usepackage{algpseudocode}
\usepackage{caption}
\usepackage{subcaption}
\usepackage{comment}
\usepackage{mathtools}
\usepackage[dvipsnames]{xcolor}
\usepackage{thmtools}
\usepackage{thm-restate}
\usepackage[utf8]{inputenc}
\usepackage{amsmath}
\usepackage{amssymb}
\usepackage{ulem}
\usepackage{bm}      
\usepackage[colorlinks,linkcolor=blue,citecolor=red]{hyperref}

\usepackage{array}
\usepackage{multirow}
\usepackage{hhline}

\usepackage{color}
\usepackage{graphicx}
\usepackage{amsfonts}
\usepackage{soul}

\usepackage[paperwidth=199.8mm,paperheight=287mm,centering,hmargin=25mm,vmargin=2cm]{geometry}
\usepackage{theorem}
\usepackage{minitoc}
\usepackage[capitalize]{cleveref}

\usepackage{appendix}

\AtBeginEnvironment{appendices}{\crefalias{section}{appendix}}

\makeatletter
\def\algbackskip{\hskip-\ALG@thistlm}
\makeatother

\newcommand{\bp}[0]{\theta}

\newcommand{\SC}[0]{\mathbb{W}_{\rho_{\beta}}\left(\cE\right)}

\newcommand{\jiaqing}[1]{{\color{orange}{[JJQ: #1]}}}

\newcommand{\cB}[0]{\mathcal{B}}

\newcommand{\ba}[0]{\mathbf{a}}
\newcommand{\bb}[0]{\mathbf{b}}
\newcommand{\bc}[0]{\mathbf{c}}
\newcommand{\bd}[0]{\mathbf{d}}

\newcommand{\bbS}[0]{\mathbf{S}}

\newcommand{\bB}[0]{\mathrm{Bohr}(H)}
\newcommand{\cD}[0]{\mathcal{D}}

\newcommand{\cP}[0]{\mathcal{P}}
\newcommand{\cF}[0]{\mathcal{F}}

\newcommand{\bV}[0]{\mathbb{V}\mathrm{ar}}
\newcommand{\bCov}[0]{\mathbb{C}\mathrm{ov}}
\newcommand{\bCor}[0]{\mathbb{C}\mathrm{orr}}
\newcommand{\bPr}[0]{\mathbb{P}\mathrm{r}}

\newcommand{\ket}[1]{\left\vert #1 \right\rangle}

\newcommand{\bra}[1]{\left\langle #1 \right\vert}
\newcommand{\ketbra}[2]{\vert#1\rangle\!\langle#2\vert}

\newcommand{\bM}[0]{\mathbb{M}}

\newcommand{\bC}[0]{\mathbb{C}}
\newcommand{\bR}[0]{\mathbb{R}}
\newcommand{\bE}[0]{\mathbb{E}}
\newcommand{\cT}[0]{\mathcal{T}}
\newcommand{\cO}[0]{\mathcal{O}}
\newcommand{\cN}[0]{\mathcal{N}}
\newcommand{\cM}[0]{\mathcal{M}}
\newcommand{\cE}[0]{\mathcal{E}}

\newcommand{\cL}[0]{\mathcal{L}}
\newcommand{\cH}[0]{\mathcal{H}}
\newcommand{\LL}[0]{\mbox{\bf L}}

\newtheorem{definition}{Definition}
\newtheorem{assumption}{Assumption}
\newtheorem{observation}{Observation}

\newtheorem{remark}{Remark}
\newtheorem{proposition}[definition]{Proposition}
\newtheorem{lemma}[definition]{Lemma}
\newtheorem{fact}{Fact}
\newtheorem{theorem}{Theorem}
\newtheorem{corollary}[definition]{Corollary}

\newenvironment{proof}{{\bf Proof:}}{\hfill\rule{2mm}{2mm}}

\renewcommand{\d}{\mathrm{d}}
\newcommand{\Tr}{\operatorname{tr}}
\newcommand{\Var}{\operatorname{var}}
\newcommand{\polylog}{\operatorname{poly}\log}
\newcommand{\poly}{\operatorname{poly}}

\newcommand{\bO}[0]{O}

\begin{document}

\title{Predicting properties of quantum thermal states from a single trajectory}

\author[1,\dag] {Jiaqing Jiang}
\author[1,2,\dag] {Jiaqi Leng}
\author[2,3,*] {Lin Lin}
\affil[1]{Simons Institute for the Theory of Computing, University of California, Berkeley, USA}
\affil[2]{Department of Mathematics, University of California, Berkeley, USA}
\affil[3]{Applied Mathematics and Computational Research Division, Lawrence Berkeley National Laboratory, Berkeley, USA}

\renewcommand{\LL}[1]{\textcolor{blue}{[LL:#1 ]}}
\newcommand{\JL}[1]{\textcolor{orange}{[JL:~#1]}}

\maketitle

\renewcommand{\thefootnote}{\dag}
\footnotetext{These authors contributed equally to this work.}
\renewcommand{\thefootnote}{*}
\footnotetext{\href{mailto:linlin@math.berkeley.edu}{linlin@math.berkeley.edu}}
\renewcommand{\thefootnote}{\arabic{footnote}} 

\begin{abstract}
Estimating thermal expectation values of observables is a fundamental task in quantum physics, quantum chemistry, and materials science. While recent quantum algorithms have enabled efficient quantum preparation of thermal states, observable estimation via sampling remains costly: a straightforward implementation separates successive measurements by a full mixing time in order to ensure samples are approximately independent. In this work, we show that the sampling cost can be substantially reduced by using a single Gibbs-sampling trajectory. After a single burn-in period, we interleave coherent measurements that satisfy detailed balance with respect to the target Gibbs state. The efficiency of this approach rests on the fact that, in many settings, the autocorrelation time can be significantly shorter than the mixing time. For energy estimation (and more generally for observables commuting with the Hamiltonian), we implement the required measurements using Gaussian-filtered quantum phase estimation with only logarithmic overhead. We also introduce a weighted operator Fourier transform technique to mitigate measurement-induced disturbance for general observables.
\end{abstract}

\tableofcontents

\section{Introduction}

Predicting the properties of quantum thermal states is a fundamental task in quantum many-body physics and is central to understanding the thermodynamic behavior of molecules and materials, including phase diagrams, binding energies, and response functions. 
Beyond physical simulations, thermal expectation estimation is also closely linked to quantum machine learning, for instance in training quantum Boltzmann machines that require efficient evaluation of gradients and expectation values~\cite{amin2018quantum,wiebe2019generative}.

Recently, substantial progress has been made in dissipative quantum algorithms for preparing thermal states, including Lindblad dynamics~\cite{chen2025efficient,ding2025efficient,rall2023thermal}, quantum Metropolis algorithm~\cite{jiang2024quantum,temme2011quantum,gilyen2024quantum}, and system--bath interaction methods~\cite{ding2025end,hagan2025thermodynamic,Hahn2025Efficient}. 
These quantum algorithms can be viewed as quantum generalizations of Markov chain Monte Carlo (MCMC) methods, and offer a promising path toward predicting thermal properties of quantum systems, i.e., estimating expectation values of observables with respect to the Gibbs state. A straightforward sampling-based approach is to run multiple independent trajectories of a Gibbs-sampling algorithm and measure at the end of each trajectory; because measurements disturb the system, two consecutive measurements along a single trajectory may need to be separated by a time comparable to the mixing time (the worst-case time required to drive an arbitrary initial state close to stationarity) to ensure that the resulting samples are approximately independent. In this work, we show that this Gibbs-sampling cost can be significantly reduced: using a single trajectory, one can obtain effectively independent samples on a timescale that can be much shorter than the mixing time (for the observable of interest), enabling a more efficient method for observable estimation.

At the heart of this reduction is the notion of \textit{autocorrelation time}~\cite{sokal1997monte}, which is a standard concept in classical MCMC. In classical simulation, a common practice~\cite{gilks1995markov} is to run a single chain, first performing a warm-up (or \textit{burn-in}) until the chain is close to stationarity, and then to collect samples spaced so that successive samples are effectively independent for the observable of interest. The autocorrelation time quantifies this spacing, and can be much smaller than the mixing time.

In the following, we focus first on estimating the average energy of quantum thermal states to highlight the challenges, intuitions, and strategies for using a single trajectory to estimate observable values in the quantum setting.  We then discuss how the approach extends to observables whose estimation can be implemented in a thermodynamically reversible manner (i.e., via an operation that satisfies detailed balance). We also investigate a weighted operator Fourier transform technique to mitigate disturbances when estimating more general observables.

\subsection{Energy estimation from a single trajectory}

 To begin with, we introduce some notation. For a local Hamiltonian $H$ and inverse temperature $\beta$, the thermal equilibrium state is defined as the Gibbs state 
$\rho_{\beta} := \exp(-\beta H)/Z_{\beta},$
where $Z_{\beta}$ is the normalization factor known as the partition function. For concreteness, in this section, we focus on the task of estimating the average energy of the Gibbs state, i.e., $\Tr(H \rho_{\beta})$. We denote the variance of $H$ as 
$\Var_H := \Tr(H^2 \rho_{\beta}) - \Tr(H \rho_{\beta})^2.$

\begin{figure}[h]
    \centering\includegraphics[width=0.95\textwidth,trim={0 5cm 0 5cm},clip]{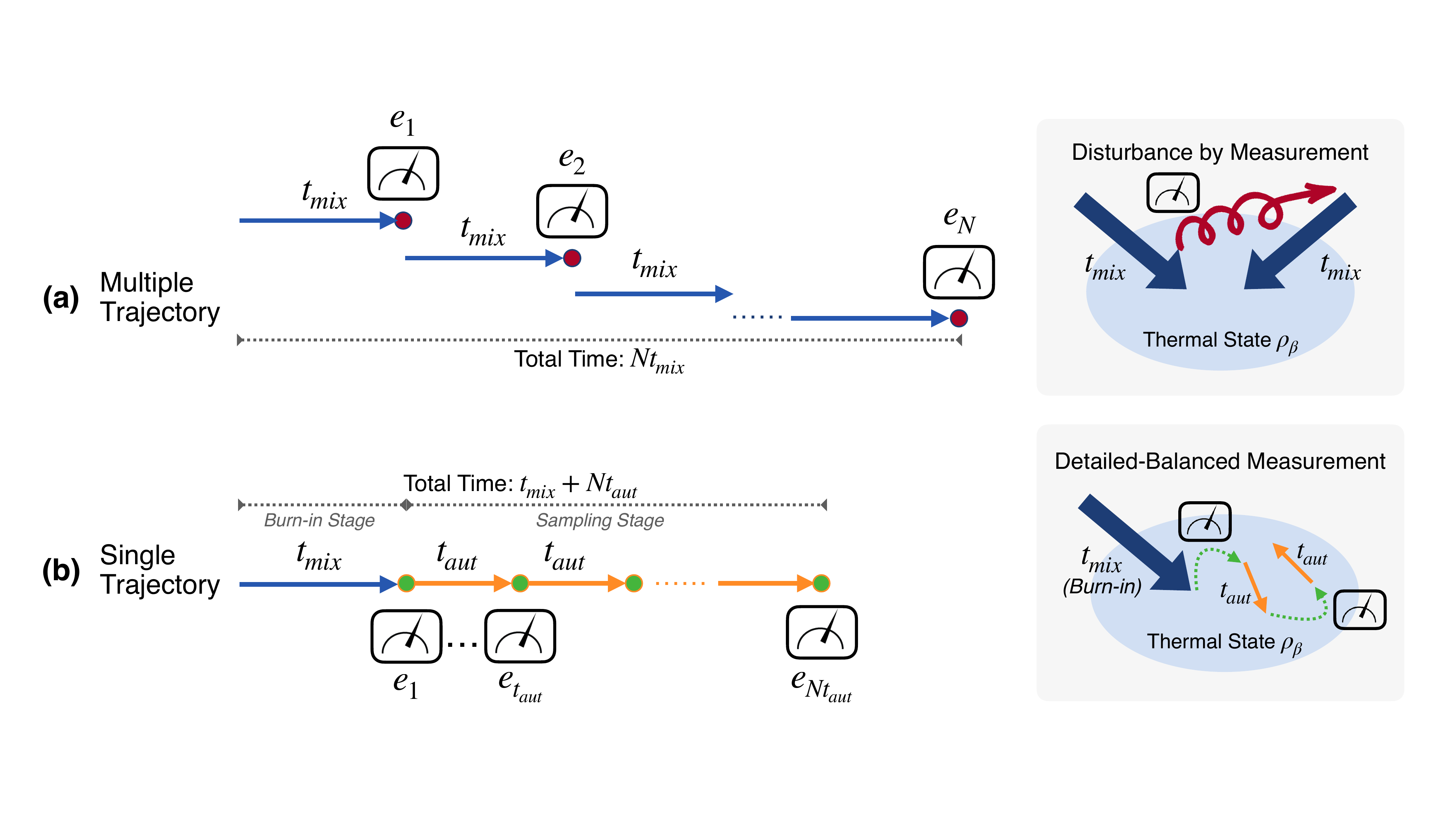}
\captionsetup{width=0.9\textwidth, font=small}
    \caption{Two approaches for estimating $\Tr(H \rho_{\beta})$ using a quantum Gibbs sampling algorithm (Gibbs sampling channel).
    (a) Multiple-trajectory approach: Run $N$ independent trajectories, each applying the Gibbs sampling channel for the mixing time $t_{mix}$ followed by a measurement of $H$. Conventionally, $t_{mix}$ is used since quantum measurement collapses the state and may effectively reinitialize the Gibbs sampling dynamics. Estimating $\Tr(H\rho_\beta)$ to precision $\epsilon$ requires $N = \mathcal{O}(\Var_H / \epsilon^2)$.  
  (b) Single-trajectory approach (this work):  First, run a burn-in stage by applying the Gibbs sampling channel for the mixing time $t_{mix}$. Then, enter the sampling stage, where measurements are performed after each fixed evolution time 
  $\Delta t$ (a tunable parameter, set to $\Delta t = 1$ for simplicity).  
    Measurements are chosen to satisfy detailed balance, implemented via Gaussian-filtered QPE, thus preserving the Gibbs state ensemble and avoiding extra burn-in periods. 
    In this approach, effectively independent samples are obtained approximately every autocorrelation time $t_{aut}$. The autocorrelation time leverages a warm start from the previous measurement and depends on the observable of interest, and can therefore be much shorter than the mixing time.
  As a result, 
  the total Gibbs sampling evolution time needed is  
  $t_{mix} + N t_{aut}$, which is significantly smaller than the 
$N t_{mix}$ required in the multiple-trajectory approach. 
}
\label{fig:single_trajectory}
\end{figure}

\paragraph{Multiple-trajectory approach} 
As illustrated in \cref{fig:single_trajectory}(a), a straightforward way to estimate 
$\Tr(H\rho_{\beta})$ 
using a Gibbs sampling algorithm is to run multiple 
independent trajectories. Each trajectory evolves for the mixing time $t_{mix}$.
These trajectories then yield independent samples from which the average energy can be estimated. To achieve $\epsilon$-precise estimation of $\Tr(H\rho_{\beta})$, it suffices to run 
$N=\mathcal{O}(\Var_H/\epsilon^2)$ independent trajectories. 
Equivalently, this 
multi-trajectory procedure can be viewed as running a single long trajectory in which consecutive samples are  separated by  $t_{mix}$, where $t_{mix}$ is conventionally used  since a quantum measurement  collapses the state and may effectively reinitialize the Gibbs sampling dynamics.

\paragraph{Single-trajectory approach}

Instead of separating samples by the full mixing time $t_{mix}$, as in 
\cref{fig:single_trajectory}(b), we use a two-stage algorithm to estimate 
$\Tr(H\rho_{\beta})$. Starting from an arbitrary initial state, 
our algorithm proceeds as follows:
\begin{itemize}
    \item[(1)]  We first perform a \textit{burn-in stage}, during which the Gibbs sampling channel is applied for one mixing time, allowing the chain to become approximately stationary, i.e., making the current state close to the Gibbs state.
    \item[(2)] 
   We then proceed to the \textit{sampling stage}, 
  where measurements are performed after each fixed evolution time $\Delta t$ along the trajectory. Note that the measurements are chosen to satisfy detailed balance, implemented via Gaussian-filtered quantum phase estimation (GQPE, explained later), \textit{ensuring that during the sampling stage the chain remains in the Gibbs state ensemble} so that no additional burn-in is required. The evolution time $\Delta t$ between successive measurements is treated as a tunable parameter. For simplicity we set $\Delta t=1$.
\end{itemize}

Denoting the outcome of the $j$-th measurement as $e_j$, we estimate $\Tr(H \rho_{\beta})$ as the empirical average of these outcomes along the trajectory, that is, 
\begin{equation}
    X_K:=\frac{1}{K}(e_1 + \cdots + e_K).
\end{equation}
To build intuition, note that when $H$ is a classical Hamiltonian like the Ising model, the Gibbs-sampling algorithm reduces to a classical Markov chain over the computational basis, and $e_j$ is simply the 
 energy of the current configuration.

This two-stage approach is widely used in classical MCMC and can substantially reduce the Gibbs sampling costs in practice. By the ergodic theorem for classical Markov chains, the empirical average $X_K$ converges to $\Tr(H \rho_\beta)$, with the convergence rate determined by how quickly the samples $\{e_j\}_j$ become effectively independent --- a timescale known as the \textit{autocorrelation time}.
The key observation that leads to the Gibbs sampling reduction is that the autocorrelation time can be much smaller than the mixing time, for two intuitive reasons. First, the autocorrelation time is an equilibrium notion, which quantifies decorrelation once the chain is already in stationarity. In contrast, the mixing time measures the \textit{worst-case time} required to drive an arbitrary initial state to the Gibbs distribution. Second, the autocorrelation time is \textit{observable-dependent}, and can therefore be shorter than the mixing time for observables that are local or exhibit certain symmetries. To build intuition, in \cref{sec:example} we provide examples where the autocorrelation time can be substantially smaller than the mixing time.
More generally, for reversible classical Markov chains, the autocorrelation time can be bounded by the inverse spectral gap, while mixing-time bounds typically include an additional factor of order $\log(\pi_{\min}^{-1})$~\cite{levin2017markov}. In many models this factor scales linearly in the system size.
In specific systems, the separation can be even more dramatic: for example, for low-temperature Glauber dynamics of the 2D Ising model, the mixing time can be sub-exponential in the system size~\cite{cesi1996two}, with scaling of the form $\exp(\Theta(\sqrt{n}))$ for $n$ spins~\cite{cesi1996two}; while existing literature~\cite{sokal1997monte,martinelli2004lectures,gheissari2022low} suggests that the autocorrelation time for the energy, which is symmetric under global spin flips, might grow at most polynomially with system size.

Despite the  
success in classical MCMC, however, it is not obvious that the same two-stage strategy  also works well when applied to quantum 
Gibbs sampling without incurring significant overhead. In the following section, we outline the 
two main challenges and summarize our observations and strategies for addressing them.
In particular, we formalize the notion of autocorrelation time for quantum Gibbs sampling and establish an upper bound in terms of the spectral gap of the underlying quantum Gibbs sampler, as summarized in the following:

\begin{theorem}[Informal]\label{thm:intro_main} 
 The autocorrelation time $t_{aut}$ in the single-trajectory algorithm in \cref{fig:single_trajectory} is upper bounded by the reciprocal of the spectral gap of the Gibbs sampler.
 To estimate $\Tr(H \rho_\beta)$ to precision $\epsilon$ with high probability, it suffices to take 
 \begin{equation}
K = N  t_{aut},\,\, \text{ for } N=\mathcal{O}(\Var_H / \epsilon^2).
 \end{equation}
 Consequently, the total Gibbs sampling evolution time, $t_{mix} + N t_{aut}$, can be much smaller than $N t_{mix}$, substantially reducing the Gibbs sampling cost. 
Moreover, each measurement in our single-trajectory algorithm requires only a logarithmic number of ancilla qubits and logarithmic-time controlled Hamiltonian evolution of $H$, where the logarithmic scaling is with respect to both the precision $\epsilon$ and the spectral norm $\|H\|$.  
\end{theorem}

\cref{thm:intro_main} can be generalized to any observable that commutes with $H$, such as $H^k$. Formal statements are provided in   \cref{lem:aut_gap}, \cref{thm:time_average} and \cref{thm:M_com}. 
The total cost of implementing measurements can be reduced by skipping unnecessary measurements, as discussed in \cref{sec:time_average} and \cref{remark:skip_measure}.

\vspace{0.3em}

\underline{\textit{Practical implementation}}\,\,
We remark that our algorithm can also be made \textit{agnostic}, i.e., without prior knowledge of the mixing time or autocorrelation time. In applications where reliable estimates of these times are unavailable, one can run the burn-in period for an arbitrary fixed time $T_{\rm burn}$, then enter the sampling period and take measurements after every fixed evolution time $\Delta t$. 
The empirical average along the trajectory, $(e_1 + \dots + e_K)/K$, still converges to the observable value $\Tr(H \rho_\beta)$, albeit at a slower rate due to the biased samples collected before the chain reaches stationarity. 

In practice, one can monitor the empirical average along the Gibbs sampling trajectory and stop once it converges. In this sense, the agnostic version of our algorithm can be used as an empirical method for certifying the convergence of quantum Gibbs sampling, analogous to the Gelman--Rubin diagnostic used to assess convergence in classical MCMC~\cite{gelman1992inference,vats2021revisiting}.
To our knowledge, no similar diagnostic protocol exists for certifying the convergence of quantum Gibbs sampling in the current literature; the only related work we are aware of certifies convergence for certain quantum Markov chains using a coupling-from-the-past technique~\cite{francca2017perfect}.

\subsection{Two challenges and our strategy}\label{sec:2challenge}

To help better understand \cref{thm:intro_main}, in this section we discuss two main challenges encountered when estimating $\Tr(H \rho_{\beta})$ from a single trajectory, and explain how we address them.

\vspace{0.3em}

\underline{\textit{First challenge: Disturbance from measurement}}\,\,
The first challenge is that, unlike in classical Gibbs sampling, where evaluating $e_j= \langle x_j | H | x_j \rangle$ simply queries the current configuration $\ket{x_j}$ and leaves the Markov chain untouched, any intermediate measurement performed during quantum Gibbs sampling inevitably interferes with the sampling process. This interference occurs even if we assume that $H$ can be measured perfectly in its eigenbasis.
The reason is that in many quantum Gibbs sampling algorithms~\cite{chen2025efficient,ding2025efficient}, the transitions (or ``jumps'') between intermediate states are highly entangling. Even if the system begins in an eigenstate of $H$, the application of a single update step typically produces a state that is a superposition of many energy eigenstates, often spanning a non-negligible energy range.
Measuring the system in the energy eigenbasis at such an intermediate step collapses the superposition and therefore disrupts the intended evolution of the quantum Gibbs sampler. Consequently, inserting measurements can potentially increase the autocorrelation time, introducing an undesirable overhead.
It is not even immediately clear whether inserting measurements necessarily worsens the mixing time; that is, the mixing time of the Gibbs sampler composed with a measurement channel could in principle be larger than that of the original Gibbs sampler.

We observe that although intermediate measurements inevitably disturb the Gibbs-sampling dynamics, their effect on the mixing time can be controlled when the measurement is compatible with the Gibbs state. In this case, one can show that the spectral gap of the Gibbs sampler composed with the measurement is never smaller than that of the original Gibbs sampler. Intuitively, the required compatibility condition is that the measurement leaves the Gibbs state unchanged, so no bias is introduced and no additional burn-in period is needed. Technically, the analysis requires a slightly stronger condition, namely, \textit{detailed balance}, meaning that the measurement channel is thermodynamically reversible. In the specific task of estimating the energy $\Tr(H \rho_{\beta})$, this compatibility condition is naturally satisfied when performing a coherent measurement of $H$, for instance through quantum phase estimation. To avoid confusion, we emphasize that even if the measurement fixes the Gibbs state, one must still evolve under the Gibbs-sampling channel for one autocorrelation time to obtain effectively \textit{independent} samples. In summary, we prove that

\begin{lemma}[Informal]
    Let $\mathcal{N}$ denote a quantum Gibbs-sampling channel that satisfies detailed balance, such as those constructed in~\cite{chen2025efficient,ding2025efficient}. Let $\mathcal{M}$ be a measurement channel that also satisfies detailed balance with respect to $\rho_{\beta}$. Then the spectral gap of the composed channel $\mathcal{M}\mathcal{N}\mathcal{M}$ is at least as large as the spectral gap of $\mathcal{N}$. Moreover, the autocorrelation time is  upper bounded by  (up to a small factor)  the relaxation time of the Gibbs sampling channel, i.e., the reciprocal of the spectral gap of $\mathcal{N}$.
\end{lemma}

A fully formal version of the above lemma appears in \cref{lem:MNM_N} and \cref{lem:aut_gap}.

\vspace{0.3em}

\vspace{0.3em}

\underline{\textit{Second challenge: High cost of QPE}}\,\, The second challenge is that high-precision coherent quantum measurements can be costly, potentially erasing any advantage gained by replacing the mixing time with the autocorrelation time. In particular, implementing an $\epsilon$-precision quantum phase estimation (QPE) algorithm for $H$ requires a cost scaling of at least $\|H\|/\epsilon$~\cite{dalzell2310quantum,lin2022heisenberg,nielsen2010quantum}. For local Hamiltonians written as a sum of $\kappa$ bounded-norm terms, one typically has $\|H\|=\mathcal{O}(\kappa)$. In contrast, in classical Gibbs sampling, evaluating $\langle x_j|H|x_j\rangle$ exactly incurs a cost that is only linear in $\kappa$ without the additional $1/\epsilon$ overhead.

To address this issue, we observe that high-precision coherent measurements are unnecessary: 
it suffices to use a coherent measurement that is  \textit{unbiased}. 
The key insight is that an unbiased measurement is sufficient because many samples are collected along the trajectory anyway. For measuring the energy $\Tr(H \rho_{\beta})$ (or more generally any observable that commutes with $H$), we show that such a measurement can be implemented with only $\log(1/\epsilon)$ overhead relative to classical Gibbs sampling, in contrast to the $1/\epsilon$ overhead required for high-precision QPE. 

Specifically, we construct the measurement from Gaussian-filtered QPE with adjustable variance (GQPE), as introduced in~\cite{moussa2019low}. In particular, we set the variance of this GQPE to be of order one, which only slightly increases the variance of measuring $\rho_{\beta}$, i.e. $\Var_H$, by a constant, while significantly reducing the measurement overhead from $1/\epsilon$ to $\log(1/\epsilon)$.  
For concreteness, \cref{tab:intro_compare} provides a detailed comparison of the costs incurred by different methods of measuring $H$. Further details on this comparison can be found in \cref{sec:energy}. We present a rigorous resource analysis of GQPE in \cref{sec:GQPE}, which, to our knowledge, is not available in prior literature and may be of independent interest.

\begin{center}
\begin{table}[H]
    \centering
    \begin{tabular}{|c|c|c|c|c|c|c|c|}
    \hline
   Hamiltonian  & Measurement    & Gate count        & Depth      & Variance  & Total gate  & Total depth \\
   $H$  &  $\cM$         &     $c_{\cM}$    &    $d_{\cM}$ &   $\Var_{\cM}$        &        $c_{\cM}\times \Var_{\cM}$                  &    $d_{\cM}\times \Var_{\cM}$    \\
    \hline
  Classical &  HP  & $n$      & $1$  & $n$        & $n^2$ & $n$\\
     \hhline{--------}
    \multirow{5}{*}{Quantum} & HP QPE & $n\epsilon^{-1}$ & $\epsilon^{-1}$ & $n$ & $n^2 \epsilon^{-1}$ & $n\epsilon^{-1}$\\
    \hhline{~-------}
     \multirow{5}{*}{} &  *local & $1$ &  $1$ &  $n^2$ &  $n^2$ &  $n^2$  \\
     \hhline{~-------}
     \multirow{5}{*}{} & \textbf{GQPE} & $n$ & $1$ & $\mathbf{n}$ & $\mathbf{n^2}$  & $\mathbf{n}$\\
     \hline 
    \end{tabular}
    \captionsetup{width=0.9\textwidth, font=small}
    \caption{Comparison of costs for different methods of measuring $H$. To make meaningful comparison of circuit depth here we assume $H$ is an $n$-qubit Hamiltonian on a 2D lattice. Gate and depth costs of a single measurement are denoted by $c_{\mathcal{M}}$ and $d_{\mathcal{M}}$, and $\Var_{\mathcal{M}}$ is the measurement variance w.r.t. $\rho_{\beta}$. The variance scalings shown in the table are schematic and model-dependent (they can change with temperature and correlations). Entries suppress a $\polylog(n\epsilon^{-1})$ factor, e.g., $n$ represents $\mathcal{O}(n\,\polylog(n\epsilon^{-1}))$. Measurement channels include: (HP) exact classical evaluation of $\langle x_j | H | x_j \rangle$; (HP QPE) high-precision QPE; (*local) measuring a randomly selected local term, here * represents that the measurement is non-coherent and significantly disturbs the Gibbs state;
     (GQPE) Gaussian-filtered QPE adapted in this work. The key quantities of interest are (1) the variance $\Var_{\mathcal{M}}$, which determines the number of samples need to be taken thus determine the total evolution time of the Gibbs-sampling channel; (2) the total gate and depth costs, $c_{\mathcal{M}}\times \Var_{\mathcal{M}}$ and $d_{\mathcal{M}}\times \Var_{\mathcal{M}}$, which determines the cost of implementing the measurement.
GQPE achieves the lowest costs up to a polylog factor for all the three quantities of interest. 
}
    \label{tab:intro_compare}
\end{table}
\end{center}  

\vspace{0.3em}

\underline{\textit{Mitigating measurement-induced disturbance for general observables}}\,\,
From the previous sections, the \textit{intuition} behind the two-stage algorithm in \cref{fig:single_trajectory} is as follows: if a measurement leaves the Gibbs state ensemble unchanged, no additional burn-in is required during the sampling stage, potentially reducing the overall Gibbs sampling cost. This property generally does not hold for observables that do not commute with the Hamiltonian $H$: estimating $\Tr(O \rho_\beta)$ via a projective measurement of $O$ can significantly disturb the state away from $\rho_\beta$. 

In \cref{sec:WOFT_noncom}, we take a first step toward mitigating this disturbance by exploring a weighted operator Fourier transform (WOFT) technique~\cite{chen2023efficient}. Roughly speaking, for a small parameter $\tau>0$, WOFT constructs an observable $\widehat{O}(\tau)$ that nearly commutes with $\rho_{\beta}$ (equivalently, with $H$); that is, the commutator
$\left[\widehat{O}(\tau),\rho_{\beta}\right]$ approaches zero as $\tau\rightarrow 0$.  
 Moreover, it is guaranteed that $\Tr(\rho_{\beta} \widehat{O}(\tau))= \Tr(\rho_{\beta} O)$. Therefore, to estimate $\Tr(\rho_{\beta} O)$, it suffices to measure the observable $\widehat{O}(\tau)$, which nearly commutes with $\rho_{\beta}$. As a result of this near-commutation property, 
we show that measuring  $\widehat{O}(\tau)$
  via GQPE perturbs $\rho_{\beta}$ only slightly.

We note that although GQPE measurement for $\widehat{O}(\tau)$ perturbs the Gibbs state only slightly when $\tau$ is sufficiently small,  implementing this measurement could incur a higher cost  than the logarithmic overhead achieved previously for observables that commute with $H$, unless additional assumptions are introduced. Besides, our analysis for bounding the autocorrelation time does not directly apply since this measurement does not satisfy \textit{exact} detailed balance, though approximately fixes the Gibbs state.

\subsection{Examples of mixing time and autocorrelation time comparison}\label{sec:example}

We give examples in which the autocorrelation time of a fixed observable is much smaller than the mixing time. The reason is that the mixing time is a worst-case convergence time to the Gibbs state, whereas the autocorrelation time characterizes the decay of temporal correlations of an observable when the chain is started in stationarity (i.e., a warm start). In addition, the autocorrelation time is observable-dependent, and can be much smaller for observables that are insensitive to certain slow modes of the dynamics. 

We first illustrate this phenomenon with classical examples, and then present a family of quantum Hamiltonians for which we rigorously establish an exponential separation between the mixing time and the energy autocorrelation time.

\begin{figure}[ht]
	\centering
\includegraphics[width=0.8\textwidth,trim={2cm 5cm 2cm 5cm},clip]{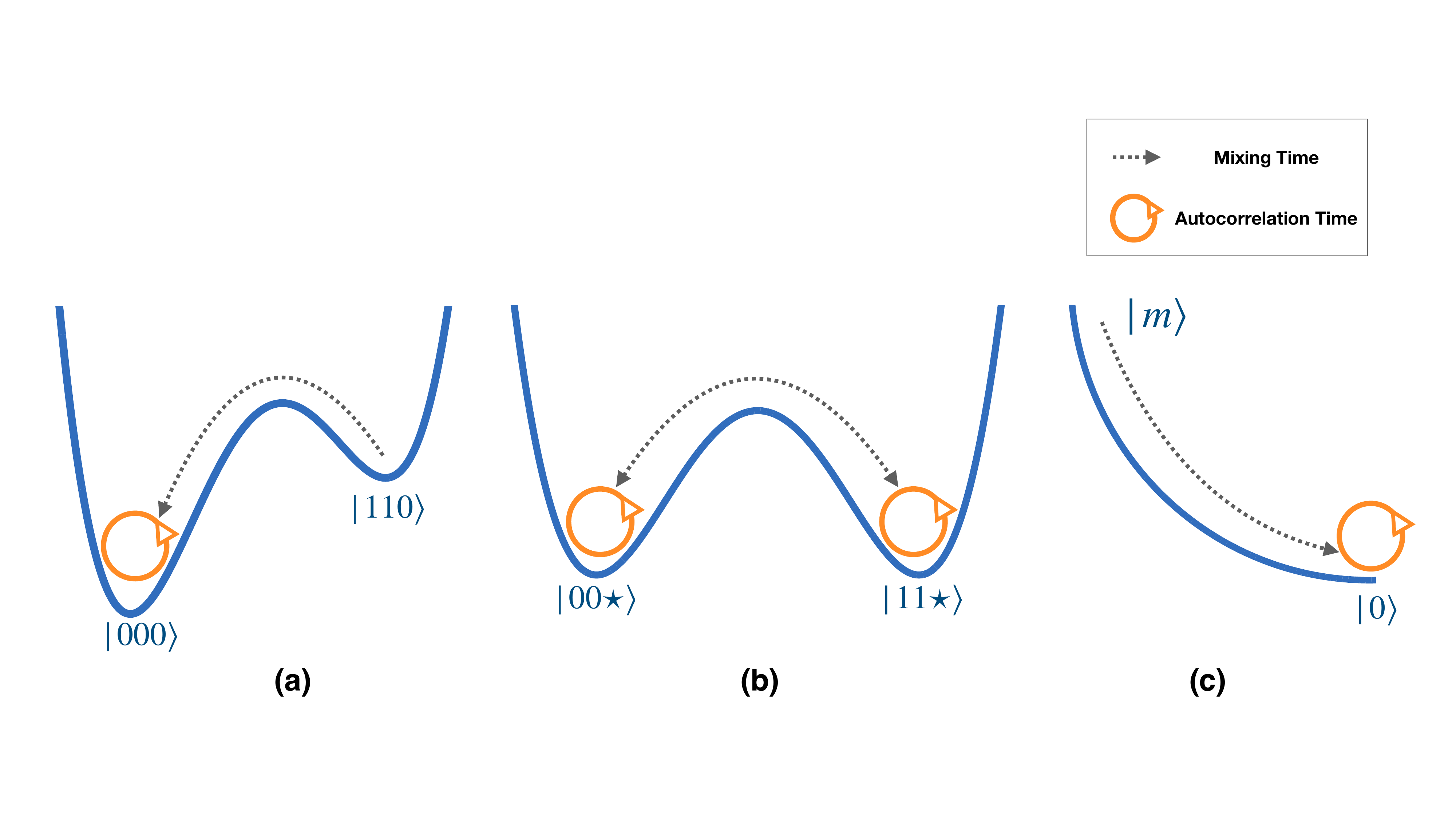}\captionsetup{width=0.9\textwidth, font=small}
    \caption{
    Three examples illustrating that the autocorrelation time of the energy observable can be much smaller than the mixing time. Let $\beta$ be a constant inverse temperature. (a)(b) correspond to the three-qubit Ising model $H = - \alpha Z_1 Z_2 - h (Z_1 + Z_2) - \gamma Z_3$ under Glauber dynamics: (a) Asymmetric double-well ($\alpha \gg h \gg \gamma > 0$) with negligible Gibbs weight in the smaller well; the mixing time is $\Omega(e^{\beta \alpha})$ since transitions between the two wells require overcoming an energy barrier of order $\alpha$, while the autocorrelation time is $\mathcal{O}(1)$ because typical stationary samples lie in the dominant well. (b) Symmetric double-well ($h=0$ and $\gamma>0$); the mixing time is $\Omega(e^{\beta \alpha})$, while the autocorrelation time is $\mathcal{O}(1)$ because the energy does not distinguish the two wells and its dominant fluctuations come from the decoupled spin $Z_3$. (c) A classical birth--death chain on $\ket{0},\dots,\ket{m}$; the mixing time scales as $\Omega(m)$, since starting from $\ket{m}$ the chain requires $\Omega(m)$ time to reach the low-energy states, whereas the autocorrelation time is $\mathcal{O}(1)$ since the Gibbs distribution is heavily concentrated near the low-energy states.}
    \label{fig:example}
\end{figure}

\paragraph{Illustrative classical examples} 
In \cref{fig:example}(a)(b), we consider a three-qubit Ising model with external magnetic fields,
\begin{align}
    H = - \alpha Z_1 Z_2 - h (Z_1 + Z_2) - \gamma Z_3, \label{eq:Ising}
\end{align}
Here $Z_i$ denotes the Pauli-$Z$ operator acting on qubit $i$. The Gibbs sampling algorithm considered here is the (single-spin-flip) Glauber dynamics: at each step, a spin is chosen uniformly at random and flipped with probability $\min\{1, e^{-\beta \Delta}\}$, where $\Delta$ denotes the energy difference between the configurations before and after the spin flip.
Here $\alpha, h, \gamma \ge 0$. Suppose our goal is to estimate $\Tr(H \rho_{\beta})$ to finite precision, e.g., $\epsilon = 10^{-6}$.

In the regime $\alpha \gg h \gg \gamma > 0$, as in \cref{fig:example}(a), the Hamiltonian in \cref{eq:Ising} forms an asymmetric double-well landscape: the global minimum $\ket{000}$ lies in the dominant well, and a local minimum $\ket{110}$ lies in the smaller well. The Gibbs probabilities of these minima differ exponentially,
\[
\frac{\Pr[\ket{000}]}{\Pr[\ket{110}]} = e^{4 \beta h}.
\]  
For example, with $\beta = 5$ and $h = 1$, one has $\exp(4 \beta h) \approx 10^{8}$. 

Under Glauber dynamics, escaping the smaller well requires overcoming an energy barrier of order $\alpha$, so the mixing time scales as $\Omega(e^{\beta \alpha})$. In contrast, the autocorrelation time is $\mathcal{O}(1)$ (treating $\gamma$ as a constant), because the stationary distribution is concentrated in the dominant well. Quantitatively, relative to $\ket{000}$, the smaller-well minimum $\ket{110}$ is suppressed by a factor $e^{-4\beta h}$ (about $2\cdot 10^{-9}$ for $\beta=5$ and $h=1$), while configurations with $Z_1Z_2=-1$ are further suppressed by at least $e^{-\beta(2\alpha+2h)}$. Hence, at accuracy $\epsilon=10^{-6}$, stationary samples effectively remain in the dominant well, and the energy decorrelates on the constant timescale of local updates (in particular, rapid flips of the decoupled spin $Z_3$).
  
Let us also consider a different parameter regime $h=0$ and $\gamma>0$. As in \cref{fig:example}(b), the Hamiltonian in \cref{eq:Ising} forms a symmetric double-well landscape in the $Z_1Z_2$ degree of freedom, with two basins corresponding to $\ket{00*}$ and $\ket{11*}$; the global minima are $\ket{000}$ and $\ket{110}$.

Under Glauber dynamics, escaping one well to reach the other requires overcoming an energy barrier of order $\alpha$, so the mixing time scales as $\Omega(e^{\beta \alpha})$. However, since the energy observable $H$ does not distinguish the two wells (it depends on $Z_1Z_2$ and $Z_3$), its autocorrelation time does not reflect the rare inter-well transitions. Instead, the autocorrelation time remains $\mathcal{O}(1)$ because the dominant energy fluctuations come from the local relaxation of the decoupled spin $Z_3$ (the $-\gamma Z_3$ term), which flips at a rate independent of the barrier between the $Z_1Z_2$ configurations.

We remark that the same mechanisms can lead to a sub-exponential separation between the mixing time and the autocorrelation time. For instance, global mixing in the low-temperature 2D Ising model is sub-exponentially slow in the system size, with scaling of the form $\exp(\Theta(\sqrt{n}))$ for $n$ spins~\cite{cesi1996two}. Existing studies~\cite{sokal1997monte, martinelli2004lectures, gheissari2022low} suggest that the dynamics restricted to a single well (phase) can mix in polynomial time. This suggests that the autocorrelation time of the energy may grow at most polynomially, since the energy is a symmetric observable whose relaxation is dominated by rapid fluctuations within a single well (phase) rather than by exponentially slow tunneling between two wells.

Both \cref{fig:example}(a) and (b) involve energy barriers. Finally, even without energy barriers, there can be a polynomial separation between mixing time and autocorrelation time (\cref{fig:example}(c)).
Consider a classical birth--death chain on states $\ket{0},\ket{1},\dots,\ket{m}$. 
The transition probabilities are defined as follows: from state $\ket{k}$, the chain transitions to $\ket{k-1}$ with probability $p$, and to $\ket{k+1}$ with probability $q$, where $p+q=1$, $\log(p/q)=\beta$ and $p>q$. At the boundaries, the chain transitions from $\ket{0}$ to $\ket{1}$ with probability $q$ and stays at $\ket{0}$ with probability $p$, and transitions from $\ket{m}$ to $\ket{m-1}$ with probability $p$ and stays at $\ket{m}$ with probability $q$. One can verify that the invariant distribution $\pi$ is the Gibbs distribution, with $\pi_k\propto\left(q/p\right)^k=\exp(-\beta k)$.
The mixing time is $\Omega(m)$, since starting from $\ket{m}$, reaching the ground state $\ket{0}$ requires at least $m$ steps. In contrast, the autocorrelation time for the energy observable is $\mathcal{O}(1)$, since the spectral gap of this biased birth-death chain is $\mathcal{O}(1)$ (independent of $m$). Intuitively, the autocorrelation time is small because typical configurations sampled from the  Gibbs distribution 
are concentrated near the ground state, providing an effective warm start for the autocorrelation dynamics.

\paragraph{Exponential separation in a quantum Hamiltonian}
We also give a family of quantum Hamiltonians with noncommuting terms for which the mixing time is exponentially larger than the energy autocorrelation time.
More specifically, let $n=2N$ be an even number. Let $\sigma_i^a$, for
$a\in\{x,y,z\}$, denote the corresponding Pauli operator acting on qubit $i$.
Define the total spin operators by
\begin{align}
    S_a=\frac{1}{2}\sum_{i=1}^n \sigma_i^a,
    \qquad a\in\{x,y,z\},
    \qquad
    \bbS^2=S_x^2+S_y^2+S_z^2.
\end{align}
Up to an additive constant, the easy-axis mean-field Heisenberg Hamiltonian~\cite{pearce1975anisotropic,kastner2010nonequivalence} is
\begin{align}
    H_{\alpha}
    :=
    -\frac{J}{n}\left(\bbS^2+(\alpha-1)S_z^2\right),
    \qquad J>0,\quad \alpha>1. \label{eq:easy_axis_H}
\end{align}
where the factor $1/n$ is the standard mean-field normalization, ensuring that the total energy scales linearly with $n$. 
In  Section  \ref{sec:example_quantum} we prove that

\begin{theorem}[Informal]
\label{thm:easy_axis_separation}
Fix an inverse temperature $\beta$  and let the parameters $J,\alpha$ satisfy
$\beta J<2<\beta J\alpha.$
For the corresponding Hamiltonian $H_\alpha$, consider the Gibbs sampler given by the Davies generator with single-qubit Pauli operators as its jump operators.  Then there exist constants $c,C,\delta>0$ such that, for all
sufficiently large even $n$, the mixing time is at least
$c n^{-4}e^{\delta n}$, while the energy autocorrelation time is at most
$Cn^5$. Consequently, the mixing time is exponentially larger than the
energy autocorrelation time.
\end{theorem}

\subsection{Outlook}

In this work, we introduce a framework for estimating quantum thermal expectation values using a single trajectory. For observables that can be measured in a thermodynamically reversible way (i.e., satisfying detailed balance), our approach achieves substantial reductions in sampling cost compared with multi-trajectory methods. The core of our approach is an extension of the notion of autocorrelation from classical Markov chain Monte Carlo to quantum Gibbs sampling, combined with a careful treatment of measurement to control the effects of measurement-induced disturbance.

An exciting direction is to generalize our framework to estimate thermal expectation values for general observables. Here, we investigate the idea of using weighted operator Fourier transform and show that, in certain cases, it can mitigate measurement-induced disturbances. It would be interesting to explore its effectiveness in more general settings through experiments or numerical simulations. Beyond the weighted operator Fourier transform idea, one could also explore methods for measuring general observables in a thermodynamically reversible manner, potentially linking to techniques that map general observables to detailed-balanced Lindbladians~\cite{chen2023efficient,ding2025efficient}. Intuitions from discrete-time Metropolis algorithms may also provide useful guidance~\cite{jiang2024quantum,temme2011quantum,gilyen2024quantum}.

Given the wide success of the corresponding classical MCMC analogy, we expect our single-trajectory framework to have a broad practical impact. An important question for future work is how to accurately estimate the autocorrelation time and the resulting effective sample size, which could be explored through theoretical analysis or numerical simulations, similar to methods used in classical MCMC~\cite{sokal1997monte,goodman2010ensemble,thompson2010comparison}.

\section{Preliminaries and Notation}\label{sec:prelim}

\paragraph{Local Hamiltonian and Gibbs state}
We use $H=\sum_{i=1}^{\kappa} H_i$ to denote an $n$-qubit local Hamiltonian where $\kappa=\poly(n)$ is the number of terms, each $H_i$ acts nontrivially on at most $O(1)$ qubits, and $\|H_i\|\leq 1$.
Given a local Hamiltonian $H$ and inverse temperature $\beta$, the Gibbs state is defined as
\begin{align}
    \rho_{\beta} &= \exp(-\beta H)/Z_\beta,
\end{align}
where $Z_\beta = \Tr(\exp(-\beta H))$ is the normalization factor.

\paragraph{Quantum channel and weighted inner product} We use quantum channels to denote completely positive and trace preserving maps (CPTP).  Let $\cT$ be a quantum channel with Kraus operators  $\{K_u\}_u$, i.e. $\cT(\cdot):= \sum_u K_u \left(\cdot\right) K_u^\dagger$. The map in the Heisenberg picture is defined as 
$$\cT^\dagger(\cdot):= \sum_u K_u^\dagger  \left(\cdot\right) K_u,$$ which is the
dual map w.r.t. the Hilbert-Schmidt inner product 
$\langle A,B\rangle:= \Tr(A^\dagger B).$

 For any $s\in [0,1]$, the weighted inner product is defined as
\begin{align*}
     \langle M,N\rangle_s :=\Tr(M^\dagger  \rho_\beta^{1-s} N \rho_\beta^{s}).
\end{align*}

\paragraph{Quantum detailed balance and spectral gap.} 

We say a quantum channel $\cT$ satisfies $s$-detailed balance (w.r.t. $\rho_{\beta}$) if, for any operators $A$ and $B$, the map $\cT^\dagger$ is self-adjoint w.r.t. the weighted inner product,
\begin{align*}
    \langle A,\cT^\dagger(B)\rangle_s = \langle \cT^\dagger(A), B\rangle_s.
\end{align*}
 When $s=\frac{1}{2}$ we say that $\cT$ satisfies KMS (Kubo-Martin-Schwinger) quantum detailed balance.

\begin{fact}[Properties of quantum detailed balance]\label{fact:QDB}
     Any quantum channel $\mathcal T$ that satisfies $s$-detailed balance fixes the Gibbs state, i.e. $\mathcal T(\rho_\beta)=\rho_\beta$. Moreover,
     \begin{itemize}
         \item $\mathcal T$ is diagonalizable with a real spectrum contained in $[-1,1]$. When $\mathcal T$ arises from the time evolution of a Lindbladian, its spectrum is further restricted to $[0,1]$. 
         \item The spectrum of $\cT^\dagger$ is the complex conjugate of the spectrum of $\cT$; in particular, under $s$-detailed balance the spectrum is real and hence the spectra coincide. 
         \item 
         If  $\mathcal T$ has a unique fixed state and its spectrum is contained in $[0,1]$ (e.g., when $\mathcal T$ arises from the time evolution of a Lindbladian), then its spectral gap can be written as  \begin{align}
        gap(\cT) = gap(\cT^\dagger) = 1- \max_{X\ne 0, \langle X,I\rangle_s=0} \frac{\langle X,\cT^\dagger(X)\rangle_s}{\langle X, X\rangle_s}.
    \end{align}
     \end{itemize}
 \end{fact}

For completeness,  a proof of the above fact is provided in Appendix \ref{app:DB} Lemma \ref{lem:QDB}.

\paragraph{Matrix norm.} 
We use $\bM$ to denote the linear space formed by all $\bC^{2^n\times 2^n}$ matrices. 
For any matrix $A\in \bM$, we use $\|A\|_1$ and $\|A\|$ 
to denote its trace norm and operator norm respectively. For any two indices $x$ and $y$, we use $\delta_{xy}$ to denote the function 
that equals $1$ if $x = y$ and $0$ otherwise.

\paragraph{Expectation, variance and covariance.} Let $\cP$ be a quantum channel with Kraus operators $\{Q_{u}\}_{u}$, where the index set $\{u\}_u$ takes real values and denotes the possible measurement outcomes.  
In other words, $\cP$ is a quantum instrument with classical outcome $u$, so that $\sum_u Q_u^\dagger Q_u=I$.

Consider applying the quantum channel $\cP$ to a state $\rho$. We use 
 $\bE_{\rho}(\cP)$ and $\bV_{\rho}(\cP)$  to denote the expectation and variance of the measurement outcomes, 
 that is,
\begin{align}
	&\bE_{\rho}(\cP):= \sum_u u\,\, \Tr(Q_u \rho Q_u^\dagger),\\
	&\bV_{\rho}(\cP):= \sum_u \left| u -\bE_{\rho}(\cP) \right|^2\,\, \Tr(Q_u \rho Q_u^\dagger).
\end{align}

The covariance of $\cP$ is defined as the correlation between the measurement outcomes of two successive measurements, corresponding to two successive applications of the same instrument with no intermediate evolution,
\begin{align}
    \bCov_{\rho}\left(\cP \right) := \sum_{x,y} \left( x -\bE_{\rho}(\cP) \right)\left( y -\bE_{\rho}(\cP) \right)\,\, \Tr(Q_yQ_x \rho Q_x^\dagger Q_y^\dagger).
\end{align}

Note that when $\cP$ corresponds to a projective measurement, i.e., $Q_x=Q_x^\dagger$ and $Q_x Q_y = \delta_{xy} Q_x$, we have $\bV_{\rho}(\cP)=\bCov_{\rho}\left(\cP \right)$.

\section{Goal and the proposed algorithm}\label{sec:alg}

In this paper, we study the problem of estimating thermal expectation values using an existing quantum Gibbs sampling algorithm as a subroutine.

More precisely,
    we assume that we have access to a quantum Gibbs sampling algorithm, modeled as a quantum channel $\mathcal{N}$ that satisfies $s$-detailed balance and has $\rho_{\beta}$ as its unique fixed state. We assume that the spectrum of $\cN$ lies in $[0,1]$.
The quantum Gibbs-sampling channel $\mathcal{N}$ can be instantiated using known quantum Gibbs samplers~\cite{chen2025efficient,ding2025efficient} that satisfy KMS detailed balance ($s=1/2$),
which are based on simulating Davies-generator-inspired Lindbladian $\cL$. 
More specifically, one can set $\cN:=e^{t_0 \cL}$ for any chosen $t_0>0$. In particular, the spectrum of those
 Lindbladian-based channels $\cN$ is contained in $[0,1]$ by \cref{fact:QDB}.

\paragraph{Our goal and proposed algorithm.} Our goal is to use $\mathcal{N}$ as a subroutine to estimate $\Tr(O \rho_{\beta})$ for a given observable $O$. We use $\mathcal{M}$ to denote a quantum channel that implements the measurement of $O$, in the sense that, when applied to $\rho_{\beta}$, its measurement outcome has expectation value $\Tr(O \rho_{\beta})$.
The implementation of the measurement channel $\mathcal{M}$ will be specified later in \cref{sec:measure}. For illustration, one may temporarily assume $O=H$ and interpret $\mathcal{M}$ as a particular implementation of the quantum phase estimation algorithm.

As described in the introduction and shown in Figure~\ref{fig:single_trajectory} (b), we estimate observable values in quantum Gibbs sampling using a single trajectory, as outlined in Algorithm~\ref{alg:intro_main}.
For proof purposes, in line \ref{line:twice} of Algorithm~\ref{alg:intro_main}, we apply the measurement $\mathcal{M}$ twice rather than once. In particular, $\cE^2 = \cM \circ \cN\circ \cM \circ  \cM \circ \cN\circ \cM $, thus the measurement channel is applied twice consecutively between two sampling steps.
Additionally, Algorithm~\ref{alg:intro_main} is currently presented in a practical i.e., agnostic mode, where  measurements are performed at fixed intervals; in Section~\ref{sec:aut_gap}, Remark~\ref{remark:skip_measure} further discusses how the cost of measurements can be reduced by skipping unnecessary measurements.

 \begin{algorithm}
  \caption{Observable Estimation from a single trajectory}\label{alg:intro_main}
  \hspace*{\algorithmicindent} \textbf{Inputs:}   Quantum channel $\cN$ and $\cM$; Initial state $\rho$; Integers $T_{burn}$ and $K$
\begin{algorithmic}[1]
    \State  Apply $\cN$ for time $T_{burn}$ to the initial state $\rho$
    \For{$t=1$ to $K$}
    \State Apply $\cE:=\cM \circ \cN \circ \cM$   \label{line:twice}
        \State 
        Denote the measurement outcome of the first $\cM$ (i.e., the rightmost $\cM$ in the composition defining $\cE$) as $e_t$
    \EndFor
    \State Output 
    $X_K:=\frac{1}{K}\sum_{t=1}^K e_t$
\end{algorithmic}
\end{algorithm}

\paragraph{Structure of the manuscript.} 

We will first assume that $\cM$ satisfies $s$-detailed balance and analyze the performance of \cref{alg:intro_main}. In particular, in \cref{sec:DB_gap} we prove that detailed-balance measurements never decrease the spectral gap. In \cref{sec:time_average}, we introduce the concept of autocorrelation time in quantum Gibbs sampling and bound it by the relaxation time multiplied by an additional factor, which will later be shown to be constant for our choice of $\mathcal{M}$.

Then we will explain the implementation of the measurement $\cM$. In \cref{sec:M_com}, we focus on operators $O$ that commute with $H$ and construct a measurement $\cM$ that is unbiased, satisfies $s$-detailed balance, and incurs only a logarithmic overhead. As an explicit example, in \cref{sec:energy}, we focus on measuring $H$ and give a detailed comparison of the gate and depth complexity between different approaches of measuring $H$. Then we employ the operator Fourier transform technique to mitigate the measurement-induced disturbance for more general observables.  In Section~\ref{sec:example_quantum}, we give a family of quantum Hamiltonians for which there is an exponential separation between the mixing time and the energy autocorrelation time.

\section{Detailed-balanced measurements do not decrease the spectral gap}\label{sec:DB_gap}

In this section, we prove that  detailed-balanced measurements do not decrease the spectral gap. First notice that the following lemma holds. Its proof is provided in Appendix~\ref{app:DB}.

\begin{lemma}\label{lem:unique}
    Let $\cM$ and $\cN$ be quantum channels satisfying $s$-detailed 
balance with respect to $\rho_\beta$; then the composite channel 
$\cM \cN \cM$ also satisfies $s$-detailed balance. 
Moreover, if $\cN$ has $\rho_\beta$ as its unique fixed point and its 
spectrum is contained in $[0,1]$, then $\cM \cN \cM$ 
also has $\rho_\beta$ as its unique fixed point and its spectrum lies in $[0,1]$.
\end{lemma}

Now we are ready to state the lemma that connects the spectral gap of $\cM\cN\cM$ and $\cN$.

\begin{lemma}\label{lem:MNM_N}
   Consider two quantum channels $\cM$ and $\cN$ that satisfy $s$-detailed balance with respect to $\rho_\beta$. Suppose $\cN$ admits $\rho_\beta$ as its unique fixed point and has spectrum contained in $[0,1]$. Then we have that
    \begin{align}
        gap(\cM\cN\cM) \geq gap(\cN).
    \end{align}
\end{lemma}

To prove \cref{lem:MNM_N}, first notice that a detailed-balance channel is contractive under the weighted inner product.

\begin{lemma}[Contractivity of a detailed-balance channel under the weighted norm] \label{lem:contract}Let $\cT$ be a quantum channel that satisfies $s$-detailed balance. Then, for any operator $A$, we have
	\begin{align}
	   \langle A,A \rangle_s  \geq 	\langle \cT^\dagger(A),\cT^\dagger(A) \rangle_s. 
	\end{align}
\end{lemma}
\begin{proof} 
    Since $\cT$ satisfies $s$-detailed balance, we have that $\cT^\dagger$ is Hermitian w.r.t. the weighted inner product $\langle,\rangle_s$. Thus, there exists an orthonormal basis of $\bM$, that is, $\{Y_i\}_i\subseteq \bM$ where $\langle Y_i,Y_j\rangle_s =\delta_{ij}$ such that  
     $\cT^\dagger (Y_i)=\lambda_i Y_i$. Since $\cT^\dagger$ is Hermitian w.r.t. $\langle,\rangle_s$, we have $\lambda_i\in\bR$. Moreover, $|\lambda_i|\leq 1$ since $\cT$ is a quantum channel by \cref{fact:QDB}.  
     
    Since $\{Y_i\}_i$ forms a basis of $\bM$, we can  express the operator $A$ w.r.t the eigenstates $\{Y_i\}_i$,
     \begin{align}
     	A = \sum_i a_i Y_i, \text{ where } a_i=\langle Y_i,A\rangle_s.
     \end{align}
     Then it suffices to notice that $|\lambda_i|\leq 1$ implies 
     \begin{align}
         	  \langle \cT^\dagger(A),\cT^\dagger(A) \rangle_s = \sum_{i} \lambda_i^2 |a_i|^2 \leq \sum_i |a_i|^2  =\langle A,A\rangle_s.
     \end{align}
\end{proof}

Now we are prepared to prove \cref{lem:MNM_N}.

\begin{proof}[of Lemma \ref{lem:MNM_N}] 
From Lemma \ref{lem:unique} we have that $\cM \cN \cM$  has $\rho_\beta$ as its unique fixed point and spectrum in $[0,1]$. 
Since both $\cM$ and $\cN$ satisfy $s$-detailed balance, then so does $\cM\cN\cM$, 
by the definition of spectral gap [Fact \ref{fact:QDB}] we have that
\begin{align}
    gap(\cN) &= 1- \max_{X\neq 0, \langle X,I\rangle_s=0}\frac{\langle X, \cN^\dagger\left(X\right)\rangle_s}{\langle X,X\rangle_s} \\
    &\leq 1- \max_{X\neq 0, \langle \cM^\dagger(X),I\rangle_s=0}\frac{\langle \cM^\dagger(X), \cN^\dagger \cM^\dagger(X)\rangle_s}{\langle \cM^\dagger(X),\cM^\dagger(X)\rangle_s}. \label{eq:10}
\end{align}
where the last inequality comes from restricting the maximization to the subset of test operators of the form $\cM^\dagger(X)$.
Note that $\cM$ is a quantum channel thus $\cM^\dagger(I)=I$. Combine with the fact that  $\cM$ satisfies $s$-detailed  balance,  we have 
\begin{align}
    &\langle \cM^\dagger(X),I\rangle_s = \langle X,\cM^\dagger(I)\rangle_s = \langle X,I\rangle_s,\\
    & \langle \cM^\dagger(X), \cN^\dagger \cM^\dagger(X)\rangle_s = \langle X, \cM^\dagger \cN^\dagger \cM^\dagger(X)\rangle_s \geq 0. 
\end{align}
where the $\geq 0$ comes from  the fact that the spectrum of $\cN^\dagger$ lies in $[0,1]$ since $\cN$ has spectrum in $[0,1]$ [Fact \ref{fact:QDB}].
Thus by Lemma \ref{lem:contract} and Eq.~(\ref{eq:10}), we have that
\begin{align}
    gap(\cN) &\leq 1- \max_{X\neq 0, \langle X,I\rangle_s=0}\frac{\langle X, \cM^\dagger \cN^\dagger \cM^\dagger(X)\rangle_s}{\langle \cM^\dagger(X),\cM^\dagger(X)\rangle_s}\label{eq:13}\\
    &\leq 1- \max_{X\neq 0, \langle X,I\rangle_s=0}\frac{\langle X, \cM^\dagger \cN^\dagger \cM^\dagger(X)\rangle_s}{\langle X,X\rangle_s} \label{eq:variance}\\
    &= gap(\cM\cN\cM).
\end{align}
\end{proof}

\section{Time average behavior of Algorithm \ref{alg:intro_main}} \label{sec:time_average}

In this section, we analyze the performance of our single-trajectory observable estimation algorithm, Algorithm \ref{alg:intro_main}, under the assumption that the measurement $\cM$ satisfies $s$-detailed balance. 
Section \ref{sec:pesu} defines the spectral covariance weight and explains its relationship to the variance. In Section \ref{sec:time}, we present the definitions of mixing time, relaxation time, and introduce the  autocorrelation time in the context of quantum Gibbs sampling. Finally, Section \ref{sec:aut_gap} bounds the autocorrelation time and analyzes the behavior of time-averaged observable values along the trajectory.

We first introduce some notation that will be used throughout this section.

\paragraph{Channels and their eigenstates.}
In this section  we consider applying Algorithm \ref{alg:intro_main} with  quantum channels $\cN$ and $\cM$ that satisfy the following properties.
\begin{assumption}\label{ass:NM}
    Consider two quantum channels $\cN$ and $\cM$ that satisfy 
$s$-detailed balance with respect to $\rho_\beta$. Suppose that 
$\cN$ has $\rho_\beta$ as its unique fixed state and that its 
spectrum lies in $[0,1]$. 
\end{assumption}
As we discussed in Section \ref{sec:alg}, $\cN$ refers to quantum Gibbs samplers like \cite{chen2025efficient,ding2025efficient}.
 By Lemma~\ref{lem:unique}, the composite channel 
$\mathcal E := \cM \cN \cM$
also  has $\rho_\beta$ 
as its unique fixed state and its spectrum is contained in $[0,1]$.
Besides, $\cE$ also satisfies $s$-detailed balance, thus $\cE^\dagger$ is Hermitian w.r.t. the weighted inner product $\langle,\rangle_s$. Thus
     there exists a set of matrices $\{Y_i\}_i$ such that
\begin{itemize}
	\item Orthonormal w.r.t weighted inner product: $\langle Y_i,Y_j\rangle_s =\delta_{ij}$.
	\item Basis: $\{Y_i\}_i$ forms a basis of $\bM$.
    \item  Eigenstate: $\cE^\dagger (Y_i)=\lambda_i Y_i$, $\lambda_i\in [0,1]$.  In particular $Y_1=I$ i.e. the identity matrix.  Without loss of generality, we organize $\lambda_i$ in non-increasing order
     $1=\lambda_1> \lambda_2 \geq \lambda_3...$. 
\end{itemize}

\paragraph{Expectation, Variance and Covariance} 
Let $\mathcal P$ be a quantum channel with Kraus operators $\{Q_{u}\}_{u}$, 
where the index set $\{u\}_u$ takes real values and denotes the possible measurement outcomes. 

Consider applying $\cP$ to the Gibbs state $\rho_{\beta}$.  
 As defined in the preliminary section (Section \ref{sec:prelim}), we use $\bE_{\rho_{\beta}}(\cP)$, $\bV_{\rho_{\beta}}(\cP)$ and $\bCov_{\rho_{\beta}}(\cP)$ to denote the expectation, variance, and covariance of the measurement $\cP$.

\paragraph{Autocorrelation function}
For technical reasons, it is more convenient to treat 
$\mathcal E = \cM \cN \cM$ as a single channel 
rather than analyze $\cN$ and $\cM$ separately. 
In particular, we regard $\mathcal E$ as a measurement channel whose 
outcome is the outcome of the first application of $\cM$, 
consistent with Algorithm~\ref{alg:intro_main}.

Denote the Kraus operators of $\cM$ by $\{O_u\}_u$. 
Define the map that sends an operator to its 
measurement--outcome--weighted operator as follows:
\begin{align}
	&\widehat{\cM}(X) := \sum_u \left(u-\bE_{\rho_{\beta}}(\cM)\right)  \, O_u X O_u^\dagger \label{eq:hatM},\\
     &\widehat{\cE}(X):= \cM \circ \cN\circ \widehat{\cM} (X).  \label{eq:hatE}
    \end{align}

The autocorrelation function of $\mathcal E$ at time lag $t$ is the correlation 
between the measurement outcomes at time steps $0$ and $t+1$:
\begin{align}
    \bCor_{\rho_{\beta}}\left(\cE^{t}\right) &:= 
    \Tr\left(  \widehat{\cE} \circ \cE^t  \circ     \widehat{\cE}\left(\rho_\beta\right)   \right).
\end{align}

In particular, $\bCor_{\rho_{\beta}}\left(\cE^{0}\right) = 
    \Tr\left(  \widehat{\cE}   \circ     \widehat{\cE}\left(\rho_\beta\right)   \right). $

\subsection{Eigenfunctions and spectral covariance weight}\label{sec:pesu}

Recall that $\{Y_i\}_i$ are eigenstates of $\mathcal E^\dagger$ that form a basis of the operator space $\bM$. To aid the proof, we express the map
$\widehat{\mathcal E}^\dagger$ in terms of these eigenstates $\{Y_i\}_i$. That is, for any 
$Y_i$, we write
\begin{align}
	& \widehat{\cE}^\dagger(Y_i) = \sum_{j} \alpha_{ij} Y_j, \label{eq:def_alpha} \text{ where }  \alpha_{ij}:= \langle Y_j, \widehat{\cE}^\dagger(Y_i)\,\rangle_s. 
\end{align}

Recall that $Y_1=I$. Note that $\cM^\dagger(I)=I$ since $\cM$ is a quantum channel and similarly for $\cN$. 
\begin{lemma}\label{lem:alpha} We have $\alpha_{11} = 0$ since
$$\alpha_{11} = \langle I, \widehat{\cM}^\dagger\cN^\dagger \cM^\dagger \left(I\right)\rangle_s = \Tr\left(\rho_{\beta}  \widehat{\cM}^\dagger(I)\right)=\Tr\left(\widehat{\cM}(\rho_{\beta})\right)  = 0.$$ 
\end{lemma}

We define the spectral covariance weight as
\begin{definition}[Spectral covariance weight] The spectral covariance weight of the map $\cE=\cM\cN\cM$ is defined as 
    $$
\SC:= \sum_{j} |\alpha_{1j} \alpha_{j1}|.$$
\end{definition}

The spectral covariance weight will be used in Section \ref{sec:proof_aut_gap} to bound the autocorrelation time in terms of the spectral properties of the Gibbs sampling channel, namely its spectral gap.
Note that in the definition, $\alpha_{j1}$ might \textit{not} equal to $\alpha_{1j}^*$. 

In the following, we show that when $\widehat{\cM}$ also satisfies 
$s$-detailed balance, the spectral covariance weight of $\cE$ can be upper-bounded by the 
covariance of $\cM$, which equals the variance of $\cM$ if 
$\cM$ is a projective measurement. To avoid confusion, we remark 
that for operators commuting with $H$, later in Section \ref{sec:M_com} we will explicitly implement a measurement 
$\cM$ such that both $\cM$ and $\widehat{\cM}$ 
satisfy $s$-detailed balance.

\begin{lemma}\label{lem:j1_to_1j} Suppose the map $\widehat{\cM}$ satisfies $s$-detailed balance, then 
   $$
        \SC \leq   \bCov_{\rho_{\beta}}\left(\cM\right).
   $$
\end{lemma}

\begin{proof}
By Eq.~(\ref{eq:def_alpha}), we have that 
\begin{align}
    \sum_j |\alpha_{1j}|^2 = \langle \widehat{\cE}^\dagger(I),\widehat{\cE}^\dagger(I) \rangle_s  =  \langle \widehat{\cM}^\dagger(I),\widehat{\cM}^\dagger(I) \rangle_s 
\end{align}
where we implicitly use $\cM^\dagger(I)=I$, $\cN^\dagger(I)=I$ since $\cM$ and $\cN$ are quantum channels.

Besides, since $\widehat{\cM}, \cM$ and $\cN$ satisfy $s$-detailed balance, we have that
\begin{align}
    &\alpha_{j1}=\langle I, \widehat{\cE}^\dagger(Y_j)\rangle_s = \langle I, \widehat{\cM}^\dagger \cN^\dagger \cM^\dagger (Y_j)\rangle_s =  \langle   \cM^\dagger\cN^\dagger\widehat{\cM}^\dagger(I),    Y_j\rangle_s, \\
    \text{thus } & \cM^\dagger\cN^\dagger\widehat{\cM}^\dagger(I) = \sum_j \alpha_{j1}^* Y_j.
\end{align}
Thus we have
\begin{align}
    \sum_j |\alpha_{j1}|^2 &= \langle \cM^\dagger\cN^\dagger\widehat{\cM}^\dagger(I),\cM^\dagger\cN^\dagger\widehat{\cM}^\dagger(I)\rangle_s \\
    &\leq \langle \cN^\dagger\widehat{\cM}^\dagger(I),\cN^\dagger\widehat{\cM}^\dagger(I)\rangle_s\\
    &\leq \langle \widehat{\cM}^\dagger(I),\widehat{\cM}^\dagger(I)\rangle_s
\end{align}
where for the last two inequalities we use the fact that since $\cM$, $\cN$ satisfy $s$-detailed balance, then $\cM^\dagger$ and  $\cN^\dagger$ are contractive w.r.t. weighted norm [Lemma \ref{lem:contract}].
Besides, since $\widehat{\cM}$ satisfies $s$-detailed balance, we further have that 
$$\langle \widehat{\cM}^\dagger(I),\widehat{\cM}^\dagger(I)\rangle_s =  \langle I, \widehat{\cM}^\dagger\widehat{\cM}^\dagger(I)\rangle_s = \bCov_{\rho_{\beta}}\left(\cM \right).$$
We complete the proof of the Lemma by using Cauchy-Schwarz inequality.

\end{proof}

\subsection{Mixing Time, Relaxation Time, and Autocorrelation Time}\label{sec:time}

Recall that $\cN$  and $\cM$ are quantum channel satisfying $s$-detailed balance 
with respect to $\rho_\beta$, which $\cN$ represents the Gibbs sampling algorithm, and $\cM$ represents the measurement channel.

The \textbf{mixing time} of $\cN$ is 
defined as the number of channel applications required to drive any initial state close to the Gibbs state,
\begin{align}
    t_{mix}(\epsilon):= \min_{t \in\mathbb{N}} \left\{t\geq 0: \|\cN^t(\rho)-\rho_\beta\|_1 \leq \epsilon, \forall \text{ quantum state }\rho \right\}.
 \end{align}

The \textbf{relaxation time} of $\cN$ is defined as the inverse of the spectral gap
\begin{align}
    t_{rel} := \frac{1}{gap(\cN)}.
\end{align}

The relationship between mixing times and relaxation time is given by
\begin{align}
    \left(t_{rel}-1\right) \log\left(\frac{1}{\epsilon}\right) \leq t_{mix}(\epsilon) \leq  t_{rel} \log\left(\frac{1}{\epsilon\,\sigma_{min}}\right).\label{eq:mix_rel}
\end{align}
where $\sigma_{min}$ denotes the minimum eigenvalue of the Gibbs state, which is typically exponentially small. For completeness, we put a proof of Eq.~(\ref{eq:mix_rel}) in Appendix \ref{app:DB} Lemma \ref{lem:mix_rel}.

\paragraph{Autocorrelation time}
Recall that our trajectory-based observable estimation algorithm, Algorithm~\ref{alg:intro_main}, consists of two stages: a \textit{burn-in stage}, in which the chain is evolved for $T_{burn} \sim t_{mix}$ to remove dependence on the initial state and reduce bias; and a \textit{sampling stage}, where the quantum system is approximately stationary and its trajectory is used to estimate observables.
The autocorrelation time quantifies the time in the sampling stage required for two outcomes to become approximately independent, with respect to the measurement $\cM$.

Assuming the quantum system is already in the Gibbs state,
 the normalized autocorrelation function w.r.t. $\cM$ is defined as
\begin{align}
     C^{\cM}_{\rho_{\beta}}(t): = \frac{  \bCor_{\rho_{\beta}}\left(\cE^t\right)}{\bV_{\rho_{\beta}}(\cM)} \text{ where } \cE=\cM\cN\cM.\label{eq:normalized_var}
\end{align}
For a finite trajectory of length $K$, the (finite-size integrated) autocorrelation time is defined as
\begin{align}
    t_{aut,K}:= \frac{1}{2} +  \sum_{t=1}^{K} \left( 1- \frac{t}{K} \right)C^{\cM}_{\rho_{\beta}}(t-1).\label{eq:t_auk}
\end{align}
Here, $C_{\rho_\beta}^{\cM}(t-1)$ is the normalized autocorrelation function between  measurement outcomes  separated by $t$ sampling steps. For a trajectory of length $K$, there are $K-t$ such pairs, which gives the finite-length weight $(1-t/K)$. The factor $1/2$ is a normalization convention.

Note that the autocorrelation time is naturally connected to the performance of empirical average, 
by noticing that 
    \begin{align}
        K \bV_{\rho_{\beta}}\left(\cM \right) \, 2 t_{aut,K} 
         &=  K \bV_{\rho_{\beta}}\left(\cM \right) +  2\sum_{t=1}^K \sum_{p=1}^{K-t} \bCor_{\rho_{\beta}}\left(\cE^{p-1}\right)\\
        & = \bE_{\rho_{\beta}}\left( \sum_{t=1}^K e_t - K \,\bE_{\rho_{\beta}}(\cM)\right)^2,
    \end{align}
where the last equality can be verified using Eqs.~(\ref{eq:34})–(\ref{eq:37}) in Section \ref{sec:smi}.
Equivalently, dividing the above equation by $K^2$, we have
\begin{align}
  \mathbb{E}_{\rho_\beta}\left[
\left(
\frac{1}{K}\sum_{t=1}^{K}e_t
-\mathbb{E}_{\rho_\beta}(\mathcal M)
\right)^2
\right]
=
\frac{2\mathrm{Var}_{\rho_\beta}(\mathcal M)
t_{\mathrm{aut},K}}{K}
\end{align}
Therefore, under our convention, a trajectory of length $K$ has an \textit{effective sample size}
$\frac{K}{2t_{aut,K}}.$   Equivalently, after the initial burn-in, the sampling cost per effective sample is $2t_{\mathrm{aut},K}$, rather than by a full mixing time. The factor of $2$ here is only due to our normalization convention.

We will prove that the autocorrelation time $t_{aut,K}$ is upper bounded by the relaxation time $t_{rel}$ up to an additional factor (depending on the measurement) and an additive constant. Note that compared to $t_{rel}$, for $\epsilon=1/2$, the mixing time $t_{mix}(1/2)$ includes an additional factor $\log(\sigma_{min}^{-1})$, which is often proportional to the system size $n$ and represents the transport time across the state space.  We remark that $t_{rel}$ only gives an upper bound for $t_{aut,K}$.
In practice for observables of interest, the separation  between $t_{aut,K}$ and $t_{mix}$ could be even larger than this factor $\log(\sigma_{min}^{-1})$, as illustrated in the examples in Section \ref{sec:example}.

All the following lemmas are based on the finite-size autocorrelation time. For completeness, and in accordance with the classical MCMC literature, we note that the autocorrelation time for an infinitely long trajectory is defined as
\begin{align}
    t_{aut,\infty}:= \frac{1}{2} +  \sum_{t=1}^{\infty} C^{\cM}_{\rho_{\beta}}(t-1) \label{eq:aut}.
\end{align}

\subsection{Stationary moment identities}\label{sec:smi}

Before analyzing the performance of Algorithm~\ref{alg:intro_main}, we first record some basic identities for the random variable $e_t$, which represents the measurement outcome in the algorithm. We denote by $\bE_{\rho}[e_t]$ the expectation of $e_t$ when the algorithm is initialized in state $\rho$.
Similar notation is used for $\bE_{\rho}[(e_t)^2]$ 
 and the probabilities, e.g., $\bPr_{\rho}(\cdot)$.

\paragraph{Stationary case} To simplify the analysis, we temporarily assume that the initial state of the sampling stage is $\rho=\rho_{\beta}$.
To ease the notation, denote
 \begin{align}
 v:=  \bE_{\rho_{\beta}}(\cM). \label{eq:34}
 \end{align}
 Recall that the Kraus operators of $\cM$ is $\{O_u\}_u$. 
When the initial state $\rho$ is $\rho_{\beta}$, one can check that 
\begin{align}
	&\bE_{\rho_{\beta}}\left[e_t\right]   =  \sum_u\,u\,\, \Tr\left( O_u \rho_{\beta} O_u^\dagger    \right)= \bE_{\rho_{\beta}}(\cM), \label{eq:35}\\
    &\bE_{\rho_{\beta}}\left[(e_t-v)^2\right]   = \sum_u\,\left|u -\bE_{\rho_{\beta}}(\cM)\right|^2\, \Tr\left( O_u \rho_{\beta} O_u^\dagger    \right)= \bV_{\rho_{\beta}}(\cM).\label{eq:36}
\end{align}

For any $t, p>0$, we have
\begin{align}
     \bE_{\rho_{\beta}}\left[ \left( e_t -  v\right) \left( e_{t+p} -  v \right) \right]  = \bCor_{\rho_{\beta}}\left(\cE^{p-1}\right),\label{eq:37}
\end{align}
since 
    \begin{align}
         \bE_{\rho_{\beta}}\left[ \left( e_t -  v\right) \left( e_{t+p} -  v \right) \right] &=\sum_{xy}  \left(x-v\right)\left(y-v\right) \, \Tr\left(   O_y \left( \cE^{p-1}\circ\cM\circ\cN \left( O_x \rho_{\beta} O_x^\dagger\right) \right) O_y^\dagger \right) \\
         &=\sum_{xy}  \left(x-v\right)\left(y-v\right) \, \Tr\left(   \cM \circ\cN \left(O_y \left( \cE^{p-1}\circ\cM\circ\cN \left( O_x \rho_{\beta} O_x^\dagger\right) \right) O_y^\dagger \right) \right)\label{eq:222}\\
         & =\Tr\left(  \widehat{\cE} \circ \cE^{p-1}\circ \widehat{\cE}\left(\rho_{\beta} \right)\right),
    \end{align}
      where Eq.~(\ref{eq:222}) comes from the fact that $\cM\circ\cN$ is a quantum channel thus trace preserving.

\subsection{Reduction from arbitrary initialization}
Now we consider the case where the algorithm is initialized in an arbitrary state and a burn-in stage is used to reduce bias.

Denote $\rho_{burn} := \cN^{T_{burn}}(\rho)$ as the state after the burn-in period.

\begin{lemma} \label{lem:aut}
Consider  quantum channels $\cN$ and $\cM$ that satisfy the $s$-detailed balance.
Let $e_t$ be the measurement outcome  in Algorithm \ref{alg:intro_main} w.r.t  an arbitrary initial state $\rho$.
  For any $\epsilon, \eta>0$,  set $T_{burn}=t_{mix}\left( \eta\right)$ and
    \begin{align}
    	K\geq \frac{2 \bV_{\rho_\beta}(\cM)}{ \epsilon^2 \eta} \, t_{aut,K}.
    \end{align}
    We have that
    \begin{align}
        \bPr_{\rho} \left( \left|\frac{1}{K}\left(\sum_{t=1}^K e_t\,\right) - \bE_{\rho_{\beta}}(\cM)\right|\geq \epsilon \right) \leq 2\eta.
            \end{align}
\end{lemma}
\begin{proof}
First, we show that one can assume the initial state is $\rho_\beta$, at the cost of an additional error $\eta$ in the probability. 

Specifically, we treat all measurements performed after the burn-in period, together with the intermediate Gibbs sampling channels, as a single composite measurement. The composite measurement outcome is then $\vec{e} := (e_1, \dots, e_K) \in \mathbb{R}^K$, where we denote the associated Kraus operator as $Q_{\vec{e}}$. 
Denote $S$ as the set of $\vec{e}$ where the random variable $X_K=(e_1+...e_K)/K$ satisfies $\left|X_K- \bE_{\rho_{\beta}}(\cM)\right|\geq \epsilon$. 
Note that 
\begin{align}
   \left|  \bPr_{\rho} \left(\vec{e}\in S\right)-  \bPr_{\rho_{\beta}} \left(\vec{e}\in S\right) \right| =  \left|  \sum_{\vec{e} \in S}  \Tr\left(Q_{\vec{e}} \left(\rho_{burn}-\rho_{\beta}\right)Q_{\vec{e}}^\dagger \right) \right| \leq \left\|\rho_{burn}-\rho_{\beta}\right\|_1\leq \eta.
\end{align}
Thus it suffices to assume that the initial state is $\rho_{\beta}$ and change the error probability from $2\eta$ to $\eta$, that is, prove that
\begin{align}
        \bPr_{\rho_{\beta}} \left( \left|\frac{1}{K}\left(\sum_{t=1}^K e_t\,\right) - \bE_{\rho_{\beta}}(\cM)\right|\geq \epsilon \right) \leq \eta.
\end{align}
which follows directly from Chebyshev inequality and the formulas in the ideal case, i.e. Eqs.~(\ref{eq:35})(\ref{eq:36})(\ref{eq:37}).
\end{proof}

\subsection{Time average behavior of the empirical average}
\label{sec:aut_gap}

Our main result is to bound the autocorrelation time $t_{aut,K}$ by the inverse spectral gap of the original quantum Gibbs sampling channel $\cN$, up to an additional factor.

\begin{lemma}\label{lem:aut_gap} Let $\cN$ and $\cM$ be quantum channels satisfying Assumption \ref{ass:NM}.  Suppose the map  $\widehat{\cM}$ also satisfies $s$-detailed balance. 
   Then the autocorrelation time is bounded by
    \begin{align}
 t_{aut,K} \leq \frac{1}{gap(\cN)} \, \bp + \frac{1}{2}, \text{  where  }
 	\bp := \frac{\bCov_{\rho_{\beta}}\left(\cM\right)}{\bV_{\rho_{\beta}}(\cM)}. \label{eq:theta}
 \end{align}
\end{lemma}

The proof of Lemma \ref{lem:aut_gap} is deferred to Section~\ref{sec:proof_aut_gap}. 

Together with Lemma~\ref{lem:aut} and Lemma~\ref{lem:aut_gap}, we get the theorem for the performance of the empirical averages.

\begin{theorem}[Time averages]\label{thm:time_average}   Let $\cN$ and $\cM$ be quantum channels satisfying Assumption \ref{ass:NM}. Suppose the map  $\widehat{\cM}$ also satisfies $s$-detailed balance. 
Let $e_j$ be the measurement outcome  in Algorithm \ref{alg:intro_main} w.r.t an arbitrary initial state $\rho$.  
For any $\epsilon,\eta>0$,
set $T_{burn}:=t_{mix}(\eta)$ and
\begin{align}
    	&K\geq \frac{2  \bV_{\rho_\beta}(\cM)}{\epsilon^2 \eta}\left(\frac{1}{gap(\cN)} \, \bp + \frac{1}{2}
 \right) \text{ for $\bp = \frac{\bCov_{\rho_{\beta}}\left(\cM\right)}{\bV_{\rho_{\beta}}(\cM)}$}.
 \end{align}
 We have that
 \begin{align}
     \bPr_{\rho} \left( \left|\frac{1}{K}\left(\sum_{t=1}^K e_t\,\right) - \bE_{\rho_{\beta}}(\cM)\right|\geq \epsilon \right) \leq 2\eta.
    \end{align}
 \end{theorem}

\begin{remark}\label{remark:skip_measure} For observables that commute with $H$, we  construct $\cM$ where $\bp\leq 1$. Details are presented later in Theorem~\ref{thm:M_com} in Section~\ref{sec:GQPE}. 
    We remark that the measurement cost can be reduced by skipping measurements. 
    Let $c_{\cN}$ and $c_{\cM}$ denote the costs of running quantum channels $\cN$ and $\cM$, respectively. 
    In the above theorem, the total cost of Algorithm~\ref{alg:intro_main} is the sum of the following components:
\begin{align}
  \text{cost of Gibbs sampling evolution }&\propto  t_{mix}(\eta) +    \frac{1}{gap(\cN)\, \epsilon^2 \eta}  \left(c_{\cN} \bV_{\rho_\beta}(\cM) \right), \label{eq:rem_N}\\
  \text{cost of measurement }& \propto \frac{1}{gap(\cN)\, \epsilon^2 \eta}  \left(c_{\cM}\bV_{\rho_\beta}(\cM) \right).\label{eq:rem_M}
\end{align}
To ease notation, denote
        $\Delta := gap(\cN).$
Let $r:=\lceil 1/\Delta\rceil$. Note that the spectral gap of $\cN^{r}$ is
$$gap\left(\cN^{r}\right)=  1- (1-\Delta)^{r}\geq 1- e^{-r\Delta}\geq 1- e^{-1}.$$ 
When $\Delta$ is small and the measurement cost $c_{\mathcal{M}}$ is high, the total measurement cost can be reduced by skipping measurements. That is, one can replace $\mathcal{N}$ with $\mathcal{N}^{r}$, or equivalently, apply the measurement $\mathcal{M}$ every $r$ time evolutions of $\mathcal{N}$ instead of after every evolution. Note that the cost of total Gibbs sampling evolution in Eq.~(\ref{eq:rem_N}) remains of the same order, since by $gap\left(\cN^{r}\right)\geq 1- e^{-1}$ we have that
\[
\frac{c_{\mathcal{N}}}{gap(\mathcal{N})} \sim  \frac{c_{\mathcal{N}^{r}}}{gap(\mathcal{N}^{r})}.
\]

\end{remark}

\subsection{Bounding autocorrelation time by the reciprocal spectral gap}\label{sec:proof_aut_gap}

In this section we prove Lemma \ref{lem:aut_gap}. Recall that to ease notations we write 
$$
    v= \bE_{\rho_{\beta}}(\cM).$$
Recall that $Y_i,\lambda_i$ are eigenvector and eigenvalues w.r.t. the map $\cE^\dagger$.
We begin by relating the correlations of the measurement outcomes to these eigenvalues.

\begin{lemma} \label{lem:cov_eigenvalue}  Let $\cN$ and $\cM$ be quantum channels satisfying Assumption \ref{ass:NM}. 
For any integer $t,p>0$, we have that
	\begin{align}
		\bE_{\rho_\beta}\left[ \left( e_t -  v\right) \left( e_{t+p} -  v \right) \right]  =  \sum_{j\geq 2} \alpha_{1j} \alpha_{j1} \lambda_j^{p-1}. 
	\end{align}
\end{lemma}
 
\begin{proof} Recall that $Y_1=I$. From Eq.~(\ref{eq:37}) and Eq.~(\ref{eq:def_alpha}) we have that 
\begin{align}
    \bE_{\rho_{\beta}}\left[ \left( e_t -  v\right) \left( e_{t+p} -  v \right) \right]   &=\Tr\left(   \widehat{\cE} \circ \cE^{p-1}\circ \widehat{\cE}\left(\rho_{\beta} \right)\right)\\
    &= \Tr\left( \left( \widehat{\cE}^\dagger \left(\cE^{p-1}\right)^\dagger \widehat{\cE}^\dagger\right)[I]\,\, \rho_{\beta} \right)\\
    &= \sum_j \alpha_{1j} \lambda_j^{p-1} \sum_l \alpha_{jl} \, \langle Y_l, Y_1\rangle_s \\
      &= \sum_j \alpha_{1j} \alpha_{j1}  \lambda_j^{p-1},
\end{align}
where we use the property that $\langle Y_l, Y_1\rangle_s=\delta_{l1}$.
Thus we complete the proof by recalling that $\alpha_{11}=0$ [Lemma \ref{lem:alpha}].

\end{proof}

\begin{lemma} \label{lem:time_average_var} Let $\cN$ and $\cM$ be quantum channels satisfying Assumption \ref{ass:NM}. We have that
	\begin{align}
		\bE_{\rho_{\beta}}\left[ \left(\sum_{t=1}^K e_t -K v\,\right)^2 \right] \leq 		
		 K \left( \bV_{\rho_{\beta}}(\cM)   +   \frac{2}{gap(\cN)} \,\SC \right).
	\end{align}
\end{lemma} 
\begin{proof} By Lemma \ref{lem:cov_eigenvalue},
we have that
\begin{align}
		\bE_{\rho_{\beta}}\left[ \left(\sum_{t=1}^K e_t - K v\,\right)^2 \right] &= \sum_{t=1}^K \bE_{\rho_{\beta}}\left[ (e_t-v)^2\right] +  \sum_{t=1}^K \sum_{p=1}^{K-t} 2 \,\, \bE_{\rho_{\beta}}\left[ (e_t-v)\,(e_{t+p}-v)\,\right]\\
	&= K \,\bV_{\rho_{\beta}}(\cM) +    \sum_{t=1}^K \sum_{p=1}^{K-t} \sum_{j\geq 2} 2\, \alpha_{1j}\alpha_{j1}  \,\lambda_j^{p-1} 	\\
	&\leq K \,\bV_{\rho_{\beta}}(\cM) +    \sum_{t=1}^K \sum_{p=1}^{K-t} \sum_{j\geq 2} 2\, |\alpha_{1j}\alpha_{j1}|  \,\lambda_j^{p-1}. \label{eq:58}	\end{align}
	
Note that by Assumption \ref{ass:NM} and 
Lemma \ref{lem:unique}, we have  that $\cE$ also has $\rho_{\beta}$ as unique fixed state and its the spectrum is contained in  $[0,1]$. 
By Fact \ref{fact:QDB}, we know the spectrum of $\cE^\dagger$ equals to the spectrum of $\cE$ and $gap(\cE)=gap(\cE^\dagger)$. Besides, from Lemma \ref{lem:MNM_N} we know $gap(\cE)\geq gap(\cN)$.
Thus for
any $j\geq 2$,  we have $\lambda_j \in [0, 1-gap(\cN)]$.  

Thus for any $j\geq 2$, one can check that

\begin{align}
 \sum_{t=1}^K \sum_{p=1}^{K-t}  2\,\lambda_j^{p-1}
 &= 2\sum_{p=0}^{K-2} (K-p-1)\,\lambda_j^{p} \\
 &\leq 2K\sum_{p=0}^{\infty} \lambda_j^{p} = \frac{2K}{1-\lambda_j} \\
 &\leq \frac{2K}{gap(\cN)}.
\end{align}

Since $\alpha_{11}=0$ by Lemma \ref{lem:alpha}, we have that the spectral covariance weight
 \begin{align}
 	\SC:= \sum_{j} |\alpha_{1j}\alpha_{j1}|  = \sum_{j\geq 2} |\alpha_{1j}\alpha_{j1}|.
 \end{align}

Then Eq.~(\ref{eq:58}) implies

\begin{align}
	\bE_{\rho_{\beta}}\left[ \left(\sum_{t=1}^K e_t-K v\,\right)^2 \right] &\leq K \,\bV_{\rho_{\beta}}(\cM) + \sum_{j\geq 2} |\alpha_{1j}\alpha_{j1}| \,\, \frac{2K}{gap(\cN)}\\
	&=   K \,\bV_{\rho_{\beta}}(\cM)+  \frac{2K}{gap(\cN)} \,\SC.
	\end{align}
\end{proof}

Now finally we are prepared to prove Lemma \ref{lem:aut_gap}, which is to bound the autocorrelation time.

\begin{proof}[of Lemma \ref{lem:aut_gap}]
It suffices to notice that
    \begin{align}
        K \bV_{\rho_{\beta}}\left(\cM \right) \, 2 t_{aut,K} &=  K \bV_{\rho_{\beta}}\left(\cM \right) +  2\sum_{t=1}^K \sum_{p=1}^{K-t} \bCor_{\rho_{\beta}}\left(\cE^{p-1}\right)\\
        & = \bE_{\rho_{\beta}}\left( \sum_{t=1}^K e_t - Kv\right)^2.
    \end{align}
Then  \cref{lem:aut_gap} follows from \cref{lem:time_average_var} 
    and Lemma \ref{lem:j1_to_1j}.
\end{proof}

\section{Implementation of the measurement}\label{sec:measure}

In this section, we describe the implementation of the measurement $\cM$ in Algorithm~\ref{alg:intro_main}. 

\paragraph{Structure of the section.}
In Section~\ref{sec:GQPE}, we review Gaussian-filtered quantum phase estimation (GQPE)~\cite{moussa2019low} and adapt it to our setting. For observables that commute with $H$, Section~\ref{sec:M_com} uses GQPE to construct a measurement that is unbiased and satisfies detailed balance, where Theorem~\ref{thm:time_average} applies directly. For concreteness, in Section~\ref{sec:energy}, we will focus on estimating the average energy of a Gibbs state and provide a detailed comparison of the costs associated with different ways of measuring the energy. Finally, in Section~\ref{sec:WOFT_noncom}, for observables that do not commute with $H$, we combine GQPE with an operator Fourier transform technique to construct a measurement that approximately fixes $\rho_\beta$ and may mitigate disturbances from measuring general observables.

\subsection{Gaussian-filtered quantum phase estimation (GQPE) and resource estimation}\label{sec:GQPE}

In this section, we review Gaussian-filtered quantum phase estimation (GQPE)~\cite{moussa2019low} and adapt it to our setting. We  provide a rigorous resource analysis of GQPE in Appendix \ref{app:GQPE}, which to our knowledge, is not available in prior literature and may be of independent interest. 

\paragraph{Idealized Gaussian filtered quantum phase estimation}

To build intuition, here we describe the idealized setting, assuming that the ancillary quantum system takes continuous values. 
The discretization  will be discussed later.

 The idealized Gaussian Filtered quantum phase estimation w.r.t. observable $\bO$ and parameter $\gamma$ is defined w.r.t. Kraus operator 
\begin{align}
    \bO_{w*}:= \frac{1}{\sqrt{\gamma \sqrt{2\pi}}} \exp\left[-\frac{(w- \bO)^2}{4\gamma^2}\right].\label{eq:GQPE_kraus}
\end{align}
Here the subscript $*$ indicates the idealized case.
In other words, the quantum channel $\cM_*$ is defined 
as
\begin{align}
  \cM_*(\rho):= \int_{-\infty}^{+\infty} d w \, \frac{1}{\gamma \sqrt{2\pi}}  \exp\left[-\frac{(w-\bO)^2}{4\gamma^2}\right]  \, \rho \, \exp\left[-\frac{(w-\bO)^2}{4\gamma^2}\right]. \label{eq:cM_def} 
\end{align}

One can check that $\int_{-\infty}^{+\infty} \bO_{w*}^\dagger \bO_{w*} =I$ thus $\cM_*$ defines a quantum channel. 

The first property of $\mathcal{M}_*$ is that the measurement value with respect to the  observable $O$ can be recovered from the measurement outcome of $\mathcal{M}_*$: 
Note that $\int_{w} w \bO_{w*} \bO_{w*}^\dagger =  \bO$, thus
for any quantum state $\rho$, we have that 
	$$\bE_{\rho}\left(\cM_*\right)=\Tr(\bO\rho).$$

The second property is that $\cM_*$ inherits the commutativity property between $\bO$ and $\rho$. First note that 
	 if $\bO$ commutes with $\rho$,  we have that 
     $\cM_*(\rho)=\rho,$ since $\bO_{w*}$ commutes with $\rho$ and $\bO_{w*}=\bO_{w*}^\dagger$.
	 Additionally, if the reference state $\rho$ commutes with $\bO$, then both 
    $\cM_*$ and $\widehat{\cM}_*$ satisfy $s$-detailed balance with respect to $\rho$.

We remark that GQPE is not a projective measurement: for a general state \(\rho\), \(\cM^2(\rho)\neq \cM(\rho)\). Thus, the two consecutive applications of \(\cM\) in
\[
\cE^2=\cM\circ\cN\circ(\cM\circ\cM)\circ\cN\circ\cM
\]
in Algorithm~1 cannot be replaced by a single application of \(\cM\) in the current proof.

\paragraph{Resource estimation and properties of GQPE} 

As in \cite{moussa2019low}, the idealized Gaussian-filtered quantum phase estimation can be discretized and implemented on qubit-based quantum computers. 
In Appendix~\ref{app:GQPE}, we provide a detailed discretization scheme,  summarized in  Algorithm~\ref{alg:GQPE}, and a rigorous resource analysis based on the Poisson summation formula, which yields a rapidly convergent Riemann sum approximation.

With the parameter choices specified in Appendix~\ref{app:GQPE}, for any observable $O$, we use GQPE to construct a measurement $\mathcal{M}$ that has low implementation cost, enables recovery of the measurement value of $O$, and preserves quantitative commutativity between $O$ and the input state. More specifically, we show that

\begin{theorem}[Gaussian-filtered quantum phase estimation (GQPE)]\label{thm:GQPE}
    Let $O$ be an observable where $\|O\|\leq \kappa = \poly(n)$.     Let  $\epsilon > 0$ be the precision parameter.
     One can implement a quantum channel  $\cM$ where for any quantum state $\rho$: 
    \begin{itemize}
     \item \textbf{Implementation cost:} the measurement channel $\cM$ can be implemented using $m$ ancilla qubits, $t$-time controlled Hamiltonian evolution of the operator $O$, and an additional $\poly(m)$ elementary gates, such that 
     $$m=\cO\left(\log (\kappa) + \log \log (\epsilon^{-1})\right); \quad t =\cO\left(\log(\kappa) + \log (\epsilon^{-1})\right).$$
        \item \textbf{Expectation and Variance:}
        The measurement statistics  can be used to recover the measurement value of $O$. More specifically, 
    \begin{align*}
    	&|\mathbb{E}_{\rho}[\cM] - \Tr[\rho O]| \le \epsilon,\\
    	&|\bV_{\rho}[\cM] - \bV_\rho[O]| = \cO(1) ,\\
    	&| \bCov_{\rho}\left(\cM\right)|\leq \bV_{\rho}[\cM],
    \end{align*}
    where $\bV_\rho[O] := \Tr[O^2 \rho]-\left(\Tr[O\rho]\right)^2.$    
    \item\textbf{Commutativity property:} If further $\|O\|\leq 1$,
    then the quantum channel $\cM$ approximately fixes $\rho$ when $O$ and $\rho$ approximately commute, in the sense that
    \begin{align}
        \left\| \cM[\rho] - \rho\right\|_1 \leq \epsilon^{-1} \polylog (\epsilon^{-1})  \left\| \, \left[ O,\rho \right] \, \right\|_1 + 5\epsilon. \label{eq:com}
    \end{align}
    \end{itemize}
\end{theorem}

It is worth noting that in the implementation cost, the number of ancilla qubits and Hamiltonian evolution time all scale only polylogarithmically with the relevant parameters.


When $O$ and $\rho$ exactly commute, we especically have that
\begin{theorem}[Detailed Balance]\label{thm:gqpe-main-3}
    Let $\rho:=\exp(-\beta H)/Z_{\beta}$ be the Gibbs states w.r.t. Hamiltonian $H$ and inverse temperature $\beta$.
    If $[H,O] = 0$, then the GQPE measurement channel $\mathcal{M}$  and the corresponding $\widehat{\cM}$ satisfy the  $s$-detailed balance w.r.t. $\rho_{\beta}$ for any $0 \le s \le 1$. 
\end{theorem}
\begin{proof}
    Theorem \ref{thm:gqpe-main-3} directly follows from Appendix \ref{app:GQPE}  \cref{lem:Oj} and the definition of $s$-detailed balance.
\end{proof}

\subsection{Detailed balanced measurement for commuting observables with low cost}\label{sec:M_com}

Consider any observable $O$ that commutes with $H$,
such as $O=H^k$ for the integer $k$. Observables that commute with $H$ correspond to the conserved quantities of a system and play a fundamental role in understanding its thermodynamic properties. For example, $\Tr(H \rho_{\beta})$ gives the internal energy. Together with  $\Tr(H^2 \rho_{\beta})$, another observable commuting with $H$, allows one to determine the heat capacity. Moreover, using a telescoping technique~\cite{bravyi2021complexity}, one can estimate the partition function from the ability to estimate $\Tr(H \rho_{\beta})$ for various $H$. The partition function plays a central role in statistical physics, as it determines all other thermodynamic quantities, including free energy, entropy, and pressure.

As a corollary of the properties of  Theorem \ref{thm:GQPE} and Theorem \ref{thm:gqpe-main-3} we have that

\begin{theorem}[Measurement for commuting observables]\label{thm:M_com}
    Let $O$ be an observable that commutes with $H$ and satisfies $\|O\| \leq \kappa =poly(n)$.      Let  $\epsilon > 0$ be the precision parameter.  
    The quantum channel $\cM$ implemented by Algorithm \ref{alg:GQPE} w.r.t. observable $O$ satisfies that
    \begin{itemize}
        \item \textbf{Implementation cost:} 
     $m$ ancilla qubits, $t$-time controlled Hamiltonian evolution of the  operator $O$, an additional $\poly(m)$ elementary gates, for
    $$m=\cO\left(\log (\kappa) + \log \log (\epsilon^{-1})\right); \quad t =\cO\left(\log(\kappa) + \log (\epsilon^{-1})\right).$$
        \item \textbf{Expectation and Variance:}
 The measurement statistics of  $\cM$ satisfy
    \begin{align*}
    	&|\mathbb{E}_{\rho_{\beta}}[\cM] - \Tr[\rho_{\beta} O]| \le \epsilon,\\
    	&|\bV_{\rho_{\beta}}[\cM] - \bV_{\rho_{\beta}}[O]| = \cO(1) ,\\
    	&| \bCov_{\rho_{\beta}}\left(\cM\right)| \leq   \bV_{\rho_{\beta}}\left(\cM\right),
    \end{align*}
    where $\bV_{\rho_{\beta}}[O] := \Tr[O^2 \rho_{\beta}]-\left(\Tr[O\rho_{\beta}]\right)^2.$    
    \item\textbf{Fixed Point and Detailed balance} The measurement $\cM$ fixes the Gibbs states, i.e. 
    $\cM(\rho_{\beta})=\rho_{\beta}.$
    Both $\mathcal{M}$  and the corresponding $\widehat{\cM}$ satisfy the  $s$-detailed balance w.r.t. $\rho_{\beta}$ for any $0 \le s \le 1$. 
    \end{itemize}
\end{theorem}

To achieve an implementation cost that scales logarithmically with the precision, one may instantiate the $t$-time controlled Hamiltonian evolution using Hamiltonian simulation techniques with logarithmic precision dependence, such as linear-combination-of-unitaries methods~\cite{gilyen2019quantum,childs2012hamiltonian} or qubitization~\cite{low2019hamiltonian}. For geometrically local Hamiltonians, the underlying structure allows the Hamiltonian simulation to be parallelized, further reducing the circuit depth~\cite{haah2021quantum}.

Combining Theorem~\ref{thm:M_com} with Theorem~\ref{thm:time_average}, we obtain Corollary~\ref{cor:exp_com}, which applies to any observable $O$ that commutes with $H$. As an example, in Section~\ref{sec:energy} we focus on the case $O = H$ and compare the gate and depth complexity of GQPE with those of other approaches for estimating the same observable.

\begin{corollary}[Thermal expectation estimation for commuting observables]\label{cor:exp_com}
Let $\mathcal{N}$ be a Gibbs sampling algorithm that satisfies $s$-detailed balance with respect to the Gibbs state $\rho_\beta = e^{-\beta H}/Z_\beta$,  has $\rho_\beta$ as its  unique fixed point, and has a spectrum  in $[0,1]$.
For any observable $O$ commuting with $H$ and obeying $\|O\| \le \kappa=\poly(n)$, one can estimate $\Tr(O\rho_{\beta})$ to precision $\epsilon$ with probability at least $1-2\eta$ using $m$ ancilla qubits, $t_{\cN}$-time of $\cN$, and $t_{\cM}$-time of controlled Hamiltonian evolution of $O$, where
\[
m=\cO\!\left(\log\kappa+\log\!\log \epsilon^{-1}\right),\quad
t_{\cN}= t_{mix} + \cO\!\left(\frac{\bV_{\rho_{\beta}}[O]}{\epsilon^{2}\eta\, gap(\cN)}\right),\quad
t_{\cM}=\widetilde{\cO}\!\left(\frac{\bV_{\rho_{\beta}}[O]}{\epsilon^{2}\eta\,gap(\cN)}\right)
\]
where $\widetilde{\cO}(\cdot)$ hides a factor of $\mathrm{polylog}(\epsilon^{-1}) + \mathrm{polylog}(\kappa)$.
The cost $t_{\cM}$ can be reduced by a factor of $gap(\cN)^{-1}$, that is, to $\tilde{\cO}\!\left(\bV_{\rho_{\beta}}[O]/(\epsilon^{2}\eta)\right)$ by skipping measurement, as noted in Remark~\ref{remark:skip_measure}.

\end{corollary}

\subsection{Energy function estimation and cost comparison}\label{sec:energy}

In this section, we illustrate the effectiveness of Theorem~\ref{thm:GQPE} by analyzing a concrete example.  In particular, here
we focus on the case where $O=H$ and set the goal as estimating $\Tr(H\rho_{\beta})$ to precision $\epsilon$. Here for simplicity
we also
assume that $H$ is defined on a 2D lattice to enable a meaningful comparison of circuit depth.

To begin with, we clarify the key quantities that will be compared. Suppose $\cM$ is a quantum channel that implements a measurement of $H$, in the sense that its output expectation value equals  $\Tr(H\rho_{\beta})$. Let $c_{\cM}$ and $d_{\cM}$ denote, respectively, the gate cost and depth cost of performing a single use of $\cM$. 

According to Theorem~\ref{thm:time_average} and Remark~\ref{remark:skip_measure}, Eqs.~(\ref{eq:rem_N})(\ref{eq:rem_M}), the dominant factors of the total cost are the measurement variance $\Var_{\cM}:=\bV_{\rho_{\beta}}(\cM)$, which determines the total cost of the Gibbs-sampling evolution, and the product
$c_{\cM}\times \Var_{\cM}$,
which determines the total gate cost of implementing the measurement.
Similarly, we also consider the product
$d_{\cM}\times \Var_{\cM},$
which determines the total circuit depth incurred by the measurement. 
In what follows, we estimate these quantities for different choices of $\cM$ and show that the GQPE procedure described in Theorem~\ref{thm:GQPE} achieves the smallest cost (up to a $\poly\log(\epsilon^{-1})$ factor) on all those quantities. The comparison is summarized in Table \ref{tab:intro_compare} in the introduction section.   Note that Theorem \ref{thm:GQPE} implies that for GQPE, the factor $\bp$ in Theorem \ref{thm:time_average} satisfies $\bp\leq 1$, thus we ignore this factor.

\paragraph{Notations and Scaling of $\bE_{\rho_{\beta}}(H)$ and  $\bV_{\rho_{\beta}}(H)$} 
Recall that $n$ is the number of qubits in the Hamiltonian $H$. We assume that $H$ is defined on a 2D lattice, so it consists of $\cO(n)$ local terms. The expectation value $\bE_{\rho_{\beta}}(H) := \Tr(H\rho_{\beta})$ represents the internal energy of the system, and the variance 
$\bV_{\rho_{\beta}}(H) := \Tr(H^2\rho_{\beta}) - \big(\Tr(H\rho_{\beta})\big)^2$ 
is related to the heat capacity. For finite $\beta$, both $\bE_{\rho_{\beta}}(H)$ and $\bV_{\rho_{\beta}}(H)$ typically scale linearly with the system size, i.e., $\cO(n)$.

\paragraph{Classical Hamiltonian} For comparison, we also list the cost in the classical setting. First note that when $H$  is a classical Hamiltonian, give a sample $x$ from the Gibbs distribution, where $x$ is a computational basis, one can compute the corresponding energy $\langle x|H|x\rangle$ to high precision (in fact exactly) in time $\cO(n)$.   Since we assume the Hamiltonian is on 2D, the computation can be parallel to depth $\cO(1).$ The variance of $\langle x|H|x\rangle$, when consider the randomness of $x$, is the same as $\bV_{\rho_{\beta}}(H)$ which scales as $\cO(n)$.

\paragraph{High precision (HP) QPE} When $H$ is a quantum Hamiltonian, to perform a high precision energy estimation, i.e. to precision $\epsilon$,  quantum phase estimation algorithm has a gate cost as $\Omega(n\epsilon^{-1})$, and a depth cost as $\Omega(\epsilon^{-1})$ [see survey \cite{dalzell2310quantum} and \cite{nielsen2010quantum}]. The variance of this measurement is of the same order as $\bV_{\rho_{\beta}}(H)$ which scales as $\cO(n)$.

\paragraph{Local Measurement} 

Another simple way to measure $H=\sum_{i=1}^{\kappa} H_i$ is to randomly measure its local terms. Given an input state $\rho_{\beta}$, one can randomly select a local term $H_i$ according to a uniform distribution and perform a projective measurement with respect to $H_i$ on $\rho_\beta$. Denote the measurement outcome as $Y_i$ and set the energy estimate $X_i := \kappa Y_i$. One can check that $E(X_i) = \Tr(H \rho_{\beta})$.  

Since this measurement only measures a single local term at a time, we have $c_{\cM} = \cO(1)$. However, it increases the variance by a factor of $\cO(n)$ compared to $\bV_{\rho_{\beta}}(H)$, as
\begin{align*}
\Var(X_i) &= \kappa\, \Tr\Big(\sum_i H_i^2 \rho_{\beta}\Big) - \Tr^2(H \rho_{\beta}) = \cO(\kappa^2) = \cO(n^2).
\end{align*}

Another drawback of this randomized local measurement is that it is incoherent (destructive) and substantially disturbs the Gibbs state. A potential remedy is to apply a recovery channel after the destructive measurement: using the quantum Markov property, in the geometrically local case, a quasi-local recovery map can drive the perturbed state back to the Gibbs state~\cite{chen2025quantum}, at a cost that grows polynomially with $\epsilon^{-1}$ and $\beta$. Analyzing the correlation between samples in this approach is beyond the scope of this manuscript and is left as an open question.

\paragraph{GQPE} For GQPE, with the  parameters choices in Algorithm \ref{alg:GQPE}, 
according to Theorem \ref{thm:GQPE} we know the cost of implementing  a single $\cM$ is dominated by implementing the controlled Hamiltonian evolution 
\begin{align}
    W =\sum_{\ell=0}^{N-1} e^{-2\pi i \xi_{\ell} H/\kappa} \otimes \ketbra{\ell}{\ell}= e^{\pi i Nh H/\kappa}\sum_{\ell=0}^{N-1} e^{-2\pi i \ell h H/\kappa} \otimes \ketbra{\ell}{\ell},\label{eq:WWW}
\end{align}
where the maximal simulation time of $H$ is $t_{\max} \le 2\pi Nh/\kappa \le \cO(\log(n\epsilon^{-1}))$.
By~\cite{haah2021quantum}, for a simulation time $0 \le t \le t_{\max}$, the time-evolution operator $e^{-itH}$ can be implemented up to precision $\epsilon$ using $\cO(n\poly\log(n \epsilon^{-1}))$ gates and $\cO(\poly\log(n \epsilon^{-1}))$ depth. Then, using a standard technique for controlled Hamiltonian simulation~\cite[Lemma 8]{childs2017quantum}, the controlled Hamiltonian evolution $W$ in Eq.~(\ref{eq:WWW}) can be constructed using $\cO(n\poly\log(n \epsilon^{-1}))$ elementary gates and circuit depth, with an additional $\cO(\log(n))$ ancilla qubits.
Note that the circuit depth of the controlled unitary can be further reduced to $\cO(\poly\log(n \epsilon^{-1}))$ circuit depth by introducing $\cO(n\poly\log(n\epsilon^{-1}))$ copies of classical control bits, similar to techniques used in~\cite{jiang2020optimal}.

\subsection{Mitigating Measurement Disturbance for non-commuting observable via Operator Fourier Transform} \label{sec:WOFT_noncom}

From the previous sections, particularly Theorems~\ref{thm:time_average} and Theorem~\ref{thm:M_com}, we see that for observables commuting with $H$, \textit{effectively} independent samples can be obtained on a timescale shorter than the mixing time. One intuitive reason for this Gibbs sampling reduction  is that the GQPE measurement leaves the Gibbs state invariant, so no additional burn-in period is required after the measurement.

In this section, we consider measuring a local observable $O$ that does not commute with $H$. A direct projective measurement of such an observable would significantly disturb the Gibbs state and introduce bias. Guided by the intuition discussed above, here we investigate an operator Fourier transform technique~\cite{chen2023efficient} to construct a measurement that approximately preserves the Gibbs state. 

We remark that Theorem \ref{thm:time_average} does not directly apply to the measurement constructed in this section, since it does not satisfy \textit{exact} detailed balance, although it approximately preserves the Gibbs state.  We leave addressing this gap to future work.

\paragraph{Weighted Operator Fourier Transform (WOFT)~\cite{chen2023efficient}}
First we explain the weighted operator Fourier transform and the key observations.
Recall that $H$ is a local Hamiltonian and $\rho_{\beta}=\exp(-\beta H)/Z_{\beta}$ is the Gibbs state at inverse temperature $\beta$. 

  Consider the spectrum decomposition of $H$, that is,
$
H:=\sum_{k} E_k \Pi_k. 
$ where $E_k$ is the eigenvalue and $\Pi_k$ is the corresponding spectral projector.
To ease notations, define the set of spectrum as $Spec(H):=\{E_k\}_k$ and define the set of energy difference, namely Bohr frequency, as 
\begin{align}
  \bB:=\{E_i-E_j\,|\, E_i,E_j\in Spec(H)\}.
\end{align}

The Fourier components of $O$ w.r.t. frequency $\nu$ is defined as 
\begin{align}
    O_{\nu}:=\sum_{i,j:\, E_i-E_j=\nu} \Pi_i O \Pi_j.
\end{align}

\begin{definition}[Weighted Operator Fourier Transform (WOFT)~\cite{chen2023efficient}]
	Given an observable $O$ with $\|O\| \leq 1$, a reference Hamiltonian $H$, a function  $f(t)=\frac{\tau}{\sqrt{\pi}}\,e^{-\tau^2 t^2}$ with parameter $\tau$, the weighted operator Fourier transform is defined as 
 \begin{align}
    \widehat{O}(\tau) &:=\int_{-\infty}^{+\infty} e^{iHt} O e^{-iHt}  f(t) dt =\sum_{\nu\in \bB}  \exp\left(-\frac{\nu^2}{4\tau^2}\right) \,  O_{\nu}. \label{eq:OFT} 
 \end{align}
\end{definition}

The key observation is as follows:

\begin{observation} \label{obs:OFT}
For any observable $O$, define $\widehat{O}(\tau)$ as in Eq.~(\ref{eq:OFT}). Then the following holds:
\begin{itemize} 
    \item \textbf{Observable value:}  The WOFT does not change the observable value. That is, for any $\tau$, 
\begin{align*}
	\Tr\left(\rho_{\beta} \widehat{O}(\tau) \right)  =\Tr\left(\rho_{\beta}O_0 \right)= \Tr\left(\rho_{\beta}O \right).
\end{align*}
    \item \textbf{Decay of Commutator norm:}  
    WOFT filters out the components that do not commute with $H$. From Eq.~(\ref{eq:OFT}), as $\tau \rightarrow 0$, $\widehat{O}(\tau)$ converges to $O_0$ that commutes with $H$. Consequently,
\[
\left\| \left[ \widehat{O}(\tau), \rho_{\beta} \right] \right\|_1 \rightarrow 0 \quad \text{as } \tau \rightarrow 0.
\]   
   \item \textbf{Norm bound:} Note that WOFT does not increase the spectrum norm, that is, $$\|\widehat{O}(\tau)\| \leq \int_{-\infty}^{+\infty} \|e^{iHt} O e^{-iHt}\| \,  f(t) dt=1.$$ 
\end{itemize}
\end{observation}

\paragraph{Mitigate measurement disturbance by WOFT}

In this section, we consider a local observable $O$ satisfying $\|O\| \le 1$.
Observation~\ref{obs:OFT} shows that, to estimate $\Tr(O \rho_{\beta})$, one may instead measure the Fourier-transformed operator $\widehat{O}(\tau)$, since it has the same expectation value. 
Moreover, measuring $\widehat{O}(\tau)$ for small $\tau$ is expected to introduce less disturbance to the Gibbs state, since $\|[\widehat{O}(\tau), \rho_{\beta}]\|_1$ becomes small when $\tau$ is small. The measurement for the Fourier-transformed operator $\widehat{O}(\tau)$ can be implemented by GQPE combined with standard Hamiltonian simulation techniques like linear combination of unitaries. More specifically,

\begin{theorem}[Measurement for non-commuting observables]\label{thm:M_noncom}
    Consider any  local observable $O$ such that $\|O\| \leq 1$. Let $\epsilon$ be a precision parameter. Set $\tau$  small enough such that  $\left\| \left[ \widehat{O}(\tau),\rho_{\beta} \right]\right\|_1\leq \epsilon^2$. Then
    the quantum channel $\cM$ implemented by Algorithm \ref{alg:GQPE} w.r.t. discretized version of observable $\widehat{O}(\tau)$ satisfies that
    \begin{itemize}
        \item \textbf{Expectation and Variance:}
 The measurement statistics of  $\cM$ satisfy
    \begin{align*}
    	&|\mathbb{E}_{\rho_{\beta}}[\cM] - \Tr[\rho_{\beta} O]| \le 2\epsilon,\\
    	&|\bV_{\rho_{\beta}}[\cM] - \bV_{\rho_{\beta}}[O]| = \cO(1) ,\\
        &| \bCov_{\rho_{\beta}}(\cM)| \leq  \bV_{\rho_{\beta}}[\cM],
    \end{align*}
    where $\bV_{\rho_{\beta}}[O] := \Tr[O^2 \rho_{\beta}]-\left(\Tr[O\rho_{\beta}]\right)^2.$    
    \item\textbf{Approximate fixed point} The measurement $\cM$ approximately fixed the Gibbs state,
    \begin{align*}
        \left\| \cM[\rho_{\beta}] - \rho_{\beta}\right\|_1 \leq \epsilon\, \polylog (\epsilon^{-1})   + 5\epsilon. 
    \end{align*}
    \item \textbf{Implementation cost:}  
     $m$ ancilla qubits, $t$-time controlled Hamiltonian evolution of $H$,   and   $\polylog(\tau^{-1}\epsilon^{-1}\|H\| )$ additional elementary gates, for 
     \begin{align}
         m= \cO(\log(\tau^{-1}\epsilon^{-1}\|H\| )),\quad  
        t= \cO\left(\tau^{-1} \polylog\left(\epsilon^{-1} \tau^{-1}\right)\right) \label{eq:M_noncom_t}
     \end{align}
    When the commutator norm $\left\|[\widehat{O}(\tau), \rho_{\beta}]\right\|_1$ decays exponentially fast as $\mathcal{O}(\exp(-\mathcal{O}(1/\tau^2)))$, it suffices to choose $\tau^{-1} = \mathcal{O}(\log \epsilon^{-1})$. 
      
If further $O$ and $H$ are geometrically local, the Lieb--Robinson bounds imply that the implementation cost can be reduced to:  $m$ ancilla qubits, and  a $t$-time controlled Hamiltonian evolution of a quasi-local operator instead of the full $t$-time controlled Hamiltonian evolution of $H$. Since the number of controlled bits (i.e. $m$) is logarithmic, such a quasi-local evolution for time $t$ can be implemented using $\polylog(n)+\polylog(\epsilon^{-1})$ elementary gates.
    \end{itemize}
\end{theorem}

The  proof of Theorem \ref{thm:M_noncom} is put into Appendix \ref{app:LCU} and is mainly a corollary of Theorem \ref{thm:GQPE} combined with standard technique of Hamiltonian simulation for the discretized version of $\widehat{O}(\tau)$.

According the Eq.~(\ref{eq:M_noncom_t}), the implementation cost is mainly determined by $\tau^{-1}$, thus determined by the commutator norm decay rate. We note that one scenario in which the commutator norm decays exponentially is that $O$ induces a large energy disturbance, especially in the low-energy eigenspace. For example, if $H$ is the 2D Toric code and $O$ a single-qubit Pauli $X$,  then applying $O$ to an energy eigenstate either leaves its energy unchanged or changes it by $2$ in magnitude. Hence the smallest nonzero energy change induced by $O$ is $2$,
this will implies
$\left\| [\widehat{O}(\tau), \rho_{\beta}] \right\|_1 \leq 2 \exp\left(-\frac{1}{\tau^2}\right).$ 
Another example arises when measuring the zero-temperature Gibbs state (i.e. ground state) of a gapped system. In particular, when $\beta = +\infty$ and the Hamiltonian has a spectral gap $c$ between the ground-state energy and the first excited energy, then for any observable $O$ with $\|O\|\leq 1$, one have that $\left\| [\widehat{O}(\tau), \rho_{\beta}] \right\|_1 \leq2\exp(-\frac{c^2}{4\tau^2})$.
Further details on these two examples are given in Appendix~\ref{app:commutator_decay}.

We remark that although Theorem~\ref{thm:M_noncom} shows that the implemented measurement $\cM$ approximately fixes the Gibbs state---suggesting that an additional burn-in period may not be necessary---we cannot directly apply \cref{lem:aut_gap} and \cref{thm:time_average}, since $\cM$ does not satisfy \textit{exact} detailed balance. We leave addressing this gap to future work.

\section{Exponential mixing-autocorrelation separation in the easy-axis mean-field Heisenberg model}\label{sec:example_quantum}

In this section, for a family of quantum Hamiltonians, we prove an exponential separation between the mixing time and the autocorrelation time of the energy observable. More specifically, recall the easy-axis mean-field Heisenberg Hamiltonian defined in Section \ref{sec:example} Eq.~(\ref{eq:easy_axis_H}): for even $n$, the Hamiltonian is

\begin{align}
    &H_{\alpha}
    :=
    -\frac{J}{n}\left(\bbS^2+(\alpha-1)S_z^2\right),
    \qquad J>0,\quad \alpha>1,\\
    &S_a=\frac{1}{2}\sum_{i=1}^n \sigma_i^a,
    \qquad a\in\{x,y,z\},
    \qquad
    \bbS^2=S_x^2+S_y^2+S_z^2.
\end{align}
where \(\sigma_i^a\) denotes the Pauli operator acting on qubit \(i\). 

Let $\beta$ be the inverse temperature. The parameter region we consider is 
\begin{align}
    \beta J <2<\beta J \alpha.
\end{align}

In Section~\ref{sec:example_quantum_eigenstate}, we describe the eigenstates of the easy-axis mean-field Heisenberg Hamiltonian using the standard total-spin basis. In Section~\ref{sec:example_quantum_davies}, we introduce the corresponding Davies generator.

Then
 in Section~\ref{sec:example_quantum_mix}, we prove that the mixing time of the corresponding Davies generator is exponentially large; see Lemma~\ref{lem:example_mix}.
      In Section~\ref{sec:example_quantum_auto}, we prove that the energy autocorrelation time is upper bounded by a polynomial; see Lemma~\ref{lem:example_auto}.
Combining Lemma~\ref{lem:example_mix} and Lemma~\ref{lem:example_auto}, we conclude that there is an exponential mixing-autocorrelation time separation for the easy-axis mean-field Heisenberg model, thereby proving the Theorem~\ref{thm:easy_axis_separation} in Section \ref{sec:example}.

\subsection{The eigenspace structure}\label{sec:example_quantum_eigenstate}

Define the ladder operator 

\begin{align}
 S_\pm=S_x\pm iS_y=\sum_{i=1}^n\sigma_i^\pm, \qquad \sigma_i^\pm
 =\frac{\sigma_i^x\pm i\sigma_i^y}{2}.
\end{align}

Although the eigenstates of \(H_\alpha\) are generally not product states in the computational basis, they can be labeled by total-spin quantum numbers. More specifically, since \(H_\alpha\) commutes with both \(\mathbf S^2\) and \(S_z\), its eigenstates can be chosen as simultaneous eigenstates of these two operators. We denote them by \(\{\ket{s,m,r}\}\)~\cite{basso2026spectral}, where
\begin{align}
&\bbS^2 \ket{s,m,r} = s(s+1)\ket{s,m,r},\qquad S_z \ket{s,m,r} = m\ket{s,m,r},\\
&S_+\ket{s,m,r}
 =\sqrt{(s-m)(s+m+1)}\,\ket{s,m+1,r},
 \\
&S_-\ket{s,m,r}=\sqrt{(s+m)(s-m+1)}\,\ket{s,m-1,r}.
 \label{eq:spin-ladder}
\end{align}
The labels have ranges
\begin{align}
 s=0,1,\ldots,N, \qquad
 m=-s,-s+1,\ldots,s, \qquad
 r=1,\ldots,d_s, \qquad d_s
 =\frac{2s+1}{N+s+1}\binom{n}{N-s}.
 \label{eq:ds}
\end{align}

Here \(s\) denotes the total spin, \(m\) denotes the magnetization, and
\(r\) distinguishes states with the same \(s\) and \(m\). The ladder
operators \(S_+\) and \(S_-\) increase and decrease \(m\), respectively,
by one. Since $s$ and $m$ do not uniquely specify a state,  an additional index $r=1,\ldots,d_s$ is introduced to distinguish states with the same $s$ and $m$.

As a consequence, we have 
\begin{align}
		H_\alpha  \ket{s,m,r} =  E_{s,m}\ket{s,m,r},\qquad  E_{s,m} = -\frac{J}{n}(s(s+1)+(\alpha-1)m^2).
\end{align}

To ease notations, we define the projection 
\begin{align}
	\Pi_{s,m}:=\sum_{r=1}^{d_s} \ket{s,m,r}\bra{s,m,r},\qquad \Pi_E:=\sum_{(s,m):E_{s,m}=E} \Pi_{s,m}.
\end{align}

\subsection{Davies generator and the corresponding classical Markov Chain}\label{sec:example_quantum_davies}

\paragraph{The Davies generator}
The Gibbs sampler considered here is 
 the Davies generator w.r.t. jump operators chosen to be all single-qubit Pauli matrices. More specifically, let $\operatorname{spec}(H_\alpha)$ be the spectrum of $H$. Define the set of Bohr frequency 
\begin{align}
 \operatorname{Bohr}(H_\alpha):=\{E-E':E,E'\in\operatorname{spec}(H_\alpha)\}.
\end{align}
For every jump operator $\sigma_{i}^a$, its Fourier component with frequency
$\omega\in\operatorname{Bohr}(H_\alpha)$ is defined as
\begin{align}
	\sigma_i^a(\omega):=\sum_{E,E':E-E'=\omega} \Pi_E \sigma_i^a \Pi_{E'}.
\end{align}

The corresponding Davies generator is defined as
\begin{equation}
\label{eq:Davies}
\begin{aligned}
	&\cD(\rho):= \sum_{i=1}^{n}\sum_{a=x,y,z} \sum_{\omega\in Bohr(H_{\alpha})}  \gamma(\omega) \left( \sigma_i^a(\omega) \,\rho\, \sigma_i^a(\omega)^\dagger -\frac{1}{2}\left\{ \sigma_i^a(\omega)^\dagger\sigma_i^a(\omega),\rho \right\} \right),\\
	&\gamma(\omega) := \min\{1,e^{-\beta \omega}\}.
\end{aligned}
\end{equation}

\paragraph{The classical chain} 
By the symmetry of the Hamiltonian and the symmetry of the collection of jump operators,
Lemma~\ref{lem:classical_chain} shows that if the initial state is a convex
combination of the normalized sector states $\Pi_{s,m}/d_s$, then its evolution
is described by a classical Markov chain.  For completeness we put a proof 
 in Appendix \ref{app:reduction_to_classical}. 
 Note that restricting the initial state to be of this form is sufficient for our purpose: it provides a lower bound on the mixing time, while for the autocorrelation analysis the initial state is the Gibbs state, which itself is of this form.

\begin{lemma}[Classical dynamics of the total-spin sectors]\label{lem:classical_chain} 
For the Davies generator $\cD$ defined in Eq.~(\ref{eq:Davies}), there are transition rates
$Q_{(s,m) \rightarrow (s',m')}$ such that
\begin{align}
&\cD
\left(\frac{\Pi_{s,m}}{d_s}\right)
=
\sum_{(s',m')}
Q_{(s,m) \rightarrow (s',m')}
\frac{\Pi_{s',m'}}{d_{s'}}, \qquad\sum_{(s',m')} Q_{(s,m) \rightarrow (s',m')}=0.
\end{align}
In particular, for $(s,m)\neq (s',m')$, the transition rates satisfy
\begin{align}
	&Q_{(s,m)\rightarrow (s',m')}= \frac{\gamma(E_{s',m'}-E_{s,m})}{d_s} \sum_{i=1}^n \sum_{a=x,y,z}tr\left(\Pi_{s',m'}\sigma_i^a \Pi_{s,m}\sigma_{i}^a  \right)\geq 0, \label{eq:Q}\\
&Q_{(s,m)\rightarrow (s',m')}= 0 \quad \text{unless $|m'-m|\leq 1.$ }
\end{align}

\end{lemma} 

 Thus, if the initial state has the form
$\rho=\sum_{s,m} p_{s,m}\frac{\Pi_{s,m}}{d_s},$
then its evolution under \(e^{t\mathcal D}\) is exactly described by a classical continuous-time reversible Markov chain. The chain has state space \(\{(s,m)\}_{s,m}\), generator \(Q\), initial distribution \(\{p_{s,m}\}_{s,m}\), and stationary distribution \(\{\pi(s,m)\}_{s,m}\), where
\begin{align}
  \pi(s,m) = \frac{d_s}{Z_n} \exp(-\beta E_{s,m}),\qquad
  Z_n=\sum_{s,m}d_s e^{-\beta E_{s,m}}.
\end{align}

 Since
$\gamma(\omega)=e^{-\beta\omega}\gamma(-\omega)$, one can check that Eq.~\eqref{eq:Q} implies $Q$ is reversible with invariant distribution $\pi$, that is
\begin{align}
 \pi(s,m)Q_{(s,m)\to(s',m')}
 =\pi(s',m')Q_{(s',m')\to(s,m)}.
 \label{eq:Q-detailed-balance}
\end{align}

\subsection{The mixing time is exponentially large}\label{sec:example_quantum_mix}

For convenience, we set the target accuracy in the definition of the mixing time to $1/4$.
We use $t_{mix}^\cD(1/4)$ as the mixing time of the Davies generator in Eq.~(\ref{eq:Davies}),  and  use $t_{mix}^Q(1/4)$ as the mixing time of the classical continuous-time  Markov chain defined in Lemma \ref{lem:classical_chain} w.r.t. generator $Q$.  

By Lemma~\ref{lem:classical_chain}, the classical chain $Q$ exactly describes the dynamics of $\cD$ for the class of initial states considered there. Since the mixing time of $\cD$ is defined over all initial states, we have
\begin{align}
t_{mix}^{\cD}(1/4)\geq t_{mix}^{Q}(1/4).
\label{eq:D_Q}
\end{align}

We use a conductance bound to show the spectral gap of the classical chain is small, which implies $ t_{mix}^Q$ is large. Let
\begin{align}
 A_-=\{(s,m):m<0\},
 \qquad
 A_+=\{(s,m):m>0\}.
\end{align}

We first give a bound to estimate weights for different sectors.
This proof is an elementary binomial estimate and is provided in Appendix~\ref{app:sector_weights}.

\begin{lemma}[Weight of the phases and the boundary]
\label{lem:sector_weights}
There are constants \(c,C,\delta>0\) and a value \(s_*=\Theta(n)\) such that
\begin{align}
 Z_n\pi(s_*,s_*)=Z_n\pi(s_*,-s_*)
 &\ge c n^{-2}2^ne^{\delta n},  \label{eq:phaseweight}\\
 Z_n\pi(m=-1),\ Z_n\pi(m=0)
 &\le Cn2^n. \label{eq:boundaryweight}
\end{align}
\end{lemma}

\begin{lemma}[Exponential lower bound on mixing]\label{lem:example_mix}
 There exists constant $C,\delta$ such that the mixing time of the classical  chain and the Davies generator satisfies 
\begin{align}
 t_{mix}^\cD(1/4) \geq  t_{mix}^Q(1/4) \geq  C n^{-4} e^{\delta n}.
\end{align}
\end{lemma}
\begin{proof}
Lemma~\ref{lem:sector_weights} immediately gives that there exists constant $C$, such that
\begin{align}
 \pi(m=-1),\ \pi(m=0)\le Cn^3e^{-\delta n}.  \label{eq:boundaryprob}
\end{align}
By spin-flip symmetry, we also have for all sufficiently large \(n\),
\begin{align}
 \pi(A_-)=\pi(A_+)=\frac{1-\pi(m=0)}{2}\ge\frac13\label{eq:phaseprob}
\end{align}
Note that
the total transition rate out of any state $(s,m)$ is at most $3n$, 
\begin{align}
\sum_{(s',m')\ne(s,m)} Q_{(s,m)\rightarrow (s',m')}
\leq  \frac{1}{d_s} \sum_{i=1}^n \sum_{a=x,y,z}
\operatorname{tr}\left(\sigma_i^a \Pi_{s,m}\sigma_{i}^a  \right)
= 3n.
\end{align}
 Besides, by Lemma \ref{lem:classical_chain}, we have $Q_{(s,m)\rightarrow (s',m')}= 0$ unless $|m'-m|\leq 1.$  Thus any transition from \(A_-\) to its complement
must start at \(m=-1\). 
Thus the flow across the cut $A_{-}$ and its complement satisfies 
\begin{align}
 \sum_{(s,m)\in A_-,\,(s',m')\notin A_-}\pi(s,m)Q_{(s,m)\to (s',m')}
 \le\pi(m=-1)
 3n
 \le Cn^4e^{-\delta n}.   
 \label{eq:cutflow}
\end{align}

Since \(1/3\le\pi(A_-)\le1/2\), the conductance of this cut is
\begin{align}
 \Phi(A_-)
 :=\frac{\sum_{x\in A_-,\,y\notin A_-}\pi(x)Q_{x\to y}}
 {\pi(A_-)}
 \le 3Cn^4e^{-\delta n}.  \label{eq:conductance}
\end{align}
Let $\operatorname{gap}(Q)$ be the spectral of the classical
continuous-time chain w.r.t. generator $Q$. The Cheeger's inequality [see \cite{levin2026markov}]  gives
\begin{align}
 \operatorname{gap}(Q) \le2\Phi(A_-) \leq 6Cn^4e^{-\delta n}.\label{eq:smallgap}
\end{align}
Then we conclude the theorem by the standard relaxation-time lower bound
\(t_{{mix}}^Q(1/4)\ge(\log2)/\operatorname{gap}(Q)\) and recall Eq.~(\ref{eq:D_Q}).
\end{proof}

\subsection{The energy autocorrelation time is polynomial bound}\label{sec:example_quantum_auto}

\paragraph{The folded chain}

Note that the global spin-flip operator
\(U_{\mathrm{flip}}:=\sigma_1^x\sigma_2^x\cdots\sigma_n^x\)
commutes with \(H_\alpha\), maps \(\Pi_{s,m}\) to
\(\Pi_{s,-m}\), and preserves the collection of Pauli jump operators.
Consequently,
\begin{align}
 Q_{(s,m)\to(s',m')}
 =Q_{(s,-m)\to(s',-m')}.
 \label{eq:flipsym}
\end{align}
For diagonal entries of $Q$, the same identity follows from the zero-row-sum convention.

The energy observable $H_\alpha$ is insensitive to the sign of $m$; indeed,
$E_{s,m}=E_{s,-m}$. To estimate the energy autocorrelation time, it is
convenient to define a folded chain by identifying the states $(s,m)$ and
$(s,-m)$.
  View $u=|m|$, define the state space to be $\{(s,u)\}$ for $0\leq u\leq s\leq N$, 
 define the invariant distribution and the transition rate to be
\begin{align}
 \pi_{{fold}}(s,u)
 &=\begin{cases}
 \pi(s,0),&u=0,\\
 2\pi(s,u),&u>0,
 \end{cases}  \label{eq:piven}\\
 Q^{{fold}}_{(s,u)\to(s',u')}
&=\sum_{m':\,|m'|=u'}Q_{(s,u)\to(s',m')},
 \qquad (s,u)\ne(s',u').\label{eq:Qfold}\\
  Q^{{fold}}_{(s,u)\to(s,u)} &= - \sum_{(s',u')\neq (s,u)}   Q^{{fold}}_{(s,u)\to(s',u')} \label{eq:Qflold_ss}
\end{align}
Note that Eq.~\eqref{eq:Qfold} does not contain an additional sum
\(\sum_{m:\,|m|=u}\) over the source states, since folding the states
\((s,m)\) and \((s,-m)\) doubles both the source probability and the
corresponding probability flow, but not the transition rate. 

By Eqs.~\eqref{eq:Q-detailed-balance} and \eqref{eq:flipsym},
one can directly check that \(Q^{fold}\) is reversible with invariant
distribution \(\pi_{fold}\).

Since the folded chain does not distinguish between positive and negative magnetization, the small cut used to prove slow mixing of $Q$ disappears after folding. In fact, we can show that $Q^{{fold}}$ has a polynomially large spectral gap. The proof is a standard canonical-path proof and we put it in
Appendix~\ref{app:canonical}.

\begin{lemma}[Polynomial gap of the folded chain]
\label{lem:gap_fold}
There is a constant \(c>0\) such that
\begin{align}
 \operatorname{gap}(Q^{{fold}})\ge cn^{-5}.\label{eq:evengap}
\end{align}
\end{lemma}

\paragraph{Energy autocorrelation time}

Here we assume that the energy observable $H_\alpha$ is measured by a
perfect projective measurement. This suffices because
Appendix~\ref{sec:proj_GQPE} shows that, for the model considered here,
a polynomial upper bound on the energy autocorrelation time with respect to
projective measurement implies a polynomial upper bound on the energy
autocorrelation time with respect to GQPE.

Let $t_{aut,K}^{\cD}$ be the autocorrelation time w.r.t. the Gibbs sampler $e^{t_0\cD}$, energy observable $H_{\alpha}$, and Gibbs state $\rho_{\beta H_\alpha}$.
Similarly, let 
 $t_{aut,K}^{fold}$ denote the autocorrelation time w.r.t. the Gibbs sampler $e^{t_0 Q^{fold}}$, observable value $E_{s,u}$ at the state $(s,u)$, and Gibbs distribution $\pi_{fold}$. For simplicity we set $t_0=1$.

\begin{lemma}[Polynomial energy autocorrelation time]\label{lem:example_auto}
Since the Gibbs state satisfies
 $\rho_{\beta H_{\alpha}}
 =\sum_{s,m}\pi_{s,m} \Pi_{s,m}/{d_s},$
Lemma~\ref{lem:classical_chain} shows that its evolution  under $\cD$ is described by
the classical chain \(Q\). Moreover, since $E_{s,m}=E_{s,-m}$, identifying the states $(s,m)$ and $(s,-m)$ does not change the energy observed along the trajectory.  Therefore,
\begin{align}
 t_{aut,K}^{\cD}=t_{aut,K}^{fold}.
\end{align}
Thus by Lemma \ref{lem:aut_gap} and Lemma~\ref{lem:gap_fold} we have
\begin{align}
 t_{aut,K}^{\cD}
 =t_{aut,K}^{fold}
 \leq \frac{1}{\operatorname{gap}(Q^{fold})}+\frac{1}{2} =\mathcal{O}(n^5).
\end{align}
\end{lemma}

\section*{Acknowledgement} 
This work is partially supported by the Simons Quantum Postdoctoral Fellowship (J.J., J.L.), by a Simons Investigator Award in Mathematics through Grant No. 825053 (J.L., L.L.), and by the Challenge Institute for Quantum Computation (CIQC) funded by National Science Foundation (NSF) through grant number OMA-2016245 (L.L.). We thank Joao Basso and Sandy Irani for helpful discussions.

\paragraph{AI disclosure.}
The main body of this work was completed before the widespread use of AI tools. AI assistance was limited to the new easy-axis mean-field Heisenberg example, which was developed through discussions with GPT-5.6 Sol Ultra guided by the intuition in Figure~\ref{fig:example}. The authors independently verified and rewrote the proof and take full responsibility for its correctness.

\appendix

\section{Proof of Theorem \ref{thm:M_noncom}}
\label{app:LCU}

In this section, we give the proof for Theorem \ref{thm:M_noncom}.

\paragraph{Discretization of $\widehat{O}(\tau)$}
First, we specify the discretization of $\widehat{O}(\tau)$. Set  integers
\begin{align} T=\cO(\tau^{-1}\log(\epsilon^{-1}\tau^{-1})), \quad L=\cO(\epsilon^{-1}\|H\| \tau^{-1}).
\end{align}
 We truncate the infinite integral in $\widehat{O}(\tau)$ to $[-T,T]$ and discretize the truncated integral with step size $L^{-1}$. More specifically, define the quadrature points as $t_j=-T + jL^{-1}$ for $j=0$ to $2TL-1$. The discretized version of $\widehat{O}(\tau)$ is
\begin{align}
    \widehat{O}_{\mathrm{dis}}(\tau)
    := \sum_{j=0}^{2TL-1} \frac{1}{L}\,
    e^{iHt_j} O e^{-iHt_j} f(t_j).
\end{align}
Using standard analyses of integral discretization, one can verify that
\begin{align}
   \left\| \widehat{O}(\tau)-\widehat{O}_{\mathrm{dis}}(\tau)\right\| \le \epsilon. \label{eq:OOt}
\end{align}

\paragraph{Proof of Theorem \ref{thm:M_noncom}}
The measurement $\mathcal{M}$ is implemented by Algorithm~\ref{alg:GQPE} with respect to the observable $\widehat{O}_{dis}(\tau)$. First, note that the expectation and variance properties in Theorem~\ref{thm:M_noncom} follow as a corollary of Theorem~\ref{thm:GQPE}. In particular, since $\|\widehat{O}_{dis}(\tau) - \widehat{O}(\tau)\| \leq \epsilon$ and $\|\widehat{O}(\tau)\| \leq 1$ by Observation~\ref{obs:OFT}, Theorem~\ref{thm:GQPE} guarantees that the expectation and variance bounds hold. Furthermore, the approximate fixed-point property in Theorem~\ref{thm:M_noncom} directly follows from the commutativity property established in Theorem~\ref{thm:GQPE}.

From Theorem \ref{thm:GQPE},  the implementation cost of $\cM$ is dominated by controlled Hamiltonian simulation  for  $\widehat{O}_{dis}(\tau)$ for time
$t=\cO(\log \epsilon^{-1})$, which is dominated by  Hamiltonian simulation  for  $\widehat{O}_{dis}(\tau)$. 
In the following, we analyze the simulating   $e^{i t\widehat{O}_{dis}(\tau)}$ by standard techniques for 
$$t=\cO(\log \epsilon^{-1}).$$

\paragraph{Geometrically local case} We begin with  the special case when both $H,O$ are geometrically local on a $D$-dimensional lattice.  
  In addition assume $\tau^{-1} = \cO(\log \epsilon^{-1})$ thus  $T=\polylog(\epsilon^{-1})$.

  Let $\Omega$ denote a set of qubits on the lattice whose support is slightly larger than that of $O$. Specifically, the distance between $supp(O)$ and any qubit
outside $\Omega$ is
$\mathcal{O}(T + \log \epsilon^{-1} + \log \lVert H \rVert)$.
  Since $|t_j| \le T$ for all $j$, by the Lieb--Robinson bound (Lemma~5 in \cite{haah2021quantum}), for any $j$ one can approximate 
$e^{i H t_j} O e^{-i H t_j}$ to precision $\epsilon/(2 T L)$ by an operator supported on $\Omega$. 
Consequently, $\widehat{O}_{dis}(\tau)$ can be approximated
to precision $\epsilon$ by an operator supported on $\Omega$,
whose locality is
$\cO(T^D + \log^D(\epsilon^{-1}) + \log^D(\|H\|))$,
where $D$ is the lattice dimension.
Thus, one can simulate $e^{i t \widehat{O}_{dis}(\tau)}$ for $t = \mathcal{O}(\log \epsilon^{-1})$ with gate complexity $\mathcal{O}\left(\polylog(\epsilon^{-1}) \times \polylog(\lVert H \rVert)\right)$ using Hamiltonian simulation methods that achieve logarithmic dependence on the precision~\cite{gilyen2019quantum,childs2012hamiltonian,low2019hamiltonian}.

\paragraph{General case} For general case where $H,O$ might not be geometrically local and  $\tau^{-1}$ may be large. Recall that 
\begin{align}
    T=\cO(\tau^{-1}\log(\epsilon^{-1}\tau^{-1})), \quad L=\cO(\epsilon^{-1}\|H\| \tau^{-1}).
\end{align}
A standard approach to simulating $e^{i t \widehat{O}_{dis}(\tau)}$ is to use the linear combination of unitaries (LCU) technique. General references on LCU can be found in \cite{gilyen2019quantum,childs2017lecture}.

Consider decomposing $O$ into a sum of Paulis, that is,
 $$O=\sum_{k=1}^{\kappa_O} u_k U_k,$$ where each $U_k$ is a local Pauli operator (unitary and Hermitian). Since we assume that $O$ is local and $\|O\|=\cO(1)$, we have $\kappa_O=\cO(1)$ and $|u_k|=\cO(1)$.  Absorbing the sign of $u_k$ into the Pauli, we may write equivalently $O=\sum_{k=1}^{\kappa_O} a_k P_k$ with $a_k:=|u_k|\ge 0$ and $P_k:=\mathrm{sgn}(u_k)U_k$.
 Then 
\begin{align}
    \widehat{O}_{\mathrm{dis}}(\tau)
    &:=\sum_{k=1}^{\kappa_O}   \sum_{j=0}^{2TL-1}a_k f(t_j)  \frac{1}{L} 
    e^{iHt_j}    P_k e^{-iHt_j}.\\
    t_j & = -T + jL^{-1}.
\end{align}

Set $s_1=\cO(1)$ to be an integer such that $2^{s_1}\geq \kappa_O$. 
For simplicity,  we choose $T,L$ such that $2TL=2^{s_2}$ for some integer 
$$s_2=\cO(\log(\tau^{-1}\epsilon^{-1}\|H\|))$$
 
Define $R$ to be a $(s_1+s_2)$-qubit unitary that can prepare the following state,
\begin{align}
	R\ket{0^{s_1}}_1\ket{0^{s_2}}_2 :=  \left(\frac{1}{\sqrt{C_O}}\sum_{k=1}^{\kappa_O} \sqrt{a_k} \ket{k}_1 \right)\otimes\left(\frac{1}{\sqrt{C_f}} \sum_{j=0}^{2TL-1} \frac{1}{\sqrt{L}} \sqrt{f(t_j)} \ket{j}_2\right).
\end{align}
where the subscripts $\ket{\cdot}_1$ and $\ket{\cdot}_2$ refer to different registers.
Note that the normalization factor $C_O:=\sum_{k=1}^{\kappa_O} a_k=\cO(1)$ since $\kappa_O=\cO(1)$ and $a_k=\cO(1)$.
 Besides, by the choice of parameters we know that 
$C_f$ satisfies
$$|C_f -1 |:= \left|\sum_{j=0}^{2TL-1} \frac{1}{L}f(t_j)-1\right|\leq \epsilon.$$ 
Thus the total normalization factor satisfies
\begin{align}
\alpha:= C_O C_f =\cO(1).	
\end{align}

Define an $(s_1+s_2+n)$-qubit  unitary (which is called the select oracle) 
\begin{align}
	U:= \sum_{k=1}^{\kappa_O}\sum_{j=0}^{2TL-1} \ket{k}_1\bra{k}_1\otimes \ket{j}_2\bra{j}_2\otimes e^{iHt_j} P_k e^{-iHt_j} .  
\end{align}
Then the $(s_1+s_2+n)$-qubit unitary $R^{\dagger} \,U R$ is a $(\alpha,s_1+s_2,0)$-block encoding of $\widehat{O}_{dis}(\tau)$, in the sense that
\begin{align}
	\left\|\widehat{O}_{dis}(\tau) - \alpha \left(\left\langle 0^{s_1+s_2}|_{12}\otimes I \left|\, R^{\dagger}\,UR\,\right| 0^{s_1+s_2}\right\rangle_{12} \otimes I \right)\right\|_2 =0.
\end{align}

Recall that $\alpha =\cO(1)$.  
By \cite[Definition~43 and Corollary~60]{gilyen2019quantum}, one can implement the Hamiltonian evolution $e^{i t\widehat{O}_{dis}(\tau)}$ for $t=\cO(\log \epsilon^{-1})$ to precision $\epsilon$ using 
\begin{itemize}
	\item  $(s_1+s_2)$ ancilla qubits, where $s_1=\cO(1)$ and  $s_2=\cO(\log(\tau^{-1}\epsilon^{-1}\|H\|))$,
	\item $\cO(t + \log \epsilon^{-1})$ applications of the block-encoding $R^{\dagger} U R$ (and its inverse). Recall in our case $t=\cO(\log \epsilon^{-1})$.
\end{itemize}

 where $R$ can be implemented with 
 $2^{s_1}\times \poly(s_2)$ elementary gates using Gaussian state preparation algorithms such as \cite{kitaev2008wavefunction}, and since $|t_j|\leq T$, we have that $U$ corresponds to performing controlled Hamiltonian evolution of $H$ for time
 $$\cO(\tau^{-1} \log(\epsilon^{-1}\tau^{-1})).$$

\section{Quantum detailed balance and its property}\label{app:DB}

 In this section, we prove some basic properties of quantum detailed balance.
 We use $\bM$ to denote the linear space formed by all $\bC^{2^n\times 2^n}$ matrices.  Let $\rho_{\beta}$ be the Gibbs state. For any $s\in[0,1]$,
define the linear map $\cF_s:\bM\rightarrow \bM$ to be
\begin{align}
\cF_{s}(A) = \rho_{\beta}^{1-s} A\rho_{\beta}^{s},\,\,  \forall A\in \bM.
\end{align}
Note that $\rho_{\beta}$ is full-rank thus $\cF_s$ is invertible.

For two matrices 
$A,B\in \bM$, the \textit{Hilbert-Schmidt inner product} (HS) is defined as
\begin{align}
\langle A,B\rangle:= \Tr(A^\dagger B).
\end{align}
The weighted inner product is defined as
\begin{align}
    \langle A,B\rangle_s := \langle A, \cF_s(B)\rangle =\Tr(A^\dagger  \rho_{\beta}^{1-s} B \rho_{\beta}^{s}).
\end{align}
One can verify that $\langle ,\rangle_s$  defines a valid inner product. In particular,  $\langle ,\rangle_{\frac{1}{2}}$ is known as the KMS  inner product.
One can check that
\begin{lemma}
Since $\rho_{\beta}$ is a Hermitian operator, we have that
    $\cF_s$ is  Hermitian  w.r.t. HS. That is $\langle A, \cF_s(B)\rangle =\langle \cF_s(A),B\rangle.$
\end{lemma}

Recall that given a linear map $\cT$, the $\cT^\dagger$ is its dual map w.r.t. HS.

\begin{definition}[Quantum detailed balance]
We say that a 
     linear map $\cT:\bM\rightarrow \bM$ satisfying the $s$-quantum detailed balance w.r.t. $\rho_{\beta}$,  if $\cT^\dagger$ is Hermitian w.r.t. the weighted inner product  $\langle , \rangle_s,$ that is 
     \begin{align}
     	\langle A, \cT^\dagger(B)\rangle_s = \langle \cT^\dagger(A), B\rangle_s. \label{eq:DB}
     \end{align}
     which is equivalent to \begin{align}
         \cT \circ \cF_{s}  = \cF_{s} \circ \cT^\dagger. \label{eq:fixedp}
    \end{align}
\end{definition}
Note that $\cT \circ \cF_{s}  = \cF_{s} \circ \cT^\dagger$ implies that      
    \begin{align}
          \cF_{s}^{-1/2} \circ \cT \circ  \cF_{s}^{1/2} &=    \cF_{s}^{1/2} \circ \cT^\dagger \circ  \cF_{s}^{-1/2},\label{eq:QDB2}\\
         \cF_{s}^{-1} \circ \cT &=  \cT^\dagger \circ \cF_{s}^{-1},   \label{eq:QDB}
     \end{align}
where Eq.~(\ref{eq:QDB}) implies that $\cT$ is Hermitian w.r.t. another weighted inner product $\langle ,\rangle_{*s}$,
\begin{align}
   & \langle A, \cT(B)\rangle_{*s} = \langle \cT(A), B\rangle_{*s} \label{eq:DB_inv}\\
\text{for} &\langle A , B\rangle_{*s}:= \Tr(A^\dagger \rho_{\beta}^{s-1} B \rho_{\beta}^{-s}).
\end{align}

\begin{lemma}
[Properties of quantum detailed balance] \label{lem:QDB} Consider a linear map  $\cT:\bM\rightarrow \bM$ satisfying the $s$-quantum detailed balance w.r.t. $\rho_{\beta}$. Then
\begin{enumerate}
    \item[(1)] If in addition $\cT^\dagger$ is unital, i.e. $\cT^\dagger(I)=I$, then $\rho_{\beta}$ is a fixed state of $\cT$, i.e. $\cT(\rho_{\beta})=\rho_{\beta}$. Note that when $\cT$ is a quantum channel, $\cT^\dagger$ is unital. 
    \item[(2)]  $\cF_{s}^{-1/2} \circ \cT \circ  \cF_{s}^{1/2}$ is Hermitian w.r.t. HS. 
    \item[(3)] The maps $\cF_{s}^{-1/2} \circ \cT \circ  \cF_{s}^{1/2}$,\, $\cF_{s}^{1/2} \circ \cT^\dagger \circ  \cF_{s}^{-1/2}$,\, $\cT$ and $\cT^\dagger$  are  diagonalizable. Besides, all the four maps have the same real spectrum.
    \item[(4)] If $\cT$ is a quantum channel, then its spectrum lies in $[-1,1]$. 
    \item[(5)]  If $\cT=e^{t \cL}$ for a Lindbladian $\cL$ and $t>0$, and moreover $e^{t\cL}$ satisfies the $s$-quantum detailed balance condition w.r.t. $\rho_{\beta}$, then its spectrum lies in $[0,1]$.
    \item[(6)] When $\cT$ is a quantum channel and $\rho_{\beta}$ is its unique fixed state,  its spectral gap can be characterized by
    \begin{align}
        gap(\cT) = gap(\cT^\dagger) = 1- \max_{X\in \bM, \langle X,I\rangle_s=0} \frac{\langle X,\cT^\dagger(X)\rangle_s}{\langle X, X\rangle_s}
    \end{align}
\end{enumerate}
\end{lemma}
\begin{proof}
    (1) comes from applying both sides of Eq.~(\ref{eq:fixedp}) to $I$. (2) comes from the fact that $\cF_s$ is Hermitian w.r.t. HS and Eq.~(\ref{eq:QDB2}).
    
    For (3), from (2) we know that   $\cF_{s}^{-1/2} \circ \cT \circ  \cF_{s}^{1/2}$ is Hermitian w.r.t. HS, thus it is diagonalizable w.r.t.  basis $\{Y_1,...,\}\subseteq \bM$, and has a real spectrum, 
    \begin{align}
        &\cF_{s}^{-1/2} \circ \cT \circ  \cF_{s}^{1/2} [Y_i] = \lambda_i Y_i, \text{ where } \lambda_i \in \bR,
    \end{align}
    which implies that $\cF_{s}^{1/2}Y_i$ is the eigenvector of $\cT$ w.r.t. eigenvalue $\lambda_i$. Thus $\cT$ is also diagonalizable and has the same spectrum.  Then it suffices to recall Eq.~(\ref{eq:QDB2}) which says that  $$\cF_{s}^{-1/2} \circ \cT \circ  \cF_{s}^{1/2} =    \cF_{s}^{1/2} \circ \cT^\dagger \circ  \cF_{s}^{-1/2}.$$

    For (4), let $Y\in \bM$ be an eigenvector of $\cT$ with eigenvalue $\lambda$. Since quantum channel does not increase the trace-norm, $ |\lambda|\, \|Y\|_1 =\|\cT(Y)\|_1 \leq \|Y\|_1$, which implies $\lambda\in [-1,1]$. 

For (5), let $\mu$ be an eigenvalue of $\cL$. By the semigroup property, $e^{t\mu}$ is an eigenvalue of $e^{t\cL}$ for all $t\ge 0$. By the additional assumption, each $e^{t\cL}$ satisfies $s$-detailed balance, hence by (3) its spectrum is real, and by (4) it lies in $[-1,1]$. Therefore $e^{t\mu}\in[-1,1]$ for all $t\ge 0$, which forces $\mu\le 0$. Hence the spectrum of $\cT=e^{t\cL}$ is contained in $[0,1]$. (Here we use the fact that $e^{t\cL}$ is a quantum channel for all $t\ge 0$; see, e.g., \cite[Theorem 7.1]{wolf2012quantum}.)

 For (6), by (3) we know that $\cT$ and $\cT^\dagger$ have the same spectrum thus (i) They have the same spectral gap, (ii) $\cT^\dagger$ should also have unique fixed state, which is $I$. Then (6) comes from  the fact that $\cT^\dagger$ is Hermitian w.r.t. $\langle,\rangle_s$, its maximum eigenvalue is $1$ w.r.t. eigenstate $I$, and the definition of spectral gap.
\end{proof}

\begin{lemma}
\label{lem:strict_contract}
	Suppose $\cT$ is a quantum channel that satisfies $s$-detailed balance w.r.t. $\rho_{\beta}$. Then for any operator $A$ we have that 
    \begin{align}
	   \langle A,A \rangle_{*s}  \geq 	\langle \cT(A),\cT(A) \rangle_{*s}.  \label{eq:994}
	\end{align} If further $\cT$ has $\rho_{\beta}$ as its unique fixed state, the spectrum of $\cT$ lies in $[0,1]$,
    and the operator $A$ is non-trivial and  traceless, i.e. $A\neq 0$ and $\Tr(A)=0$, then we have that $\cT$ is strictly contractive, i.e. the $\geq$ in Eq.~(\ref{eq:994}) becomes $>$.
\end{lemma}

\begin{proof}
From Eq.~(\ref{eq:DB_inv}) we know that $\cT$ is Hermitian w.r.t. $\langle,\rangle_{*s}$, thus there exists a set of eigenbasis $\{Z_i\}_i\subseteq \bM$ which is orthonormal w.r.t. $\langle,\rangle_{*s}$, and satisfies $\cT(Z_i)=\lambda_i Z_i$. From Lemma \ref{lem:QDB} (4) we know that  $\lambda_i \in [-1,1]$. Consider decomposing $A$ onto the basis  $\{Z_i\}_i$ as $A = \sum_i \alpha_i Z_i$. 
  We have 
 \begin{align}
	   \langle A,A \rangle_{*s} =\sum_i |\alpha_i|^2   \geq \sum_i \lambda_i^2 |\alpha_i|^2 =	\langle \cT(A),\cT(A) \rangle_{*s}.  
	\end{align}
If further $\cT$ has $\rho_{\beta}$ as its unique fixed state and has a spectrum in $[0,1]$, we have that $1=\lambda_1> \lambda_2...\geq 0$ and $Z_1=\rho_{\beta}$.
For the case of traceless $A$, it suffices to notice that  $\Tr(A)=0$ implies 
$$\alpha_1=\langle A, Z_1\rangle_{*s} =\langle A, \rho_{\beta}\rangle_{*s} =0,$$ Thus  $\langle A,A \rangle_{*s}  > 	\langle \cT(A),\cT(A) \rangle_{*s},$ since $\forall i\geq 2, |\lambda_i|<1$ and $A\neq 0$.
\end{proof}

Finally we prove Lemma \ref{lem:unique}.

\begin{proof}[of Lemma \ref{lem:unique}]
One can verify that $\cM \cN \cM$  satisfies $s$-detailed 
balance by directly applying the definition of detailed balance.
 We prove the fixed point is unique by contradiction.
For any operator $A$, denote $\|A\|_{*s}:=\sqrt{\langle A,A\rangle_{*s}}$. Note that  $\cM\cN\cM$ satisfies $s$-detailed balance thus fixes $\rho_{\beta}$.
With contradiction assume that $\cM\cN\cM$ has two different fixed states $\rho_1$ and $\rho_2$, thus 
\begin{align}
  0<   \|\rho_1- \rho_2\|_{*s} &=  \| \cM\cN\cM\left(\rho_1- \rho_2\right)\|_{*s} \label{eq:MNM_unique_1}.
\end{align}
If $\cM(\rho_1-\rho_2)=0$, then the above equation implies $ 0<   \|\rho_1- \rho_2\|_{*s} =0$  which leads to a contradiction. Otherwise assume  $\cM(\rho_1-\rho_2)\neq 0$.
Since $\cM$ is a quantum channel thus trace-preserving, we have $\cM(\rho_1-\rho_2)$ is traceless. 
Then use Lemma \ref{lem:strict_contract} for  $\cM$, $\cN$, $\cM$ sequentially we have that
\begin{align}
    \| \cM\cN\cM\left(\rho_1- \rho_2\right)\|_{*s} 
    \leq  \| \cN\cM\left(\rho_1- \rho_2\right)\|_{*s}
    < \|\cM\left(\rho_1- \rho_2\right)\|_{*s}
    \leq  \|\rho_1- \rho_2\|_{*s} \label{eq:MNM_unique_2}.
\end{align}
It suffices to notice that Eqs.~(\ref{eq:MNM_unique_1})(\ref{eq:MNM_unique_2}) lead to a contradiction $\|\rho_1- \rho_2\|_{*s} <  \|\rho_1- \rho_2\|_{*s}$.
Thus we conclude that $\cM\cN\cM$ has a unique fixed state $\rho_{\beta}$. 

Since $\cM\cN\cM$ satisfies $s$-detailed balance and is a quantum channel, from Lemma \ref{lem:QDB} (4) we know that its spectrum lies in $[-1,1]$. To show that its spectrum can be further restricted to be in $[0,1]$ it suffices to prove all its eigenvalue are non-negative.

By Fact \ref{fact:QDB} we know that $\cM\cN\cM$ has the same spectrum as $\cM^\dagger \cN^\dagger \cM^\dagger$.
Let $X$ be an eigenvector of $\cM^\dagger\cN^\dagger\cM^\dagger$ with eigenvalue $\lambda$.
By assumption the spectrum of $\cN$ lies in $[0,1]$, from Lemma \ref{lem:QDB} (3) we know the spectrum of $\cN^\dagger$ also lies in $[0,1]$. Thus 
\begin{align}
\lambda=    \frac{\langle X,\cM^\dagger \cN^\dagger \cM^\dagger\left(X\right) \rangle_s}{\langle X,X\rangle_s} = \frac{\langle \cM^\dagger(X), \cN^\dagger \cM^\dagger\left(X\right) \rangle_s}{\langle X,X\rangle_s} \geq 0,
\end{align}

\end{proof}

\begin{lemma}\label{lem:mix_rel}
	Let $\cN$ be a quantum channel satisfying KMS detailed balance with respect to $\rho_{\beta}$, has $\rho_{\beta}$ as its unique fixed state and has spectrum contained in $[0,1]$. Then we have that
	\begin{align}
		 \left(\frac{1}{gap(\cN)}-1\right) \log\left(\frac{1}{\epsilon}\right) \leq t_{mix}(\epsilon) \leq \frac{1}{gap(\cN)} \log\left(\frac{1}{\epsilon \sigma_{min}}\right).
	\end{align}
    where $\sigma_{min}$ denotes the minimum eigenvalue of the Gibbs state.
\end{lemma}
\begin{proof} 
To ease notation, in the following we denote $\Delta:=gap(\cN)$.
 
	\underline{(i) Proof for the lower bound.} Denote the eigenvalues of  $\cN^\dagger$ as $\gamma_1=1>\gamma_2...\geq 0$, and the corresponding eigenstate  $Z_1=I,Z_2,...$. 	Note that since   $\cN^\dagger$ is Hermitian-preserving, the eigenstates $\{Z_j\}_j$ can be chosen to be Hermitian matrices. Indeed, suppose $Z_j$ is not Hermitian.  Denote the Kraus operator of the quantum channel $\cN$ is $\{K_u\}_u$,   one can see that $Z_j^\dagger$ is also an eigenstate of $\cN^\dagger$ with eigenvalue $\gamma_j$:
	\begin{align}
	\gamma_j Z_j	= \cN^\dagger(Z_j) =  \sum_u K_u^\dagger Z_j K_u \quad  \Rightarrow  \quad \gamma_j Z_j^\dagger =  \sum_u  K_u^\dagger Z_j^\dagger K_u = \cN^\dagger(Z_j^\dagger).
	\end{align}
   where we use $\gamma_j$ is real.  
	 Therefore  the eigenstates $Z_j, Z_j^\dagger$ can be replaced by the  Hermitian matrices $Z_j+ Z_j^\dagger, i(Z_j - Z_j^\dagger)$, which are also eigenstates of $\cN^\dagger$ with eigenvalue $\gamma_j$.
	
	 Let $\rho$ be any initial  state,  consider the eigenstate $Z_2$ with eignevalue $\gamma_2$, for any $t$,  we have that
	\begin{align}
		\gamma_2^t \,tr\left(\rho Z_2\right) &= tr\left(\rho \left(\cN^\dagger\right)^t\left(Z_2\right)\right)
		 = tr\left(\cN^t \left(\rho\right) Z_2\right)\\ &= tr\left(\left(\cN^t \left(\rho\right)-\rho_{\beta}\right) Z_2\right) \label{eq:rhominus}\\
		 &\leq \left\|\cN^t(\rho)-\rho_{\beta}\right\|_1 \left\|Z_2\right\| \label{eq:lowerbound}
	\end{align}
	where Eq.~(\ref{eq:rhominus}) comes from the fact that  $Z_1$ and $Z_2$ are orthogonal w.r.t. the weighted inner product, 
	$$
	tr\left(Z_1^\dagger \rho_{\beta}^{\frac{1}{2}} Z_2 \rho_{\beta}^{\frac{1}{2}}\right) =0 \quad \Rightarrow \quad tr(\rho_\beta Z_2)=0.
	$$
Since $Z_2$ is a Hermitian matrix, one can choose the quantum state $\rho$ such that $tr(\rho Z_2) = \|Z_2\|$. Thus Eq.~(\ref{eq:lowerbound}) implies that for such $\rho$, 
\begin{align}
	 \gamma_2^t = \left(1-\Delta\right)^t \leq \|\cN^t(\rho)-\rho_{\beta}\|_1. 
\end{align}
Thus 
\begin{align}
	\log (\epsilon^{-1}) \leq t_{mix}(\epsilon) \log\left(\frac{1}{1-\Delta}\right) \leq t_{mix}(\epsilon) \left(\frac{\Delta}{1-\Delta}\right).
\end{align}
which completes the proof.

\underline{(ii) Proof for the upper bound.} Using the $\chi^2_\alpha$-contraction for $\alpha=\frac{1}{2}$ (See Theorem 9, \cite{temme2010chi}), we have that for any initial state $\rho$ and any time $t$,
\begin{align}
	\| \cN^t (\rho) -\rho_{\beta}\|_1 \leq \left(1-\Delta\right)^t \sqrt{\chi^2(\rho,\rho_{\beta})}  \leq e^{-\Delta t} \sqrt{\chi^2(\rho,\rho_{\beta})}, 
\end{align}
where 
\begin{align}
	\sup_\rho \,\,\chi^2(\rho,\rho_{\beta}) = \sup_\rho \,\,\, tr(\rho \,\rho_{\beta}^{-\frac{1}{2}}\,\rho\, \rho_{\beta}^{-\frac{1}{2}}) -1 =  \sigma_{min}^{-1} -1.
\end{align}
Thus
$$
t_{mix}(\epsilon) \leq \Delta^{-1}\log\left(\frac{1}{\epsilon \sqrt{\sigma_{min}}}\right)  \leq \Delta^{-1}\log\left(\frac{1}{\epsilon \,\sigma_{min}}\right). 
$$
\end{proof}

\section{Examples for exponential decay of commutator norm}\label{app:commutator_decay}

Recall that we use $H$ and $\beta$ to denote a local Hamiltonian and an inverse temperature, and write $\rho_{\beta} = \exp(-\beta H)/Z_{\beta}$ for the corresponding Gibbs state. In this section we give two examples where the commutator norm  $\left\| \left[ \widehat{O}(\tau), \rho_{\beta} \right] \right\|_1$ decays  exponentially:

\begin{itemize}
    \item[(i1)] If $H$ is the 2D Toric code and $O$ is single-Pauli $X$, then $\forall \beta$, $\left\| \left[ \widehat{O}(\tau), \rho_{\beta} \right] \right\|_1\leq 2 e^{-1/\tau^2}$.
    \item[(i2)] Consider $H$ is a Hamiltonian where there is a constant energy gap $c$ between the ground energy and the first excited energy. Consider zero-temperature $\beta=+\infty$, where the Gibbs state $\rho_{\beta}$ becomes  the uniform mixture over the ground states. Then for any observable $O$ with $\|O\|\le 1$, we have $\left\| \left[ \widehat{O}(\tau), \rho_{\beta} \right] \right\|_1\leq 2 e^{-c^2/(4\tau^2)}$. 
\end{itemize}

To see that (i1) and (i2) hold, one can introduce a quantity that characterizes the extent to which $O$ disturbs the eigenspaces. In particular,  define
\begin{align}
    \mathcal{J}_O(\tau)
    := \sum_{k} p_k \exp\!\left(- \frac{\nu_k^{\,2}}{4\tau^{2}}\right),
\end{align}
where $p_k = \frac{\dim(\Pi_k)\, e^{-\beta E_k}}{Z_{\beta}}$ and the sum runs over all distinct eigenspaces of $H$, note that $\sum_k p_k = 1$.  
The quantity $\nu_k$ denotes the minimal nonzero energy change induced by $O$ when starting from energy $E_k$. That is, 
$\nu_k := 
    \min \Bigl\{\, |E_i - E_k| \;\Big|\; 
    E_i \in \mathrm{Spec}(H),\; E_i \neq E_k,\;
   \Pi_i O \Pi_k  \neq 0
    \Bigr\}.$ 
One can check that 
\begin{itemize}
    \item If there exists a constant \(c>0\) such that \(\nu_k\ge c\) for all \(k\), then \(\mathcal{J}_O(\tau)\le \exp(-c^2/(4\tau^2))\). For example, for the 2D Toric code and a single-qubit Pauli \(X\), we have \(\nu_k=2\) for all \(k\), thus \(\mathcal{J}_O(\tau)\leq \exp(-1/\tau^2)\).
    \item If there is a constant energy gap $c$ between the ground energy and the first excited energy and the temperature is extremely low as $\beta=+\infty$, then for any $O$, $\mathcal{J}_O(\tau)\leq \exp(-c^2/(4\tau^2))$.
\end{itemize}

To see that (i1)(i2) hold, it suffices to note that the commutator norm is upper bounded by this quantity,
 \begin{align}
\left\|\, 	\left[\widehat{O}(\tau),\rho_{\beta} \right] \,\right\|_1\leq  2 \,\mathcal{J}_{O}(\tau).
 \end{align} 	
 
 \begin{proof} Recall that the definition of $O_\nu$ as in Eq.~(\ref{eq:OFT}).
 First note that for any $k$, we have 
 $$O_0 \Pi_k = \Pi_k O \Pi_k = \Pi_k O_0,$$ thus 
 \begin{align}
  \left[\widehat{O}(\tau),\rho_{\beta} \right] = \left[\widehat{O}(\tau)- O_0,\rho_{\beta} \right].  \label{eq:OmiusO0}
   \end{align}

 	Let $\ket{\psi_k}$ be an eigenstate of $H$ with energy $E_k$.  By Eq.~(\ref{eq:OFT}), one can check that 
 	\begin{align}
 		\ket{\Phi_k(\tau)}:= \left(\widehat{O}(\tau)- O_0\right) \ket{\psi_k} = \sum_{j: \, j\neq k; \,\Pi_j O\Pi_k \neq 0} \exp\left(-\frac{(E_j-E_k)^2}{4\tau^2}\right) \Pi_j O \ket{\psi_k}.
 	\end{align}
Recall that $\|O\|\leq 1$. Thus we have
 \begin{align}
 	\left\| \ket{\Phi_k(\tau)}\right\|^2 &= \langle \psi_k | O \left(\sum_{j: \, j\neq k; \,\Pi_j O\Pi_k \neq 0}   \exp\left(-\frac{(E_j-E_k)^2}{2\tau^2}\right) \Pi_j \right)O \ket{\psi_k} \\
 	&\leq \left\|\sum_{j:j\neq k,  \Pi_j O \Pi_k \neq 0}   \exp\left(-\frac{(E_j-E_k)^2}{2\tau^2}\right) \Pi_j\right\|\\
 	&=  \exp\left(-\frac{\nu_k^2}{2\tau^2}\right).
 \end{align} 
 
 Thus 
  \begin{align}
 	\left\|\, \left[\widehat{O}(\tau)- O_0, \ket{\psi_k}\bra{\psi_k} \right] \,\right\|_1 \leq  2 \left\| \ket{\Phi_k(\tau)}\bra{\psi_k} \right\|_1 \leq 2\left| \ket{\Phi_k(\tau)}\right\| \leq  2\exp\left(-\frac{\nu_k^2}{4\tau^2}\right).
 \end{align}

 To prove the Lemma it suffices to express $\rho_{\beta}$ in the eigenstate of $H$ and use triangle inequality for the trace norm.
 \end{proof}

\section{Gaussian-filtered quantum phase estimation (GQPE)}\label{app:GQPE}

In this section, we present a detailed discretization scheme for the idealized Gaussian-filtered quantum phase estimation described in Section~\ref{sec:GQPE}. In particular, Sections~\ref{sec:ideal_GQPE} and \ref{sec:implement_qubit} describe the discretization and its implementation on a qubit-based quantum computer, with the parameters summarized in Algorithm~\ref{alg:GQPE}. From Section~\ref{sec:RE_PP} onward, we provide a rigorous resource analysis that establishes Theorem~\ref{thm:GQPE}.

\subsection{Implementation for Idealized Gaussian Filtered quantum phase estimation} \label{sec:ideal_GQPE}

To build intuition, here we provide an (imprecise) explanation of the implementation of the idealized Gaussian-filtered quantum phase estimation,  assuming the ancilla system takes continuous values, i.e., the quantum channel $\cM_*$ introduced in Section~\ref{sec:GQPE}. 
In the following section, we will discretize this process for implementation on qubits and carefully analyze the performance of the discretized version.

 Define the Gaussian filter $g(\omega)$ and its Fourier transform:
\begin{equation}
    g_{\gamma}(\omega) = \frac{1}{\sqrt{\gamma\sqrt{2\pi}}}e^{-\omega^2/4\gamma^2},\quad \widehat{g}_{\gamma}(\xi) = \int^{\infty}_{-\infty} e^{-2\pi i \xi\omega}g_{\gamma}(\omega)\d \omega = \left(2\sqrt{2\pi}\gamma\right)^{1/2} e^{-4\pi^2\gamma^2\xi^2}.
\end{equation}

 The measurement $\cM_*$ can be implemented in the following way. Consider two registers and  an input state $\ket{\Psi}$ on register 1,
\begin{enumerate}
	\item Prepare a resource state $\ket{\widehat{G}_*}$ on register 2 where
		$\ket{\widehat{G}_*}:=\int_{\xi=-\infty}^{+\infty} \widehat{g}_{\gamma}(\xi) \ket{\xi}\d \xi.$
	\item Define 
	 controlled Hamiltonian evolution w.r.t. $\bO$  as
		$W_* := \int_{\xi=-\infty}^{+\infty}
		e^{-2\pi i \xi \bO} \otimes \ketbra{\xi}{\xi}.$
    First perform controlled Hamiltonian evolution $W_*$ on $\ket{\Psi}\otimes\ket{\widehat{G}_*}$, then perform the inverse Fourier transform on register $2$, which results in the state	
    \begin{align}
       \ket{\tilde{\Psi}_*} &:= \int^\infty_{\omega=-\infty}\left(\int^\infty_{\xi= -\infty} e^{2\pi i \xi \omega}\widehat{g}_{\gamma}(\xi) e^{-2\pi i \xi \bO} \d \xi \right) \ket{\Psi}\otimes \ket{\omega}, \label{eq:ideal_main}\\
       &= \frac{1}{\sqrt{\gamma\sqrt{2\pi}}}\int^{\infty}_{\omega=-\infty} \exp\left[-\frac{(\omega-\bO)^2}{4\gamma^2}\right] \ket{\Psi}\otimes \ket{\omega} ;
	  \end{align}
	\item  Measuring the  readout register (i.e.,  $\ket{\omega}$) in computational basis.
\end{enumerate}

\subsection{Circuit Implementation on qubits} \label{sec:implement_qubit}

 In this section we implement a measurement $\cM$ on qubit systems which corresponds to a discretization of the idealized measurement $\cM_*$.

Recall  Eq.~(\ref{eq:ideal_main}) that describes the ideal target state,
     $$  \ket{\tilde{\Psi}_*} := \int^\infty_{\omega=-\infty}\left(\int^\infty_{\xi= -\infty} e^{2\pi i \xi \omega}\widehat{g}_{\gamma}(\xi) e^{-2\pi i \xi \bO} \d \xi \right) \ket{\Psi}\otimes \ket{\omega}$$

To perform numerical integration, we discretize this double integral on the following quadrature points ($0\le j, \ell \le N-1$):
\[\omega_j = \left(j - \frac{N}{2}\right)h, \quad \xi_\ell = \left(\ell - \frac{N}{2}\right)h,\]
where $N$ is a discretization number  and $h > 0$ is a step size. Their values will be specified later.
The resulting numerical summation reads:
\begin{equation}\label{eqn:SNh}
    S^{[N]}_h = \sum^{N-1}_{j =0}\sqrt{h}\left(\sum^{N-1}_{\ell=0} h e^{2\pi i \xi_\ell \omega_j} \widehat{g}_{\gamma}(\xi_\ell) e^{-2\pi i \xi_\ell \bO} \right) \ket{\Psi} \otimes \ket{j},
\end{equation}
Here we use $\ket{j}$ to denote $\ket{\omega_j}$. 
Note that the outer integral with respect to $\omega$ has a measure element $\sqrt{h}$ to guarantee $L^2$-normalization of quantum states, 
and the inner integral (for the inverse Fourier transform) has a measure element $h$.

 For convenience,
in what follows, we assume there is a positive integer $m$ and $h = 2^{-m}$. We also choose the discretization number such that
\begin{equation}
    N = \frac{1}{h^2} = 2^{2m},
\end{equation}
so that the energy readout register $\{j\}^{N-1}_{j=0}$ can be represented using $2m$ qubits.  It is worth noting that $S^{[N]}_h$ is not a quantum state as it does not have a unit $2$-norm:
\begin{equation}\label{eqn:CNh-defn}
    C^{[N]}_h \coloneqq \left\|S^{[N]}_h\right\| = \sqrt{h^3N \sum^{N-1}_{\ell=1}|\widehat{g}_{\gamma}(\xi_\ell)|^2} 
    = \sqrt{h \sum^{N-1}_{\ell=1}|\widehat{g}_{\gamma}(\xi_\ell)|^2}.
\end{equation}
While the normalization factor $C^{[N]}_h$ is not unit, it converges exponentially fast to $1$ as we increase $N$ due to the rapid decay of $|\widehat{g}_{\gamma}|^2(\xi)$, as we will show later in the proof analysis.
Define the following two quantum operations: 
\begin{itemize}
    \item A (shifted) quantum Fourier transform:
    \begin{align}
    	\cF_s \ket{\ell} \mapsto \frac{1}{\sqrt{N}}\sum^{N-1}_{j=0}\exp\left(\frac{2\pi i (\ell-N/2)(j-N/2)}{N}\right)\ket{j}. \label{eq:Fs}
    \end{align}
       Since we choose $N = 1/h^2$, this sQFT implements:
    \[\cF_s \ket{\ell} \mapsto h\sum^{N-1}_{j=0}e^{2\pi i \xi_\ell \omega_j}\ket{j}.\]
    Note that
        $\frac{2\pi i (\ell-N/2)(j-N/2)}{N}= \frac{2\pi i \ell j}{N} - \pi i (\ell+j) + \frac{\pi i N}{2}.$
     Thus $\cF_s$ can be implemented by applying appropriate diagonal phase gates, then performing the standard quantum Fourier transform, up to a global phase $e^{\pi i N/2}$.
      \item A controlled unitary (i.e., select oracle):
    \begin{align}
    W = \sum^{N-1}_{\ell=0}e^{-2\pi i \xi_\ell \bO} \otimes \ketbra{\ell}{\ell} = \sum^{N-1}_{\ell=0}e^{-2\pi i (\ell-N/2) h \bO} \otimes \ketbra{\ell}{\ell}. \label{eq:W}
      \end{align}
\end{itemize}

We are now ready to implement the measurement $\mathcal{M}$ on qubits as required for Theorem~\ref{thm:GQPE}. The corresponding quantum circuit is summarized in Algorithm~\ref{alg:GQPE}.

\begin{algorithm}[ht]
 \caption{Implementation of Gaussian Filtered Quantum Phase estimation (GQPE) $\cM$}\label{alg:GQPE}
  \hspace*{\algorithmicindent} \textbf{Inputs:} an observable $O$ with $\|O\|\leq \kappa$, a quantum state $\ket{\Psi}$, a precision parameter $\epsilon$,  \\
  \hspace*{\algorithmicindent} \textbf{\quad\quad\quad}
    $2m$ ancilla qubits, and a parameter $\gamma>0$ (specified in Eqs.~(\ref{eqn:param_thm3})(\ref{eq:m})).\\
   \hspace*{\algorithmicindent} \textbf{Outputs:} 
 a measurement outcome $W \in \mathbb{R}$.
 \\
    \hspace*{\algorithmicindent} \textbf{Registers:} 
  Register $1$ is of $n$ qubits initialized as $\ket{\Psi}$; Register $2$ is of $2m$ ancillary qubits.
\begin{algorithmic}[1]
    \State Normalize $O$ 
    by setting $\tilde{O}=O/\kappa$.\label{line:normal}
    \State   Denote $N=2^{2m}$ and $h=1/\sqrt{N}=1/2^{m}$.
	\State 
    Prepare a Gaussian resource state $\ket{\widehat{G}}$ on register $2$, defined as (where $\xi_\ell$ and $C^{[N]}_h$ are the same as in~\cref{eqn:CNh-defn})
	\begin{align}
    \ket{\widehat{G}} = \frac{1}{C^{[N]}_h}\sum^{N-1}_{\ell=1} \sqrt{h}\,\,\widehat{g}_{\gamma}(\xi_\ell) \ket{\ell}.
	\end{align}
	\State Apply the controlled unitary $W$ and $I\otimes \cF_s$ to  $\ket{\Psi}\otimes \ket{\widehat{G}}$ and get the target state $\ket{\widehat{\Psi}}$, where $\cF_s$ defined in Eq.(\ref{eq:Fs}) and $W$ defined in Eq.~(\ref{eq:W}) w.r.t. $\widetilde{O}$:  
    \label{line:4}
\begin{align}
    \ket{\Psi}\otimes \ket{\widehat{G}} &\xrightarrow[]{W} \frac{\sqrt{h}}{C^{[N]}_h} \sum^{N-1}_{\ell=1} \widehat{g}_{\gamma}(\xi_\ell)  e^{-2\pi i \xi_\ell \tilde{O}}\ket{\Psi}\otimes \ket{\ell} \nonumber\\
    &\xrightarrow[]{I\otimes \cF_s} \frac{h^{3/2}}{C^{[N]}_h} \sum^{N-1}_{\ell=1}\sum^{N-1}_{j=0} e^{2\pi i \xi_\ell \omega_j} \widehat{g}_{\gamma}(\xi_\ell) e^{-2\pi i \xi_\ell \tilde{O}}\ket{\Psi} \otimes \ket{j} = \ket{\widehat{\Psi}}.
\end{align}
	\State  Measure register $2$ in the computational basis and obtain an outcome $j \in \{0,\dots,N-1\}$. (Conditioned on this outcome, the post-measurement state on register $1$ is proportional to $\bO_j\ket{\Psi}$, for $O_j$ defined in Eq.(\ref{eqn:defn-Oj}).)  Let $\omega = \left(j - \frac{N}{2}\right)h$, and we define
    \begin{equation}\label{eqn:raw_readout}
        Z \coloneqq \begin{cases}
            \omega & -2 \le \omega \le 2,\\
            0 & \mathrm{otherwise}.
        \end{cases}
    \end{equation}
    \State Output the measurement outcome: $W = \kappa Z$.
\end{algorithmic}
\end{algorithm}

\subsection{Resource Estimation: Proof of Theorem~\ref{thm:GQPE}}\label{sec:RE_PP}

\subsubsection{Outline of proof}

\cref{thm:GQPE} guarantees that the measurement $\mathcal{M}$ recovers the observable value of $O$. Furthermore, the measurement procedure leaves the Gibbs state $\rho_\beta$ undisturbed if $[\rho_\beta, O] = 0$; and only slightly perturbs $\rho_\beta$ when $[\rho_\beta, O]$ is small. We demonstrate those by characterizing the expectation, variance, covariance, and commutativity properties of the measurement channel $\mathcal{M}$. Specifically, the following subsections establish a series of technical results regarding the measurement statistics:
\begin{itemize}
\item Expectation and variance of $Z$:~\cref{cor:exp_var},
\item Covariance of two consecutive $Z$'s:~\cref{prop:gqpe-main-2},
\item Commutative properties of $\mathcal{M}$:~\cref{prop:com}.
\end{itemize}

Note that the above technical lemmas all assume that the spectral norm of the observable operator $O$ is bounded by $1$. In what follows, we show that the error bound for case $\|O\| \le 1$ implies the desired bound for the general case $\|O\| \le \kappa = \poly(n)$. More specifically, for a given precision parameter $\epsilon$, we consider the normalized observable $\tilde{O}$ and   a new precision parameter $\tilde{\epsilon}$, where $$\tilde{O} = O/\kappa, \quad \tilde{\epsilon} = \epsilon / \kappa.$$
Let $Z$ be the raw readout in~\cref{alg:GQPE}, as given in~\cref{eqn:raw_readout}. By~\cref{cor:exp_var}, we can choose
\begin{equation}\label{eqn:param_thm3}
    \gamma^{-1} =  \cO\left(\kappa \log^{1/2}\left( 1/\tilde{\epsilon} \right)\right),\quad h^{-1} = \cO\left( \gamma^{-1}\log^{1/2}\left(\gamma^2/\tilde{\epsilon}\right)\right)
\end{equation}
such that
\begin{equation}
    |\mathbb{E}[Z] - \bE_\rho[\tilde{O}]| \le \tilde{\epsilon}, \quad |\bV[Z] - \bV_\rho[\tilde{O}]| \le 3/\kappa^2,
\end{equation}
which immediately implies that the final output $W = \kappa Z$ satisfies:
\begin{equation}
    |\mathbb{E}[W] - \bE_\rho[O]| \le \epsilon, \quad |\bV[W] - \bV_\rho[O]| \le 3.
\end{equation}
The covariance property and the commutative property directly follow from this error analysis, as discussed in~\cref{prop:gqpe-main-2} and~\cref{prop:com}.

\paragraph{Implementation cost.}
With the parameter choice as in~\cref{eqn:param_thm3}, we have the number of ancilla qubits $2m$ and the maximal Hamiltonian simulation time $t_{\max}$ (of the original observable $O$) as follows:
\begin{equation}\label{eq:m}
    m = \log_2(1/h) = \cO\left(\log (\kappa) + \log\log (\epsilon^{-1})\right),\quad t_{\max} = \frac{Nh}{2\kappa} = \cO\left(\log(\kappa) + \log (\epsilon^{-1})\right)
\end{equation}
The Gaussian resource state $\ket{\widehat{G}}$ can be prepared using the methods in \cite{grover2002creating,kitaev2008wavefunction}, with gate complexity $\poly(m) = \polylog(\kappa) + \polylog(\epsilon^{-1})$.

\subsubsection{Expectation and variance}\label{app:main_tech}

In the following analysis, we define a new operator for each $j \in \{0,\dots, N-1\}$:
\begin{equation}\label{eqn:defn-Oj}
 \bO_j:= \bO_{ w_j} := \frac{h^{3/2}}{C^{[N]}_h} \sum^{N-1}_{\ell=1} e^{2\pi i \xi_\ell (\omega_j - \tilde{O})} \,\widehat{g}_\gamma(\xi_\ell) 
 =  \frac{h^{3/2}}{C^{[N]}_h} \sum^{N/2-1}_{k=1-N/2} e^{2\pi i kh (\omega_j - \tilde{O})} \,\widehat{g}_\gamma(kh)  
\end{equation}
One can check that Algorithm \ref{alg:GQPE} implement a quantum channel $\cM$ where
\begin{equation}
    \mathcal{M}[\rho] = \sum^{N-1}_{j=0} \bO_{ w_j} \rho \bO^\dagger_{ w_j}.\label{eq:MOj}
\end{equation}

\begin{proposition}\label{prop:gqpe-main-1} Let $O$ be an observable with $\|O\|\leq 1$. There exist parameter choices
\begin{equation}\label{eqn:param_choice_exp_var}
    \gamma^{-1}=\cO\left(\log^{1/2}(1/\epsilon)\right),\quad h^{-1}=\cO\left(\gamma^{-1}\log^{1/2}(\gamma^2/\epsilon)\right)
\end{equation}
such that
\begin{align}
&\left\|\sum_{|\omega_j| \le 2} \omega_j O^\dagger_j O_j -  O\right\| \le \epsilon, \quad \left\|\sum_{|\omega_j| \le 2} \omega^2_j O^\dagger_j O_j -  \left(O^2+\gamma^2 I\right)\right\| \le 2\gamma^2. \label{eqn:gqpe-main-1-eq-2}
\end{align}
\end{proposition}

The proof of~\cref{prop:gqpe-main-1} is postponed to the end of this subsection. Below, we state an immediate corollary that characterizes the measurement statistics $Z$.

\begin{corollary}[Expectation and Variance]\label{cor:exp_var} Let $O$ be an observable with $\|O\|\leq 1$. Define 
\[\bE_\rho[O] \coloneqq \Tr[\rho O],\quad \bV_\rho[ O] \coloneqq \bE_\rho[ O^2]-\bE_\rho[ O]^2\]
For any quantum state $\rho$, there exists parameter choices of $\gamma$ and $h$ as in~\cref{eqn:param_choice_exp_var} such that the measurement statistics $Z$ in~\cref{alg:GQPE} satisfy
\begin{equation}
    |\mathbb{E}[Z] - \Tr[\rho\,  O]| \le \epsilon,\quad |\bV[Z] - \bV_\rho[ O]| \le 3\gamma^2.
\end{equation}
\end{corollary}
\begin{proof}
Given a quantum state $\rho$, the expectation and variance of the random variable $Z$ are:
\begin{align*}
    \bE[Z] &= \sum_{-2 \le \omega_j \le 2}  \omega_j \Tr[O^\dagger_j O_j \rho] = \Tr\left[\left(\sum_{|\omega_j|\le2}  \omega_j O^\dagger_j O_j\right)\rho\right],\\
    \bV[Z] &= \bE[Z^2] - \bE[Z]^2,\quad \bE[Z^2] = \Tr\left[\left(\sum_{|\omega_j|\le2}  \omega^2_j O^\dagger_j O_j\right)\rho\right].
\end{align*}
Then, the proof follows from \cref{prop:gqpe-main-1} and the triangle inequality. Under the assumption that $\epsilon \ll \gamma^2$, the error in the squared expectation value can be bounded by an additional term of $\gamma^2$.
\end{proof}

\begin{lemma}\label{lem:Owg} Let $O$ be an observable with $\|O\|\leq 1$. For any $\gamma>0$, and $-2\le \omega \le 2$, define
\begin{equation}
    O_{ \omega} \coloneqq \frac{h^{3/2}}{C^{[N]}_h} \sum^{N-1}_{\ell=1} e^{2\pi i \xi_\ell (\omega - O)} \widehat{g}_{\gamma}(\xi_\ell),
\end{equation}
where $h$, $C^{[N]}_h$, and $\xi_\ell$ are the same as in~\cref{eqn:defn-Oj}.
Clearly, we have $O_{\omega_j} = O_j$ for $|\omega_j|\le 2$. Then, there is an absolute constant $c>0$ such that for any $\gamma \le \frac{1}{2\sqrt{2}\pi}$, we can choose $h \le c/\gamma$ and
\begin{equation}\label{eqn:res-lem-main-1}
    \left\|O_{\omega} - \sqrt{h}g_{\gamma}(\omega-O)\right\| \le \frac{12h^{1/2}}{\gamma^{1/2}}e^{-\frac{1}{16\gamma^2h^2}} + \frac{h^{3/2}}{2\gamma^{3/2}}e^{-\frac{2\pi^2\gamma^2}{h^2}}.
\end{equation}
\end{lemma}
\begin{proof}
To ease notation, in the proof we abbreviate $g_\gamma, \widehat{g}_{\gamma}$ as $g$ and $\widehat{g}$.
    Since $O$ is a Hermitian operator, it suffices to prove~\cref{eqn:res-lem-main-1} in the eigenbasis of $O$. In the following argument, we demonstrate that the lemma is true in any one-dimensional eigensubspace of $O$. 
    We first define
    \[I^{[N]}_h \coloneqq h \sum^{N/2-1}_{k= -N/2+1} e^{2\pi i kh (\omega - E)} \widehat{g}(kh),\]
    and by the triangle inequality, we have 
    \begin{equation}\label{eqn:res-lem-main-1-1d}
        \left|\frac{h^{3/2}}{C^{[N]}_h} \sum^{N-1}_{\ell=1} e^{2\pi i \xi_\ell (\omega - E)} \widehat{g}(\xi_\ell) - \sqrt{h}g(\omega-E)\right| \le \sqrt{h}\left[\frac{1}{\left|C^{[N]}_h\right|}\left|I^{[N]}_h - g(\omega-E)\right| + \left|\frac{1}{C^{[N]}_h}-1\right||g(\omega-E)|\right].
    \end{equation}
    
    \begin{enumerate}
        \item First, we estimate the quadrature error $|I^{[N]}_h - g(\omega-E)|$. 
        For simplicity, we denote
    \[\varphi(x) = e^{2\pi i x (\omega-E)}\widehat{g}(x).\]
    Note that $\varphi(\xi)$ is an analytic function on $\mathbb{R}$ and its Fourier transform is
    \begin{align}
        &\widehat{\varphi}(y) = \int^\infty_{-\infty}e^{-2\pi i x y}\varphi(x)~\d x = g(\omega-E- y),\\
       \text{in particular } &\widehat{\varphi}(0) = \int^\infty_{-\infty}\varphi(x)~\d x = g(\omega-E).
    \end{align}

    Then, the numerical sum can be simply written as $I^{[N]}_h = h \sum^{N/2-1}_{k = -N/2} \varphi(kh)$.
    By Lemma~\ref{lem:riemann-sum-error},
    \begin{align}\label{eqn:Oj-part-1}
        |I^{[N]}_h - g(\omega - E)| \le \sum^\infty_{n = -\infty,n\neq 0} \|\widehat{\varphi}(n/h)\| + h\sum_{|k| \ge N/2} \|\varphi(kh)\|
    \end{align}
    Since $|\omega-E| \le 3$, by choosing $h \le 1/6$,
    \[\left\|\widehat{\varphi}\left(\frac{n}{h}\right)\right\| = \left\|g\left(\omega - E - \frac{n}{h}\right)\right\| \le \left\|g\left(\frac{|n|-1/2}{h}\right)\right\| \le \frac{1}{\sqrt{\gamma \sqrt{2\pi}}}\exp\left(-\frac{(|n|-1/2)^2}{4h^2\gamma^2}\right).\]
    Plugging this upper bound into~\cref{eqn:Oj-part-1} and invoking Lemma~\ref{lem:gaussian-tail}, we have
    \begin{align*}
        |I^{[N]}_h - g(\omega-E)| &\le  \frac{2}{\sqrt{\gamma \sqrt{2\pi}}} \sum^\infty_{n = 1} \exp\left(-\frac{(|n|-1/2)^2}{4h^2\gamma^2}\right) + 2h\left(2\sqrt{2\pi} \gamma\right)^{1/2} \sum_{k\ge N/2} e^{-4\pi^2\gamma^2h^2k^2} \\ 
        &\le \frac{6}{\sqrt{\gamma\sqrt{2\pi}}}e^{-\frac{1}{16\gamma^2h^2}} + 2h\left(2\sqrt{2\pi} \gamma\right)^{1/2} \left(1+\frac{1}{4\pi^2\gamma^2}\right)e^{-\frac{\pi^2\gamma^2}{h^2}}\\
        &\le \frac{6}{\sqrt{\gamma\sqrt{2\pi}}}e^{-\frac{1}{16\gamma^2h^2}} + \frac{h\sqrt{2\sqrt{2\pi}}}{\pi^2\gamma^{3/2}}
        e^{-\frac{\pi^2\gamma^2}{h^2}}.
    \end{align*}
    Note that we use $4\pi^2\gamma^2\le 1$ in the last step, and the first part in the second inequality follows from
    \begin{align*}
        &\sum^\infty_{n = 1} \exp\left(-\frac{(|n|-1/2)^2}{4h^2\gamma^2}\right) \le e^{-\frac{1}{16h^2\gamma^2}} + \sum^\infty_{n = 1} e^{-\frac{n^2}{4h^2\gamma^2}} \le e^{-\frac{1}{16h^2\gamma^2}} + \left(1+2h^2\gamma^2\right)e^{-\frac{1}{4h^2\gamma^2}} \le 3 e^{-\frac{1}{16h^2\gamma^2}}.
    \end{align*}
    In the third step, we use Lemma~\ref{lem:gaussian-tail}; the last step follows from the fact that $2h^2\gamma^2 \le 1$.

    \item Next, we consider the normalization error $|1/C^{[N]}-1|$. 
    We denote a new integrand function $\phi(x) = |\widehat{g}(x)|^2 = 2\sqrt{2\pi}\gamma e^{-8\pi^2\gamma^2x^2}$. This is a Gaussian density function, and its Fourier transform is $\widehat{\phi}(y) = e^{-y^2/(8\gamma^2)}$. 
    Again, by~\cref{lem:riemann-sum-error}, we have
    \begin{align*}
        &\left|\left(C^{[N]}_h\right)^2 - 1\right| \le \sum^{\infty}_{n=-\infty,n\neq 0}  e^{-n^2/(8h^2\gamma^2)} + 2h \sum_{k \ge N/2} 2\sqrt{2\pi}\gamma e^{-8\pi^2\gamma^2k^2h^2}\\
        &\le 2\sum_{n\ge 1} e^{-n/(8h^2\gamma^2)} + 4\sqrt{2\pi}h\gamma \sum_{k \ge N/2} e^{-8\pi^2\gamma^2h^2k^2}\\
        &\le \frac{2}{e^{\frac{1}{8h^2\gamma^2}}-1} + 4\sqrt{2\pi}h\gamma \left(1+\frac{1}{8\pi^2\gamma^2}\right)e^{-2\pi^2\gamma^2/h^2} \le 3e^{-\frac{1}{8\gamma^2 h^2}} + \frac{h\sqrt{2}}{\pi^{3/2}\gamma}e^{-\frac{2\pi^2\gamma^2}{h^2}}.
    \end{align*}
    Note that $Nh^2=1$. We use Lemma~\ref{lem:gaussian-tail} in the second inequality, and we assume $\gamma \le \frac{1}{2\sqrt{2}\pi}$ in the last step.

    \item Note that $g(\omega-E) \le \frac{1}{\sqrt{\gamma\sqrt{2\pi}}}$. By the previous part, we can always choose $h\leq c/\gamma$ for an absolute constant $c>0$ such that $|C^{[N]}_h-1| \le 1/2$. Then, we have
    \begin{align*}
        &\left|\frac{h^{3/2}}{C^{[N]}_h} \sum^{N-1}_{\ell=0} e^{2\pi i \xi_\ell (\omega - E)} \widehat{g}(\xi_\ell) - \sqrt{h}g(\omega-E)\right| \le\\
        &\sqrt{h}\left[2\left(\frac{6}{\sqrt{\gamma\sqrt{2\pi}}}e^{-\frac{1}{16\gamma^2h^2}} + \frac{h\sqrt{2\sqrt{2\pi}}}{\pi^2\gamma^{3/2}}
        e^{-\frac{\pi^2\gamma^2}{h^2}}\right) + \frac{2}{\sqrt{\gamma\sqrt{2\pi}}} \left(3e^{-\frac{1}{8\gamma^2 h^2}} + \frac{h\sqrt{2}}{\pi^{3/2}\gamma}e^{-\frac{2\pi^2\gamma^2}{h^2}}\right)\right]\\
        &\le \frac{12h^{1/2}}{\gamma^{1/2}}e^{-\frac{1}{16\gamma^2h^2}} + \frac{h^{3/2}}{2\gamma^{3/2}}e^{-\frac{2\pi^2\gamma^2}{h^2}}.
    \end{align*}
    \end{enumerate}

    This proves the lemma for any $-2 \le \omega \le 2$.
\end{proof}

\begin{lemma}\label{lem:gaussian-truncation-error}
    Let $O$ be a quantum observable such that $\|O\| \le 1$. Then, given an $\epsilon > 0$, there exists parameter choices of $\gamma$ and $h$ as in~\cref{eqn:param_choice_exp_var} such that
    \begin{align}
        \left\| h\sum_{|\omega_j| \le 2}  \omega_j g_{\gamma}^2(\omega_j - O) -   O\right\| \le \epsilon,\quad \left\|h\sum_{|\omega_j| \le 2}  \omega^2_j g_{\gamma}^2(\omega_j - O) - \left(O^2+ \gamma^2 I\right)\right\| \le \gamma^2.
    \end{align}
\end{lemma}

\begin{proof} 
To ease notation, in the proof, we abbreviate $g_{\gamma}, \widehat{g}_{\gamma}$
as $g$ and $\widehat{g}$.
Since $O$ is Hermitian, it suffices to prove the lemma in the eigenbasis of $O$. By $\|O\| \le 1$, it reduces to proving the following inequalities for any $-1 \le E \le 1$:
\begin{equation}
    \left| h\sum_{|\omega_j| \le 2} \omega_j g^2(\omega_j - E) -  E\right| \le \epsilon,\quad \left| h\sum_{|\omega_j| \le 2} \omega^2_j g^2(\omega_j - E) - \left( E^2+\gamma^2 \right)\right| \le \gamma^2.
\end{equation}
\begin{enumerate}
    \item To prove the first part, we define the integrand function $\varphi(x) \coloneqq x g^2(x-E) = \frac{x}{\gamma \sqrt{2\pi}}e^{-(x-E)^2/2\gamma^2}$, and its Fourier transform reads $\widehat{\varphi}(y) = (E-2\pi i \gamma^2 y)e^{-2\pi i E y}e^{-2\pi^2\gamma^2 y^2}$.
    We define:
    \begin{align}
        I_h = h\sum^{\infty}_{k=-\infty} \varphi(kh),
    \end{align}
    By the triangle inequality:
    \[|h\sum_{|\omega_j| \le 2} \omega_j g^2(\omega_j - E) - E| \le |h\sum_{|\omega_j| \le 2} \omega_j g^2(\omega_j - E) - I_h| + |I_h - E|.\]
    The first term can be bounded using the fact that $|E|\le 1$,
    \begin{align}
        &|h\sum_{|\omega_j| \le 2} \omega_j g^2(\omega_j - E) - I_h| \le h \sum_{|jh| > 2} |\varphi(jh)|\\
        &\le \frac{2h}{\gamma \sqrt{2\pi}} \sum_{j\ge 1} (2+jh) e^{-(1+jh)^2/2\gamma^2}
        \le \frac{2}{\sqrt{2\pi}}\left(\gamma+\sqrt{\frac{\pi}{2}}\right)e^{-\frac{1}{2\gamma^2}}.
    \end{align}
    The second term can be bounded using the Poisson summation formula, 
    \begin{align}
        |I_h - E| &= \left| \sum^\infty_{n=-\infty,n\neq 0} \widehat{\varphi}(n/h)\right| \le \frac{8\pi \gamma^2}{h}\sum^\infty_{n=1} ne^{-2\pi^2\gamma^2n^2/h^2}\\
        &\le \frac{8\pi \gamma^2}{h} \frac{e^{-2\pi^2\gamma^2/h^2}}{(1-e^{-2\pi^2\gamma^2/h^2})^2} \le \frac{8\pi\gamma^2}{h (e^{2\pi^2\gamma^2/h^2}-2)},
    \end{align}
    where the second step assumes that $h \le 2\pi \gamma^2$ (i.e., $2\pi \gamma^2/h \ge 1 \ge |E|$). 
    Note that $h \ge e^{-\pi^2\gamma^2/h^2}$ for any $h \le \min(1,\sqrt{2}\pi\gamma)$. We further require $h \le \sqrt{2}\pi \gamma/\sqrt{\ln(4)}$ (i.e., $e^{2\pi^2\gamma^2/h^2}-2\ge \frac{e^{2\pi^2\gamma^2/h^2}}{2}$), and it follows that:
    \[\frac{8\pi\gamma^2}{h (e^{2\pi^2\gamma^2/h^2}-2)} \le \frac{16\pi\gamma^2}{h e^{2\pi^2\gamma^2/h^2}} \le \frac{16\pi\gamma^2}{e^{\pi^2\gamma^2/h^2}},\]
    which is bounded by $\epsilon/2$ if we choose $1/h \ge \frac{1}{\pi \gamma}\log^{1/2}(32\pi\gamma^2/\epsilon)$.
    Combining the above error bounds, it turns out that if we choose
    \[\frac{1}{\gamma} = \sqrt{2}\cdot\max\left(\pi,\log^{1/2}\left(\frac{4}{\epsilon}\right)\right), \quad \frac{1}{h} = \max\left(\pi, \frac{1}{\pi \gamma}\log^{1/2}\left(\frac{32\pi\gamma^2}{\epsilon}\right)\right),\] 
    we have $\left|h\sum_{|\omega_j| \le 2} \omega_j g^2(\omega_j - E) - E\right| \le \frac{\epsilon}{2}+\frac{\epsilon}{2} = \epsilon$.

    \item To prove the second part, we first define the integrand function $\phi(x) \coloneqq x^2 g^2(x-E) = \frac{x^2}{\gamma\sqrt{2\pi}}e^{-(x-E)^2/2\gamma^2}$. The Fourier transform of $\phi(x)$ is
    \[\widehat{\phi}(y) = \left(\gamma^2 + E^2 - 4\pi i E\gamma^2 y - 4\pi^2\gamma^4 y^2\right) e^{-2\pi i E y}e^{-2\pi^2\gamma^2y^2}.\]
    Note that $\int^\infty_{-\infty}\phi(x)~\d x = \widehat{\phi}(0) = E^2+\gamma^2$.
    We define $J_h = h\sum^\infty_{k=-\infty} \phi(kh)$, by the triangle inequality, 
    \[|h\sum_{|\omega_j|\le 2} \omega^2_j g^2(\omega_j-E) - (E^2+\gamma^2)| \le |h\sum_{|\omega_j|\le 2} \omega^2_j g^2(\omega_j-E) - J_h| + |J_h - (E^2+\gamma^2)|.\]
    The first term can be bounded using $|E|\le 1$,
    \begin{align}
        &|h\sum_{|\omega_j| \le 2} \omega^2_j g^2(\omega_j - E) - J_h| \le h \sum_{|jh| > 2} |\phi(jh)|\\
        &\le \frac{2h}{\gamma \sqrt{2\pi}} \sum_{j\ge 1} (2+jh)^2 e^{-(1+jh)^2/2\gamma^2}
        \le \frac{64\gamma^3}{e\sqrt{2\pi}(1+h)}e^{-\frac{(h+1)^2}{4\gamma^2}}.
    \end{align}

Again, we invoke the Poisson summation formula to bound the second term. Let
\(a:=2\pi^2\gamma^2/h^2\), we have
\begin{align}
\left|J_h-(E^2+\gamma^2)\right|
\leq
8a\sum_{n=1}^\infty n^2e^{-an^2} \leq16ae^{-a}
\leq\frac{\gamma^2}{2},
\end{align}
where the last inequality follows from the parameter choice of \(h\), choosing its constant sufficiently large.

    By choosing
    \[\gamma^{-1} = \cO\left(\log^{1/2}\left(4\epsilon^{-1} \right)\right),\quad h^{-1} = \cO\left( \gamma^{-1}\log^{1/2}\left(\gamma^2{\epsilon}^{-1}\right)\right), \] 
    we have
    \[\left|h\sum_{|\omega_j|\le 2} \omega^2_j g^2(\omega_j-E) - (E^2+\gamma^2)\right| \le \frac{64\gamma^3}{e\sqrt{2\pi}(1+h)}e^{-\frac{(h+1)^2}{4\gamma^2}}+ 16ae^{-a} \le \frac{\gamma^2}{2}+\frac{\gamma^2}{2}\le \gamma^2.\]
\end{enumerate}
\end{proof}

Now, we are ready to prove~\cref{prop:gqpe-main-1}.

\begin{proof}[of~\cref{prop:gqpe-main-1}] 
To ease notation, in the proof we abbreviate $g_\gamma, \widehat{g}_{\gamma}$ as $g$ and $\widehat{g}$.
    \begin{enumerate}
        \item By the triangle inequality, we have
    \begin{equation}
        \|\sum_{|\omega_j|\le 2} \omega_j O^\dagger_j O_j - O\| \le I_1 + I_2 + I_3, \label{eq:157}
    \end{equation}
    where 
    \begin{align*}
        I_1 &= \sum_{|\omega_j|\le 2} |\omega_j|\|O_j-\sqrt{h}g(\omega_j-O)\|^2 \le \frac{1152}{\gamma} e^{-\frac{1}{8\gamma^2h^2}} + \frac{2h^2}{\gamma^{3}}e^{-\frac{4\pi^2\gamma^2}{h^2}},\\
        I_2 &= 2\sqrt{h} \sum_{|\omega_j|\le 2} |\omega_j| \|O_j - \sqrt{h}g(\omega_j-O)\| \|g(\omega_j - O)\| \le \frac{1}{(2\pi)^{1/4}}\left(\frac{96}{\gamma}e^{-\frac{1}{16\gamma^2h^2}} + \frac{4h}{\gamma^2}e^{-\frac{2\pi^2\gamma^2}{h^2}}\right),\\
        I_3 &= \left\|h\sum_{|\omega_j|\le 2} \omega_j g^2(\omega_j-O) - O\right\|.
    \end{align*}
    We require $\gamma \le \frac{1}{2\sqrt{2}\pi} < 1$, so 
    \[I_1 + I_2 \le \frac{1213}{\gamma}e^{-\frac{1}{16\gamma^2h^2}} + \frac{5h}{\gamma^3}e^{-\frac{4\pi^2\gamma^2}{h^2}}.\]
    Without loss of generality, we assume $h \le 1$. So it suffices to choose 
    \begin{align}
        &\gamma^{-1} = \cO(\log^{1/2} \epsilon^{-1}),\\
        & h^{-1} = \max\left(4\gamma \log^{1/2}\left(\frac{4852}{\gamma\epsilon}\right), \frac{1}{2\pi\gamma}\log^{1/2}\left(\frac{20}{\gamma^2\epsilon}\right)\right) = \cO(\gamma^{-1}\log^{1/2}(\gamma^{-1},\epsilon^{-1})),
    \end{align}
    to ensure that $I_1+I_2 \le \epsilon/2$. 
    By~Lemma~\ref{lem:gaussian-truncation-error}, we can make $I_3 \le\epsilon/2$ in the same parameter regime. 
    This proves the first part of the theorem. 

    \item Similar to part 1, by the triangle inequality, we have
    \begin{equation}
        \|\sum_{|\omega_j|\le 2} \omega^2_j O^\dagger_j O_j - (O^2+\gamma^2I)\| \le J_1 + J_2 + J_3,\label{eq:158}
    \end{equation}
    where
    \begin{align*}
        J_1 &= \sum_{|\omega_j|\le 2} \omega^2_j\|O_j-\sqrt{h}g(\omega_j-O)\|^2 \le \frac{2304}{\gamma} e^{-\frac{1}{8\gamma^2h^2}} + \frac{4h^2}{\gamma^{3}}e^{-\frac{4\pi^2\gamma^2}{h^2}},\\
        J_2 &= 2\sqrt{h} \sum_{|\omega_j|\le 2} \omega^2_j \|O_j - \sqrt{h}g(\omega_j-O)\| \|g(\omega_j - O)\| \le \frac{1}{(2\pi)^{1/4}}\left(\frac{192}{\gamma}e^{-\frac{1}{16\gamma^2h^2}} + \frac{8h}{\gamma^2}e^{-\frac{2\pi^2\gamma^2}{h^2}}\right),\\
        J_3 &= \left\|h\sum_{|\omega_j|\le 2} \omega^2_j g^2(\omega_j-O) - (O^2+\gamma^2 I)\right\|.
    \end{align*}
    Similarly, we can make $J_1+J_2 \le \gamma^2$ by the specified parameters.
    choosing $h^{-1} = \gamma^{-1}\cdot\polylog(\gamma^{-1})$. 
    By~Lemma~\ref{lem:gaussian-truncation-error}, we can make $J_3 \le \gamma^2$ by the choice of parameters.
    Combining these two parts together, we prove that
    \[\left\|\sum_{|\omega_j|\le 2} \omega^2_j O^\dagger_j O_j - (O^2+\gamma^2I)\right\| \le 2\gamma^2,\]
    which immediately implies~\cref{eqn:gqpe-main-1-eq-2}. This concludes the proof of part 2 of the theorem.
    \end{enumerate}
\end{proof}

\subsubsection{Covariance}

\begin{proposition}[Covariance]\label{prop:gqpe-main-2}Let $O$ be an observable with $\|O\|\leq 1$.
   For any quantum state $\rho$,  let $Z$ be the measurement outcome from $\mathcal{M}$ on $\rho$. Then, measuring the output state again using another $\mathcal{M}$, we obtain an outcome $Z'$. We have that,
    \begin{align}
        |\bCov(Z,Z')| \leq \bV[Z]. 
    \end{align}
\end{proposition}

\begin{proof}[of~\cref{prop:gqpe-main-2}]
     By definition, 
    \begin{align*}
        \bE[ZZ'] = \sum_{j,k} \tilde{\omega}_j \tilde{\omega}_k \mathbb{P}[Z=\omega_j,Z'=\omega_k],
    \end{align*}
    where $\tilde{\omega}_j = \omega_j$ if $|\omega_j| \le 2$; otherwise $\tilde{\omega}_j=0$.
    By~\cref{lem:Oj}, all $O_j$ commute with each other and $O_j=O_j^\dagger$. Suppose $\ket{\psi}$ is a common eigenvector of  $O_j^2$, w.r.t. eigenvalue $p_j$. Note that $\sum_j p_j=1$ since $\sum_j O_j^2 =I$. Using Cauchy inequality we have that
\begin{align}
    \left(\sum_j (\tilde{\omega}_j-\bE[Z]) p_j\right)^2\leq  \sum_j (\tilde{\omega}_j-\bE[Z])^2 p_j \quad
\Rightarrow \quad
       \left( \sum_j (\tilde{\omega}_j-\bE[Z]) O_j^2 \right)^2 \leq  \sum_j (\tilde{\omega}_j-\bE[Z])^2 O_j^2. \label{eq:1700}
    \end{align}
  Since all $O_j$ commute with each other and $O_j=O_j^\dagger$, using $\Tr(A\rho B)= \Tr(BA \rho)$, one can check that
    \begin{align}
        \bCov(ZZ') &=    \Tr\left[\left(\sum_j \left(\tilde{\omega}_j -\bE[Z]\right) O^2_j \right) \left(\sum_k \left(\tilde{\omega}_k -\bE[Z]\right) O_k^2\right)\rho\right] \\
        &=    \Tr\left[\left(\sum_j \left(\tilde{\omega}_j -\bE[Z]\right) O^2_j \right)^2 \rho\right] \label{eq:1711} \\
         &\leq     \Tr\left[\left(\sum_j \left(\tilde{\omega}_j -\bE[Z]\right)^2  O_j^2 \right)\rho\right] \\
         &= \bV[Z].
    \end{align}
    where in the inequality we use Eq.~(\ref{eq:1700}).
    Finally, notice that from Eq.~(\ref{eq:1711}) we have that $\bCov(ZZ')\geq 0$. Thus from above equations we have proved $|\bCov(ZZ')|\leq \bV[Z]$.
\end{proof}

\begin{lemma}\label{lem:Oj}
$\{O_j\}_j$ defined as in Eq.~(\ref{eq:MOj}) commute with each other since each $O_j$ is a function of the same Hermitian operator $O$. Besides, $O_j$ is Hermitian.
\end{lemma}
\begin{proof}
$O_j$ is Hermitian since $\widehat{g}$ is real and symmetric. More specifically,
\begin{align*}
 \bO_j^\dagger = \frac{h^{3/2}}{C^{[N]}_h} \sum^{N/2-1}_{k=1-N/2} e^{-2\pi i kh (\omega_j - \tilde{O})} \,\widehat{g}(kh) =\frac{h^{3/2}}{C^{[N]}_h}
\sum^{N/2-1}_{k=1-N/2} e^{2\pi i kh (\omega_j - \tilde{O})} \,\widehat{g}(- kh)
 = O_j,
\end{align*}
  
Here, the second equality follows by relabeling $-k$ as $k$, and the last equality uses the symmetry $\widehat{g}(x)=\widehat{g}(-x)$.
\end{proof}

\subsubsection{Commutativity}

In this section, we prove the commutativity property of GQPE in Theorem \ref{thm:GQPE}, where we assume $\|O\|\leq \kappa=1$.
To ease notation, define 
\begin{align}
	\delta(\gamma,h) = \frac{12h^{1/2}}{\gamma^{1/2}}e^{-\frac{1}{16\gamma^2h^2}} + \frac{h^{3/2}}{2\gamma^{3/2}}e^{-\frac{2\pi^2\gamma^2}{h^2}}.
\end{align}
Recall that the parameters $\gamma$ and $h$ are chosen  such that  terms below scale as $\epsilon$,

$$\left\{e^{-\mathcal{O}(\frac{1}{ \gamma^2})},  e^{- \mathcal{O}(\frac{1}{\gamma^2 h^2})} \text{ and } e^{-\mathcal{O}(\frac{\gamma^2}{h^2})}\right\} \times \poly\left(\gamma,h, h^{-1}, \gamma^{-1}\right).$$

\begin{proposition}[Commutativity property]\label{prop:com} Suppose $\|O\|\leq 1.$  For the parameter choices of $\gamma$ and $h$ given in~\cref{eqn:param_choice_exp_var},  the quantum channel $\cM$ implemented by Algorithm \ref{alg:GQPE}
approximately fixes $\rho$ when $O$ and $\rho$ approximately commute, in the sense that
    \begin{align}
        \left\| \cM[\rho] - \rho\right\|_1 \leq \epsilon^{-1} \polylog(\epsilon^{-1})\, \left\| \left[ O,\rho \right] \right\|_1 + 5\epsilon.
    \end{align}
\end{proposition}

First, we bound the leakage probability,

\begin{lemma}[Leakage probability]\label{lem:leakage} 
Let $O$ be an observable with $\|O\|\leq 1$.
Given any quantum state $\rho$ and an $\epsilon>0$, 
for the parameter choices of $\gamma$ and $h$ given in~\cref{eqn:param_choice_exp_var},
the probability of observing a value outside $[-2,2]$ in the readout regiters is bounded by $2\epsilon$. That is, 
 \begin{align}
     \left\| \cM(\rho) - \sum_{j:|\omega_j|\leq 2} O_j \rho O_j^\dagger \right\|_1 \leq 2\epsilon. 
 \end{align}
\end{lemma}

\begin{proof}[of Lemma \ref{lem:leakage}]
Recall that $O_j$ is Hermitian.
From Lemma \ref{lem:Owg} we have for any $w_j\in[-2,2]$,
 \begin{align}
     \|O_j -  \sqrt{h} g_{\gamma}(w_j-O)\|\leq \delta(\gamma,h):=\delta.
 \end{align}

Recall that in Algorithm \ref{alg:GQPE}, $j$ takes values as $0,1,...,N-1$, thus
the size of the set $|\{j: |w_j|\le 2\}|\leq N$.
We have that  
we have that 
\begin{align}
	&\left\| \sum_{j: |w_j|\leq 2} O^2_j -\sum_{j: |w_j|\leq 2} h g_{\gamma}^2(w_j -O) \right\|  \\&\leq \left\| \sum_{j: |w_j|\leq 2} O_j \left(O_j - \sqrt{h} g_{\gamma}\left(w_j -O\right)\right) \right\|  + \left\| \sum_{j: |w_j|\leq 2}\left(O_j - \sqrt{h} g_{\gamma}(w_j -O)\right)  \sqrt{h} g_{\gamma}\left(\omega_j - O\right)\right\| \label{eq:MOj_1}\\
	&\leq N \delta + N \delta \sqrt{h}\frac{1}{\sqrt{\gamma\sqrt{2\pi}}}\\
	&\leq \epsilon.
\end{align}
 Where in Eq.~(\ref{eq:MOj_1}) we use $\|O_j\|\leq 1$ since $O_j$ is a Kraus operator of a quantum channel.
 Recall that $\|O\|\leq 1$. Note that
 for any scalar $E\in [-1,1]$, we have that
\begin{align}
	\sum_{j:|w_j|\leq 2} h g_{\gamma}^2(w_j-E) \geq  \sum_{j: w_j\in [-1,1]} h g_{\gamma}^2(w_j)  =   \sum_{k: k=-1/h}^{1/h}  h g_{\gamma}^2(kh),
\end{align}
By the Fast convergent Riemann sum, i.e., Lemma \ref{lem:riemann-sum-error} and Lemma \ref{lem:gaussian-tail}, we have that
\begin{align}
	\left|\sum_{k: k=-1/h}^{1/h}  h g_{\gamma}^2(kh) - 1\right| \leq 4 (2\sqrt{2\pi}\gamma)^{\frac{1}{2}}  e^{-4\pi^2 \gamma^2 /h^2} + \frac{1}{\sqrt{\gamma\sqrt{2\pi}}} (h+\gamma^2) e^{-\frac{1}{2\gamma^2}} \leq \epsilon
\end{align}

Thus
	\begin{align}
        \left\| \cM(\rho) - \sum_{j: |\omega_j|\leq 2} O_j \rho O_j^\dagger \right\|_1  = \left\| \sum_{j: |w_j|>2} O_j \rho O_j^\dagger \right\|_1
        = \Tr \left( \sum_{j: |w_j|>2} O_j \rho O_j^\dagger \right)
		\leq \left\| \sum_{j: |w_j| \leq 2 } O_j^2 -I \right\|
		\leq 2\epsilon
	\end{align}
\end{proof}

\begin{lemma}\label{lem:exp_commutator}
For any negative semidefinite Hermitian matrix $A$ and any
square matrix $B$ of the same size, we have that
\begin{align}
    \left\|\left[e^A,B\right]\right\|_1
    \leq \left\|\left[A,B\right]\right\|_1.
\end{align}
\end{lemma}
\begin{proof}
One can check that
\begin{align}
    \frac{\d}{\d s}\left(e^{(1-s)A}Be^{sA}\right)
    =-e^{(1-s)A}[A,B]e^{sA}.
\end{align}
Integrating from $0$ to $1$, we obtain
\begin{align}
    [e^A,B]
    =\int_0^1 e^{(1-s)A}[A,B]e^{sA}\,\d s.
\end{align}
Since $A\preceq0$, we have $\|e^{tA}\|\leq1$ for every $t\geq0$.
Thus, using $\|MXN\|_1\leq\|M\|\|X\|_1\|N\|$, we have that
\begin{align}
    \left\|[e^A,B]\right\|_1
    \leq \int_0^1
    \|e^{(1-s)A}\|\,\|[A,B]\|_1\,\|e^{sA}\|\,\d s \leq \|[A,B]\|_1.
\end{align}
\end{proof}

Now, we are ready to prove~\cref{prop:com}.

\begin{proof}[of the commutativity property] 
		Recall that $O_j$ is Hermitian.  Recall that in Algorithm \ref{alg:GQPE}, $j$ takes values as $0,1,...,N-1$, thus
the size of the set $|\{j: |w_j|\le 2\}|\leq N$.
By Lemma \ref{lem:leakage} we have that  	
  $\|\sum_{j:|w_j|\leq2} O_j^2  -I\|\leq 2\epsilon$, thus
\begin{align}
	\left\|\cM(\rho)- \rho\right\|_1 &=\left\| \sum_{j: |w_j|\leq 2} [O_j,  \rho] O_j\right\|_1 + 4 \epsilon \\
	&\leq N  \max_{j: |w_j|\leq 2} \left\|  [O_j,  \rho] \right\|_1 + 4\epsilon 
\end{align}
where in the last inequality we use the fact that the number of $j$ is bounded by $N$, 
and the norm $\|O_j\|$ is bounded by $1$ since $O_j$ is a Kraus operator of a quantum channel.

By Lemma \ref{lem:Owg}, we have that
  for any $j$ such that $w_j\in [-2,2]$,
	\begin{align}
	&\left\|O_j - \sqrt{h}g_{\gamma}(\omega_j-O)\right\|\leq \delta,
	\end{align}

For $|w_j|\le 2$, we have that $\|w_j  - O\| \leq 3$ since $\|O\|\leq 1$. 
Let $A_j\coloneqq -\frac{(\omega_j-O)^2}{4\gamma^2}$.  Since $O$ is Hermitian, we have $A_j\preceq0$. Moreover,
\begin{align}
\|[A_j,\rho]\|_1 &= \frac{1}{4\gamma^2}\,\|[(\omega_j-O)^2,\rho]\|_1 \\
&= \frac{1}{4\gamma^2}\,\| -2\omega_j[O,\rho] + [O^2,\rho]\|_1 \\
&\le \frac{1}{4\gamma^2}\Big(2|\omega_j|\,\|[O,\rho]\|_1 + \|O[O,\rho]+[O,\rho]O\|_1\Big) \\
&\le \frac{1}{4\gamma^2}\Big(2|\omega_j|\,\|[O,\rho]\|_1 + 2\|O\|\,\|[O,\rho]\|_1\Big)
\le \frac{3}{2\gamma^2}\,\|[O,\rho]\|_1.
\end{align}
Thus by Lemma \ref{lem:exp_commutator} we have that

\begin{align}
\left\|\cM(\rho)-\rho\right\|_1	& \leq 4\epsilon+  N \sqrt{h}  \frac{1}{\sqrt{\gamma \sqrt{2\pi}} }\, \max_{j: |w_j|\leq 2} \left\|\,  \left[\exp\left(-\frac{(\omega_j -O)^2}{4\gamma^2}\right),  \rho \right] \,\right\|_1  + N 2 \delta \|\rho\|_1 \\
	&\leq 4\epsilon+ N \sqrt{h}  \frac{3}{2\gamma^2\sqrt{\gamma \sqrt{2\pi}} }  \left\| \left[ O,\rho \right] \right\|_1  + N 2\delta\\
	&\leq \epsilon^{-1} \polylog (\epsilon^{-1})  \left\| \left[ O,\rho \right] \right\|_1 + 5\epsilon.
\end{align}
where we use the parameter choices specified at the beginning
of this subsection,
$\gamma^{-1}=\cO\left(\log^{1/2}(\epsilon^{-1})\right)$ and
$h^{-1}=\cO\left(\gamma^{-1}\log^{1/2}(\gamma^2/\epsilon)\right)$,
together with $N=h^{-2}$, and choose the constants in the
parameter prescription sufficiently large so that
$2N\delta\leq\epsilon$.
\end{proof}

\subsubsection{Technical lemmas for numerical analysis}

\begin{lemma}[Fast convergent Riemann sum]\label{lem:riemann-sum-error}
    Let $\varphi(x)$ be an integrable function (i.e., $\varphi \in L^1(\mathbb{R})$) with Fourier transform $\widehat{\varphi}(y) \coloneqq \int^\infty_{-\infty} e^{-2\pi i xy}\varphi(x)\d x$.
    Then, for $N \in \mathbb{Z}^+$ and $h > 0$, we have
    \begin{equation}\label{eqn:riemann-sum-error}
        \left|h\sum^{N/2-1}_{k=-N/2}\varphi(kh) - \int^\infty_{-\infty}\varphi(x)~\d x\right| \le \sum^{\infty}_{n=-\infty,n\neq 0} |\widehat{\varphi}(n/h)| + h \sum_{|k| \ge N/2}|\varphi(kh)|.
    \end{equation}
\end{lemma}
\begin{proof}
    First, by the Poisson summation formula~\cite[Theorem~4.2.2]{pinsky2023introduction}, we have
    \begin{equation}
        h\sum_{k\in \mathbb{Z}} \varphi(kh) = \sum_{n \in \mathbb{Z}} \widehat{\varphi}(n/h).
    \end{equation}
    Note that $\widehat{\varphi}(0) = \int^\infty_{-\infty}\varphi(x)~\d x$. \cref{eqn:riemann-sum-error} follows from the triangle inequality.
\end{proof}

\begin{lemma}[Gaussian sum]\label{lem:gaussian-tail}
    Let $N \in \mathbb{Z}^+$ be an even integer and $A>0$ be a positive scalar. Then, we have
    \begin{equation}
        \sum_{k \ge N/2}e^{-Ak^2} \le \left(1+\frac{1}{AN}\right) e^{-AN^2/4}.
    \end{equation}
\end{lemma}
\begin{proof}
    This summation can be upper-bounded by a Gaussian integral:
    \begin{align*}
        \sum_{k \ge N/2}e^{-Ak^2} &= e^{-AN^2/4} + \sum^\infty_{k = N/2+1} e^{-Ak^2} \le e^{-AN^2/4} + \int^\infty_{N/2} e^{-Ax^2}~\d x\\ 
        &\le e^{-AN^2/4} + \frac{1}{N/2}\int^\infty_{N/2} x e^{-Ax^2}~\d x \le \left(1 + \frac{1}{AN}\right)e^{-AN^2/4}.
    \end{align*}
\end{proof}

\section{Calculations for the easy-axis mean-field Heisenberg Hamiltonian}\label{app:reduction_to_classical}

\subsection{Relation between projective energy measurement and GQPE}\label{sec:proj_GQPE}

For simplicity, we consider the idealized GQPE; as described in Section \ref{sec:measure}, its discretized implementation approximates it with arbitrarily small error.

Let $\mathcal{P}$ denote the perfect projective measurement of
$H_\alpha$.  Following the notation in the main text, for $\sigma>0$, the idealized
GQPE measurement $\mathcal{M}_*$ for $O=H_\alpha$ has Kraus operators
\begin{align}
 O_{w*}
 :=
 \frac{1}{\sqrt{\sigma\sqrt{2\pi}}}
 \exp\left[-\frac{(w-H_\alpha)^2}{4\sigma^2}\right].
\end{align}
Since $H_\alpha\Pi_{s,m}=E_{s,m}\Pi_{s,m}$, we have
\begin{align}
 O_{w*}\Pi_{s,m}
 =
 \frac{1}{\sqrt{\sigma\sqrt{2\pi}}}
 \exp\left[-\frac{(w-E_{s,m})^2}{4\sigma^2}\right]\Pi_{s,m}.
\end{align}

Thus, if the input state is a sector state $\Pi_{s,m}/d_s$,  then even conditioned on the GQPE measurement outcome $w$, the normalized post-measurement state remains $\Pi_{s,m}/d_s$; the measurement outcome is a Gaussian random variable with mean
$E_{s,m}$ and variance $\sigma^2$. In particular, in Algorithm \ref{alg:intro_main} we perform 
$$\cE_*:=\cM_*\circ \cN \circ \cM_*,$$  the leftmost (unrecorded) GQPE does not change the sector state produced by $\cN$. Hence the underlying $(s, m)$ trajectory is still exactly the classical chain $Q$. 

Let $e_j$ denote the outcome of the projective energy measurement at
the $j$-th measurement time and let $w_j$ denote the corresponding
GQPE outcome in the $j$-th application of $\cE_*$.  The above observation allows us to couple the two measurement procedures with
\begin{align}
 w_j=e_j+\xi_j,
\end{align}
where the $\{\xi_j\}_j$ are independent Gaussian random variables
with mean $0$ and variance $\sigma^2$, and are independent of $\{e_j\}_j$.  Consequently, for $i\neq j$,
\begin{align}
 \operatorname{Cov}(w_i,w_j)
 =\operatorname{Cov}(e_i,e_j),\qquad
 \operatorname{Var}(w_i)
 =\mathbb{V}\mathrm{ar}_{\rho_{\beta H_\alpha}}(H_\alpha)
   +\sigma^2.
\end{align}

Thus for every pair of distinct measurement times, the normalized autocorrelation function  (Eq.~(\ref{eq:normalized_var}))
of the GQPE outcomes $C_{\rho_{\beta H_\alpha}}^{\mathcal{M}_*}(t)$ is related to the  normalized autocorrelation function for projective measurement $C_{\rho_{\beta H_\alpha}}^{\mathcal{P}}(t)$ by
\begin{align}
 C_{\rho_{\beta H_\alpha}}^{\mathcal{M}_*}(t)
 =
 \frac{\operatorname{Cov}(w_j,w_{j+t+1})}
 {\mathbb{V}\mathrm{ar}_{\rho_{\beta H_\alpha}}(H_\alpha)+\sigma^2}
 =
 \frac{\mathbb{V}\mathrm{ar}_{\rho_{\beta H_\alpha}}(H_\alpha)}
 {\mathbb{V}\mathrm{ar}_{\rho_{\beta H_\alpha}}(H_\alpha)+\sigma^2}
 C_{\rho_{\beta H_\alpha}}^{\mathcal{P}}(t).
\end{align}

Substituting these identities into the definition of the
autocorrelation time (Eq.~(\ref{eq:t_auk})) gives
\begin{align}
 t_{aut,K}^{\mathcal{M}_*}
 =
 \frac{1}{2}
 +
 \frac{\mathbb{V}\mathrm{ar}_{\rho_{\beta H_\alpha}}(H_\alpha)}
 {\mathbb{V}\mathrm{ar}_{\rho_{\beta H_\alpha}}(H_\alpha)+\sigma^2}
 \left(
 t_{aut,K}^{\mathcal{P}}-\frac{1}{2}
 \right).
 \label{eq:gqpe_projective_aut}
\end{align}
Therefore, a polynomial upper bound on
$t_{aut,K}^{\mathcal{P}}$ gives a polynomial upper bound on
$t_{aut,K}^{\mathcal{M}_*}$.  For simplicity, in the proof  we consider the
perfect projective measurement $\mathcal{P}$.

\subsection{Reduction to classical chain}

Let $\pi$ denote a permutation of the qubit index  $\{1,...,n\}$, and let $U_{\pi}$ be the corresponding unitary that permutes the tensors, where for $a\in \{x,y,z\}$, 
		$U_{\pi}\sigma^a_{i} U^{-1}_{\pi} = \sigma^a_{\pi(i)}.$
We first prove a structure theorem for operators that satisfies two symmetries. 
\begin{lemma}\label{lem:structure}
 Let $\cB$ be the set of Hermitian operators that commutes with $U_{\pi}$ for any $\pi$, and commutes with $S_z$. Then for any operator $B\in \cB$, it  must be of the form
\begin{align}
	B=\sum_{s,m} b_{s,m} \Pi_{s,m}, \qquad b_{s,m}\in \bR.
\end{align}
\end{lemma}
\begin{proof}
Fix a value $m$.
Define $\cH_m$ as the space where the magnetization is $m$, that is
\begin{align}
	\cH_m:=span \{\ket{s,m,r}: \,\,s=|m|,|m|+1,...,N,\,\, r=1,..,d_s\}.
\end{align}
Since $B$ commutes with the magnetization operator $S_z$, we have $B$ fixes $\cH_m$, thus $B|_{\cH_m}$ and $\cB|_{\cH_m}$ are well-defined. 
Note that if we use $\ket{x}\in\{0,1\}^n$ as the normal computational basis, let $k$ denote the number of 1s in $x$, 
then $S_z \ket{x}= \frac{1}{2}(n-2k)  \ket{x}=  (N-k)\ket{x}$. Thus the space $\cH_m$ is equal to
 \begin{align}
 	\cH_m:=span \{\ket{x}: \text{ $x$ contains $k=(N-m)$ entries equal to $1$} \}.\label{eq:number}
 \end{align}

Note that $U_{\pi}$ also fixes $\cH_m$, thus $U_{\pi}|_{\cH_m}$ is well-defined.  Besides, the assumption that $U_{\pi}$ commutes with $B$ implies that  $U_{\pi}|_{\cH_m}$ commutes with $B|_{\cH_m}$. 
Thus for any $\ket{x},\ket{y}\in \cH_{m}$, we have 
\begin{align}
	\langle x| B|y\rangle = \langle \pi(x) | B | \pi(y)\rangle, \label{eq:dimension}
\end{align}
where both $x$ and $y$ have $k=N-m$ entries equal to $1$.
Note that the pairs $(x,y)$ and  $(x',y')$ can be related by a permutation iff $|x\cap y| = |x'\cap y'|$, where here $|x\cap y|$ denotes the number of sites that both $x$ and $y$ take values as $1$, i.e. $|\{j, x_j=y_j=1\}|$. 

By Eq.~(\ref{eq:number}),
the possible values of $|x\cap y|$ satisfies 
\begin{align}
    \max\{0,2k-n\} \leq |x\cap y| \leq k.
\end{align}
and hence the number of possible values is 
$$k- \max\{0,2k-n\}+1= \min\{k,n-k\}+1 = N- |m|+1.$$

 Thus Eq.~(\ref{eq:dimension}) implies the linear space $\cB|_{\cH_m}$ has dimension at most $N-|m|+1$. On the other hand,  since $\{\Pi_{s,m}: s=|m|, |m|+1, ...,N \}$  consists of $N-|m|+1$ linear independent operators in the space $\cB|_{\cH_m}$. It is therefore a basis, which proves the lemma.
\end{proof}

Then we are ready to prove Lemma \ref{lem:classical_chain}.

\begin{proof}[Proof of Lemma \ref{lem:classical_chain}]
We first show that the operator $\cD(\Pi_{s,m})$ satisfies two symmetry: 

\begin{itemize}
	\item Let $\pi$ denote a permutation of the qubit index  $\{1,...,n\}$, and let $U_{\pi}$ be the corresponding unitary that permutes the tensors, where for $a\in \{x,y,z\}$, 
		$U_{\pi}\sigma^a_{i} U^{-1}_{\pi} = \sigma^a_{\pi(i)}.$
		
	Thus $U_{\pi}$ commutes with $\bbS^2, S_z$, thus $U_{\pi}$ also commutes with $\Pi_{s,m}$ and also $\Pi_E$.
	Use the identity $U_\pi A B CU^{-1}_{\pi} = U_\pi AU^{-1}_{\pi} U_\pi BU^{-1}_{\pi} U_\pi CU^{-1}_{\pi}$,
	 one can check that
	\begin{align}
		U_{\pi} \cD(\Pi_{s,m})U^{-1}_{\pi}  =  \cD(\Pi_{s,m}) \qquad \Rightarrow \qquad  [U_{\pi},  \cD(\Pi_{s,m})]=0.\label{eq:com_Upi}
	\end{align} 
	In particular, note that $\bbS^2$ can be written as a sum of permutation unitary. Denote $U_{(i,j)}$ as the swap gate between qubit $i$ and $j$, we have 
		$\bbS^2 = \sum_{i<j} U_{(ij)} - \frac{n(n-4)}{4}I$, thus  from Eq.(\ref{eq:com_Upi}), we have
		\begin{align}
			[\bbS^2, \cD(\Pi_{s,m})]=0. \label{eq:com_S2}
		\end{align}
	\item Let $R_\theta =e^{-i\theta S_z}$. 
	Note that $[S_z,\bbS^2]=0$, thus we have $[R_\theta, \bbS^2]=0$. Since in addition $[R_\theta,S_z]=0$, thus  $R_\theta$ commutes $\Pi_{s,m}$, so does $\Pi_E$. Note that
	\begin{align}
		&R_\theta\, \sigma_i^x (\omega)\,  R_\theta^{-1} = \cos \theta \sigma_i^x(\omega) + \sin \theta\sigma_i^y(\omega),  \\
		&R_\theta\, \sigma_i^y (\omega)\,  R_\theta^{-1} = -\sin \theta \sigma_i^x(\omega) + \cos \theta\sigma_i^y(\omega), \\
		&R_\theta\, \sigma_i^z (\omega)\,  R_\theta^{-1} = \sigma_i^z (\omega).
	\end{align}
	Since $\sin^2 \theta + \cos^2 \theta=1$,  one can check that 
	\begin{align}
		R_\theta \cD(\Pi_{s,m}) R_\theta^{-1} = \cD(\Pi_{s,m}) \qquad \Rightarrow \qquad [R_\theta, \cD(\Pi_{s,m})]=0, \qquad \Rightarrow \qquad [S_z, \cD(\Pi_{s,m})]=0. \label{eq:com_Rtheta}
	\end{align}
\end{itemize}

 Since $\cD(\Pi_{s,m})$ is Hermitian, Eqs.~\eqref{eq:com_Upi}
 and~\eqref{eq:com_Rtheta}, together with Lemma~\ref{lem:structure}, imply that
 
\begin{align}
\cD(\Pi_{s,m})= \sum_{s',m'} b^{(s,m)}_{s',m'} \Pi_{s',m'},
\qquad b^{(s,m)}_{s',m'}\in \bR.\label{eq:bsm}
\end{align}

Finally, note that by the form of $\cD(\cdot)$, we have 
$\operatorname{tr}\!\left(\cD(\Pi_{s,m})\right)=0$ which implies that 
$0 =\sum_{s',m'}b^{(s,m)}_{s',m'}d_{s'}$. Thus we can rewrite Eq.~(\ref{eq:bsm}) as
\begin{align}
	&\cD\left(\frac{\Pi_{s,m}}{d_s}\right) = \sum_{s',m'} Q_{(s,m)\rightarrow (s',m')} \frac{\Pi_{s',m'}}{d_{s'}}, \qquad  Q_{(s,m)\rightarrow (s',m')}\in \bR,\\
	&\sum_{s',m'}  Q_{(s,m)\rightarrow (s',m')} =0.
\end{align}

Moreover, for $(s,m)\neq (s',m')$, we have
\begin{align}
	Q_{(s,m)\rightarrow (s',m')}  =  tr\left( \Pi_{s',m'} \cD\left(\frac{\Pi_{s,m}}{d_s}\right)\right) =  \sum_{i=1}^n \sum_{a=x,y,z} \frac{\gamma(E_{s',m'}-E_{s,m})}{d_s} tr\left( \Pi_{s',m'} \sigma_i^a   \Pi_{s,m} \sigma_i^a  \right) \geq 0,
\end{align}
where the $\geq 0$ comes from the fact that both $\Pi_{s',m'} \geq 0$ and $\sigma_i^a   \Pi_{s,m} \sigma_i^a \geq 0$.

Finally,  $Q_{(s,m)\rightarrow (s',m')}=0$ unless $|m'-m|\leq 1$ since  $\sigma_i^a$ only acts on the i-th site thus only change magnetization at most by $1$. 
\end{proof}

\subsection{Estimates the weights in sectors}\label{app:sector_weights}

\begin{proof}[Proof of Lemma~\ref{lem:sector_weights}]
Let \(h(p)=-p\ln p-(1-p)\ln(1-p)\).  Since
 $h\left(\frac{1-z}{2}\right)
 =\ln2-\frac{z^2}{2}+O(z^4),$
the assumption \(\beta J\alpha>2\) allows us to fix a sufficiently small
constant \(\eta>0\) such that
\begin{align}
 \delta :=
 h\left(\frac{1-\eta}{2}\right)
 +\frac{\beta J\alpha}{4}\eta^2-\ln2>0.\label{eq:delta}
\end{align}
Set \(s_*=\lfloor N\eta\rfloor\) .  The bounds
 $\frac1{n+1}e^{nh(k/n)}\le\binom nk$
and Eq.~(\ref{eq:ds}) give
\begin{align}
 Z_n\pi(s_*,\pm s_*)
 &=d_{s_*}\exp\left[\frac{\beta J}{n}
 \left(\alpha s_*^2+s_*\right)\right]\\
 &\ge c n^{-2}
 \exp\left\{n\left[
 h\left(\frac{1-\eta}{2}\right)
 +\frac{\beta J\alpha}{4}\eta^2\right]\right\}
 =c n^{-2}2^ne^{\delta n}.
\end{align}

For \(m\in\{-1,0\}\), put \(y=s/N\). Since $\binom{n}{k} \leq e^{n h(k/n)}$,  we have
\begin{align}
 Z_n\pi(s,m)
 \le C\exp\left\{n\left[
 h\left(\frac{1-y}{2}\right)+\frac{\beta J}{4}y^2
 \right]\right\}
\end{align}
The expression in brackets has derivative
\begin{align}
	-\frac{1}{2} \ln \frac{1+y}{1-y} +  \frac{\beta J}{2} y  =  -\operatorname{arctanh}(y)+\frac{\beta J}{2}y \leq 0,
\end{align}
where the last inequality follows from \(\beta J<2\), with strict inequality for \(y>0\) and equality at \(y=0\).
It is therefore at most its value \(\ln2\) at
\(y=0\) which is $C 2^n$.  Summing over the at most \(N+1\) values of \(s\) proves
\(Z_n\pi(m=-1),Z_n\pi(m=0)\le Cn2^n\).
\end{proof}

\subsection{Edges in the graph for $Q^{fold}$}

\begin{lemma}[Vertical edges]
\label{lem:vertical-edges}
There is a constant $c>0$ independent of $n$ such that
for every \(0\le u<s\),
\begin{align}
 Q^{fold}_{(s,u)\to(s,u+1)}\geq \frac{c}{n}, \qquad  Q^{fold}_{(s,u+1)\to (s,u)}\geq \frac{c}{n}.
\end{align}
\end{lemma}

\begin{proof}
Define
\(\sigma_i^\pm=(\sigma_i^x\pm i\sigma_i^y)/2\) and
\(S_+=\sum_i\sigma_i^+\).  Since \(\sigma_i^-\) lowers the
magnetization, we have
 $\Pi_{s,u+1}\sigma_i^-\Pi_{s,u}=0.$ Use $\sigma_i^x = \sigma_i^+ + \sigma_i^-$ and  $\sigma_i^y = (-i) (\sigma_i^+ - \sigma_i^-$), we have
\begin{align}
tr\!\left(\Pi_{s,u+1}\sigma_i^x\Pi_{s,u}\sigma_i^x\right) +
  tr\!\left(
\Pi_{s,u+1}\sigma_i^y\Pi_{s,u}\sigma_i^y\right) 
 =2 tr\!\left(
\Pi_{s,u+1}\sigma_i^+\Pi_{s,u}\sigma_i^-\right).
 \label{eq:xy-to-pm}
\end{align}
Thus, directly from the definition of \(Q\), i.e. Eq.~(\ref{eq:Q}), retaining only the
\(\sigma_i^x\) and \(\sigma_i^y\) contributions,
\begin{align}
 Q^{fold}_{(s,u)\to(s,u+1)}
 &\ge
 \frac{2e^{-2\beta J\alpha}}{d_s}
 \sum_{i=1}^n
  tr\!\left(
 \Pi_{s,u+1}\sigma_i^+\Pi_{s,u}\sigma_i^-\right)\notag\\
 &\ge
 \frac{2e^{-2\beta J\alpha}}{nd_s}
  tr\!\left(
 \Pi_{s,u+1}S_+\Pi_{s,u}S_-\right)\notag\\
 &=
 \frac{2e^{-2\beta J\alpha}}{n}
 (s-u)(s+u+1)\geq \frac{2e^{-2\beta J\alpha}}{n}.
 \label{eq:vertical-edge-positive}
\end{align}
The second inequality is the Hilbert--Schmidt Cauchy--Schwarz
inequality, and the last equality follows from
 $S_+\ket{s,u,r}
 =\sqrt{(s-u)(s+u+1)}\,\ket{s,u+1,r}.$
The last inequality comes from $0\leq u<s$.

Moreover, recall that $\alpha>1$,
\begin{align}
 \frac{\pi_{fold}(s,u)}{\pi_{fold}(s,u+1)}
 &=\frac{2-\delta_{u,0}}{2}
 \exp\left[-\frac{\beta J(\alpha-1)}{n}(2u+1)\right]
 \ge \frac12e^{-\beta J(\alpha-1)}.
\end{align}
Since $Q^{fold}$ is reversible, Eq.~\eqref{eq:vertical-edge-positive}
therefore implies
 $Q^{fold}_{(s,u+1)\to(s,u)}\ge \frac{c'}{n}$
for a constant $c'>0$ independent of $n$. This proves the lemma.

\end{proof}

\begin{lemma}\label{lem:s2s1}
For every \(0\le s<N\) and every \(i\),
\begin{align}
 \Pi_{s+1,s+1}\sigma_i^+\Pi_{s,s}
 =\sigma_i^+\Pi_{s,s}.
 \label{eq:top-projection}
\end{align}
\end{lemma}

\begin{proof}
First, note that
 $S_-S_+\Pi_{s,s}
 =
 \bigl(\bbS^2-S_z(S_z+1)\bigr)\Pi_{s,s}=0.$
Since \(S_-S_+=S_+^\dagger S_+\), this implies 
 $$S_+\Pi_{s,s}=0.$$
Using
 $[S_z,\sigma_i^+]=\sigma_i^+,
 [S_+,\sigma_i^+]=0,$
we obtain
\begin{align}
 S_z\sigma_i^+\Pi_{s,s}
 =(s+1)\sigma_i^+\Pi_{s,s},\qquad
 S_+\sigma_i^+\Pi_{s,s}
 =\sigma_i^+S_+\Pi_{s,s}=0. \label{eq:108}
\end{align}
Consequently,
\begin{align}
 \bbS^2\sigma_i^+\Pi_{s,s}
 &=
 \bigl(S_-S_+ +S_z(S_z+1)\bigr)
 \sigma_i^+\Pi_{s,s}\notag\\
 &=(s+1)(s+2)\sigma_i^+\Pi_{s,s}.\label{eq:109}
\end{align}
In conclusion, the first equality in Eq.~(\ref{eq:108}) and Eq.~(\ref{eq:109}) together imply that \(\sigma_i^+\Pi_{s,s}\) has magnetization \(s+1\)  and total
spin \(s+1\),  thus we 
proved Eq.~\eqref{eq:top-projection}.
\end{proof}

\begin{lemma}[Edges between consecutive tower tops] 
\label{lem:top-edges} There is a constant $c>0$ independent of $n$ such that
for every \(0\le s<N\),
\begin{align}
Q^{fold}_{(s,s)\rightarrow(s+1,s+1)}\geq c, \qquad Q^{fold}_{(s+1,s+1)\rightarrow(s,s)}\geq c.
\end{align}
\end{lemma}

\begin{proof}
As in Eq.~\eqref{eq:xy-to-pm}, the \(\sigma_i^x\) and
\(\sigma_i^y\) terms  give
\begin{align}
 Q^{fold}_{(s,s)\to(s+1,s+1)}
 &\ge
 \frac{2e^{-2\beta J\alpha}}{d_s}
 \sum_{i=1}^n
  tr\!\left(
 \Pi_{s+1,s+1}\sigma_i^+
 \Pi_{s,s}\sigma_i^-\right)\notag\\
 &=2e^{-2\beta J\alpha}(N-s)\geq 2e^{-2\beta J\alpha},
\end{align}

where the  equality follows from Lemma \ref{lem:s2s1} and the fact that $\sum_{i=1}^n \sigma_i^-\sigma_i^+=N-S_z$, and the last inequality follows from  $s<N$ and the fact that $s,N$ are integers.

For the reverse direction, note that $2(s+1)\leq 2N=n$ and
Eq.~\eqref{eq:ds} gives
 $d_{s+1}/{d_s}
\le 3.$
Consequently,
\begin{align}
 \frac{\pi_{fold}(s+1,s+1)}{\pi_{fold}(s,s)}
 &=\frac{2}{2-\delta_{s,0}}\frac{d_{s+1}}{d_s}
 \exp\left[\frac{\beta J}{n}(2\alpha s+\alpha+1)\right]
 \le 6e^{\beta J\alpha}.
\end{align}
Detailed balance and the forward bound therefore imply
\begin{align}
 Q^{fold}_{(s+1,s+1)\to(s,s)}\ge c',
\end{align}
for a constant $c'>0$ independent of $n$.
\end{proof}

\subsection{Canonical-path bound for the folded chain}
\label{app:canonical}

This appendix proves Lemma~\ref{lem:gap_fold}.  Recall that the folded chain
acts on the state space
\begin{align}
 \Omega_{fold}:=\{(s,u):0\le u\le s\le N\},
\end{align}
with generator \(Q^{fold}\) and invariant distribution \(\pi_{fold}\) defined
in Eqs.~\eqref{eq:piven}--\eqref{eq:Qflold_ss}.  We use bold symbols
\(\ba,\bb,\bc,\bd\) to denote states in \(\Omega_{fold}\).

The generator \(Q^{fold}\) defines a graph on \(\Omega_{fold}\): two distinct
states \(\ba\) and \(\bb\) are connected by an edge if
\(Q^{fold}_{\ba\to\bb}>0\).  For an edge \(e=\{\ba,\bb\}\), define its
capacity by
\begin{align}
 c(e):=\pi_{fold}(\ba)Q^{fold}_{\ba\to\bb}= \pi_{fold}(\bb)Q^{fold}_{\bb\to\ba}.
 \label{eq:edge-capacity}
\end{align}

For every ordered pair of states \(\bc,\bd\in\Omega_{fold}\), choose a path
\(\Gamma_{\bc,\bd}\) from \(\bc\) to \(\bd\) along  edges.  Let
\(\lvert\Gamma_{\bc,\bd}\rvert\) be the length of the path, i.e.  the number of edges.  The congestion of
this family of paths is
\begin{align}
 \phi
 :=\max_e\frac{1}{c(e)}
 \sum_{\substack{\bc,\bd\in\Omega_{fold}\\
                  e\in\Gamma_{\bc,\bd}}}
 \lvert\Gamma_{\bc,\bd}\rvert\,
 \pi_{fold}(\bc)\, \pi_{fold}(\bd).
 \label{eq:congestion}
\end{align}
The canonical-path theorem \cite[Corollary~13.21]{levin2026markov}.\footnote{Levin and Peres~\cite[Corollary~13.21]{levin2026markov} is stated in
discrete time Markov chain but could be easily modified for continuous-time chain.} gives
\begin{align}
 \operatorname{gap}(Q^{fold})\ge \phi^{-1}.
 \label{eq:canonical-gap}
\end{align}

It therefore suffices to choose the paths and prove
\(\phi\le Cn^5\) for the chosen family of paths.

\paragraph{Choice of paths.}

Let \(\mathbf{c}=(s,u)\) and \(\mathbf{d}=(t,v)\).  Define
\(\Gamma_{\mathbf{c},\mathbf{d}}\) as follows.  If \(s=t\), move  from
\((s,u)\) to \((s,v)\) by changing the second coordinate by one at each
step.  If \(s\ne t\), first move from \((s,u)\) to \((s,s)\) by
increasing the second coordinate by one at each step.  Next move from
\((s,s)\) to \((t,t)\) by increasing both coordinates by one at each
step if \(s<t\), or decreasing both coordinates by one at each step if
\(s>t\).  Finally, move from \((t,t)\) to \((t,v)\) by decreasing the
second coordinate by one at each step.  Lemmas~\ref{lem:vertical-edges}
and~\ref{lem:top-edges} show that every step is along an edge.  Note that the
resulting path has length 
\begin{align}
\lvert\Gamma_{\mathbf{c},\mathbf{d}}\rvert\leq 3N \leq 2n.   \label{eq:path-length}
\end{align}

The edges used by these paths are of two types:
\begin{align}
 \text{vertical edge: }&(s,u)\longleftrightarrow(s,u+1),\qquad 0\le u<s,
 \label{eq:vertical-edge}\\
 \text{edge between tower tops }&(s,s)\longleftrightarrow(s+1,s+1),\qquad 0\le s<N.
 \label{eq:top-edge}
\end{align}

\paragraph{Rates of the edges used by the paths.} 
From Lemma \ref{lem:vertical-edges} and Lemma \ref{lem:top-edges}, we know that
there is a constant \(c>0\), independent of \(n\), such that
\begin{align}
 &Q^{fold}_{(s,u)\to(s,u+1)}\ge \frac{c}{n},
 \qquad Q^{fold}_{(s,u+1)\to(s,u)}\ge \frac{c}{n},  \qquad 0\le u<s, \label{eq:vertical-rate-bound}\\
 &Q^{fold}_{(s,s)\to(s+1,s+1)}\ge c,
  \qquad Q^{fold}_{(s+1,s+1)\to(s,s)}\ge c,
  \qquad 0\le s<N. \label{eq:top-rate-bound}
\end{align}

\paragraph{Congestion of a vertical edge.}
Fix a vertical edge
\begin{align}
 e=\{(s,u),(s,u+1)\}
\end{align}
and define $A_e$ as the following set of states 
\begin{align}
 A_e:=\{(s,0),(s,1),\ldots,(s,u)\}.
\end{align}
By the definition of the paths, \(e\in\Gamma_{\bc,\bd}\) exactly when one
of \(\bc,\bd\) lies in \(A_e\) and the other does not.  Therefore
\begin{align}
\sum_{\substack{\bc,\bd\in\Omega_{fold}\\
 e\in\Gamma_{\bc,\bd}}}
 \pi_{fold}(\bc)\pi_{fold}(\bd)
 =2\pi_{fold}(A_e)\bigl(1-\pi_{fold}(A_e)\bigr)
 \le2\pi_{fold}(A_e).                    \label{eq:pair-mass}
\end{align}

By the definition of \(\pi_{fold}\), for \(0\le j\le u<N\),
\begin{align}
 \frac{\pi_{fold}(s,j)}{\pi_{fold}(s,u)}
 &=
 \frac{2-\delta_{j,0}}{2-\delta_{u,0}}
 \exp\left[
 \frac{\beta J(\alpha-1)}{n}(j^2-u^2)
 \right]
 \le 1.
\end{align}
Therefore,
\begin{align}
 \pi_{fold}(A_e)
 =\sum_{j=0}^{u}\pi_{fold}(s,j)
 \le (u+1)\pi_{fold}(s,u)
 \le n\pi_{fold}(s,u).
 \label{eq:vertical-mass}
\end{align}

Combining Eqs.~\eqref{eq:congestion},\eqref{eq:path-length}, \eqref{eq:pair-mass},
\eqref{eq:vertical-mass}, and~\eqref{eq:vertical-rate-bound}, the congestion
at \(e\) is at most
\begin{align}
 \frac{4n\,\pi_{fold}(A_e)}
 {\pi_{fold}(s,u)Q^{fold}_{(s,u)\to(s,u+1)}}
 \le Cn^3.                                                   \label{eq:vertical-congestion}
\end{align}

\paragraph{Congestion of an edge between tower tops.}

We first locate the peak of the tower-top weights. Define
\begin{align}
 h(p)&=-p\log p-(1-p)\log(1-p),\\
 F(x)&=h\left(\frac{1-x}{2}\right)
       +\frac{\beta J\alpha}{4}x^2,\qquad 0\le x\le1.
 \label{eq:F}
\end{align}
Note that $ F'(x)
 =\frac{x}{2}\left[
 \beta J\alpha-\frac{1}{x}
 \log\left(\frac{1+x}{1-x}\right)
 \right].$
For \(0<x<1\),
the function
 $\frac{1}{x}\log\left(\frac{1+x}{1-x}\right)
 =2\sum_{k=0}^{\infty}\frac{x^{2k}}{2k+1}$
is strictly increasing from \(2\) to infinity.  Since
\(\beta J\alpha>2\), we have that $F'(x)$ reaches $0$  at a unique point
\(x_*\in(0,1)\).  Therefore \(F'(x)>0\) for \(x<x_*\) and
\(F'(x)<0\) for \(x>x_*\).  Hence \(F\) first increases and then
decreases.

Choose \(s_0\in\{0,\ldots,N\}\) maximizing \(F(s/N)\). Fix $s\in\{0,1,...,N-1\}$ and edge
\begin{align}
 e=\{(s,s),(s+1,s+1)\}.
\end{align}
Define
\begin{align}
 A_e&=\{(r,j):0\le j\le r\le s\},
 &&\text{if }s<s_0,\\
 A_e&=\{(r,j):s+1\le r\le N,\ 0\le j\le r\},
 &&\text{if }s\ge s_0.
 \label{eq:top-side}
\end{align}
A path uses edge \(e\) exactly when one endpoint belongs to \(A_e\) and
the other does not. Thus we have \begin{align}
\sum_{\substack{\bc,\bd\in\Omega_{fold}\\
 e\in\Gamma_{\bc,\bd}}}
 \pi_{fold}(\bc)\pi_{fold}(\bd)
 =2\pi_{fold}(A_e)\bigl(1-\pi_{fold}(A_e)\bigr)
 \le2\pi_{fold}(A_e).                    \label{eq:pair-mass-top}
\end{align}

Recall that $n=2N$ and $d_r
 =\frac{2r+1}{N+r+1}\binom{n}{N-r}.$ 
Since
 $\frac{1}{n+1}e^{n h((N-r)/n)}
 \le \binom{n}{N-r}
 \le e^{n h((N-r)/n)}$,  
 there are constant $c>0$ and $C>0$ such that
for any $0\leq r\leq N$,
\begin{align}
 c n^{-2}e^{nF(r/N)}
 \le
 (2-\delta_{r,0})d_r
 \exp\left[
 \frac{\beta J}{n}
 \bigl(r(r+1)+(\alpha-1)r^2\bigr)
 \right] 
 \le C e^{nF(r/N)}.
 \label{eq:top-weight}
\end{align}
The expression in the middle is the numerator of
\(\pi_{fold}(r,r)\).  Therefore, we have
\begin{enumerate}
    \item If \(s<s_0\) and \(r\le s\), there exists constant $C'$ such that
\begin{align}
 \frac{\pi_{fold}(r,r)}{\pi_{fold}(s,s)}
 &\le
 C' n^2
 \exp\left[
 n\left(F(r/N)-F(s/N)\right)
 \right]
 \le C' n^2,
 \label{eq:left-top-comparison}
\end{align}
where the last inequality follows because \(F\) is increasing to the
left of \(s_0/N\).
\item Similarly, if \(s\ge s_0\) and \(r\ge s+1\),
\begin{align}
 \frac{\pi_{fold}(r,r)}
      {\pi_{fold}(s+1,s+1)}
 &\le
 C' n^2
 \exp\left[
 n\left(F(r/N)-F((s+1)/N)\right)
 \right]
 \le C' n^2.
 \label{eq:right-top-comparison}
\end{align}
\item Moreover, for every \(0\le j\le r\),
\begin{align}
 \frac{\pi_{fold}(r,j)}{\pi_{fold}(r,r)}
 =
 \frac{2-\delta_{j,0}}{2-\delta_{r,0}}
 \exp\left[
 \frac{\beta J(\alpha-1)}{n}(j^2-r^2)
 \right]
 \le1.
\end{align}
\end{enumerate}

Since \(A_e\) contains at most \(n^2\) states, it follows that
\begin{align}
 \pi_{fold}(A_e)
 &\le Cn^4\pi_{fold}(s,s),
 &&\text{if }s<s_0,\\
 \pi_{fold}(A_e)
 &\le Cn^4\pi_{fold}(s+1,s+1),
 &&\text{if }s\ge s_0.
 \label{eq:top-mass}
\end{align}

Thus by Eq.~(\ref{eq:congestion})(\ref{eq:path-length})(\ref{eq:pair-mass-top}) and Eq.~(\ref{eq:top-rate-bound}), there exist a constant $C$ such that,
if \(s<s_0\), the congestion at \(e\) is at most
\begin{align}
 \frac{4n\,\pi_{fold}(A_e)}
 {\pi_{fold}(s,s)
  Q^{fold}_{(s,s)\to(s+1,s+1)}}
 \le Cn^5. \label{top-congestion-1}
\end{align}
If \(s\ge s_0\), it is at most
\begin{align}
 \frac{4n\,\pi_{fold}(A_e)}
 {\pi_{fold}(s+1,s+1)
  Q^{fold}_{(s+1,s+1)\to(s,s)}}
 \le Cn^5.
 \label{eq:top-congestion-2}
\end{align}

\paragraph{Proof for Lemma~\ref{lem:gap_fold}}
Every edge used by the chosen paths is either of the form
\eqref{eq:vertical-edge} or of the form~\eqref{eq:top-edge}.
Equation~\eqref{eq:vertical-congestion} and Equations~\eqref{top-congestion-1}\eqref{eq:top-congestion-2}
therefore imply that there exists a constant $C$ such that
\begin{align}
 \phi\le Cn^5.
\end{align}
Thus Eq.~\eqref{eq:canonical-gap} gives
\begin{align}
 \operatorname{gap}(Q^{fold})
 \ge\phi^{-1}
 \ge cn^{-5}.
\end{align}
for some constant $c$. This 
 proves Lemma~\ref{lem:gap_fold}.

\color{black}

\bibliographystyle{ieeetr}
\bibliography{ref.bib}
\end{document}